\documentclass[11pt,a4paper]{article}
\usepackage[T1]{fontenc}
\usepackage{lmodern}
\usepackage[margin=2.6cm]{geometry}
\usepackage{microtype}
\usepackage{needspace}
\usepackage{amsmath,amssymb,amsthm,booktabs,array,longtable}
\usepackage[colorlinks=true,linkcolor=blue,citecolor=blue,urlcolor=blue]{hyperref}
\hypersetup{pdftitle={Quantum gates from the middle convolution of twisted Burau representations}}

\newtheorem{theorem}{Theorem}[section]
\newtheorem{proposition}[theorem]{Proposition}
\newtheorem{lemma}[theorem]{Lemma}
\newtheorem{corollary}[theorem]{Corollary}
\newtheorem{problem}[theorem]{Problem}
\theoremstyle{definition}
\newtheorem{definition}[theorem]{Definition}
\newtheorem{remark}[theorem]{Remark}
\newtheorem{example}[theorem]{Example}

\newcommand{\C}{\mathbb C}
\newcommand{\Z}{\mathbb Z}
\newcommand{\R}{\mathbb R}
\newcommand{\Q}{\mathbb Q}
\newcommand{\e}{\mathrm e}
\newcommand{\Ht}{\widetilde H}
\newcommand{\Hq}{\overline H}
\newcommand{\Ker}{\operatorname{Ker}}
\newcommand{\Img}{\operatorname{Im}}
\newcommand{\Hom}{\operatorname{Hom}}
\newcommand{\sig}{\operatorname{sig}}
\newcommand{\spec}{\operatorname{spec}}
\newcommand{\KLM}{\mathcal{KLM}}
\newcommand{\LM}{\mathcal{LM}}
\newcommand{\ee}[1]{\e(#1)}
\newcommand{\lm}{\lambda}
\newcommand{\Wdef}{\mathcal W_{\mathrm{def}}}
\newcommand{\PB}{P\!B}

\title{Quantum gates from the middle convolution\\of twisted Burau representations}
\author{Haru Negami\\[2pt]
\small Chiba University, 1-33 Yayoi-cho, Inage-ku, Chiba 263-8522, Japan\\
\small \texttt{negamiharu@gmail.com}}
\date{September 25, 2026}

\begin{document}
\maketitle

\begin{abstract}
We study braid representations obtained from character-twisted Burau
inputs by the \mbox{Katz--Long--Moody} (KLM) construction, which corresponds
to multiplicative middle convolution for KZ-type equations. For $n\ge2$, write $q=e^{2\pi ia}$ and
$\tau=e^{2\pi ic}$ with $a,c>0$ and $nc+(n+1)a<1$. The known
invariant-form and signature results then give a positive-definite
invariant form for $\lambda=e^{2\pi il}$ with $0<l<1-nc-(n+1)a$,
yielding a unitary realization that, for each fixed input, is unique
up to simultaneous unitary conjugation. We address two questions
about this family.

First, we determine the spectra and invariant summands of the
$M=n(n+1)$-dimensional output and their consequences for quantum gates.
The generator spectrum recovers $(q,\tau)$. The representation
decomposes into two trivial summands, two reduced Burau summands, and an
$n(n-1)$-dimensional summand $\mathcal R_n$. The latter is isomorphic
to Long's reduced construction and is generically irreducible for
$n=3,4$. If $a$ or $c$ is irrational, every generator is non-Clifford
in every finite-dimensional Weyl--Heisenberg frame; if $1,a,c$ are
linearly independent over $\mathbb{Q}$, every generator is entangling
for every nontrivial bipartition. The braid image on the full output
is not projectively dense. For general positive-definite unitary
inputs, irreducibility of the full output as a braid representation
at a nonexceptional parameter forces the free-group input operators
to commute.

Second, we construct one orthonormal basis in which every generator has
a single $2(n+1)$-dimensional block, $n-2$ two-dimensional blocks, and
identity elsewhere. Each generator is a product of $2n$ rank-one phase
reflections. Using this block structure, we synthesize each generator,
in a fixed binary encoding, as a word of length $O(M)$ in one-qubit
unitaries and CNOT gates, without auxiliary qubits and including
direction preparation and its inverse. We also give finite-alphabet
approximation, conditioning and error bounds, and a matched-resource
comparison with a Quantum Shannon decomposition reference compiler.
This gate-count bound is linear in $M$, hence exponential in the
number of encoded qubits, and excludes classical preprocessing.
\end{abstract}



\section{Introduction}\label{sec:intro}

In topological quantum computation, exchanging non-Abelian anyons induces
unitary operators on an associated state space. These operators form a
representation of the braid group with noncommutative image
\cite{FLW,NayakReview,DRWReview}. These unitary operations are called
quantum gates. A quantum circuit specifies a sequence of such gates and
the subsystems on which they act. Our construction is algebraic;
an anyon realization of the resulting representations is not assumed.

\newpage
The Katz--Long--Moody (KLM) construction is an algebraic procedure
that produces a new representation of $F_n\rtimes B_n$ from a given
one, where $F_n$ is the free group of rank $n$ and $B_n$ is the braid
group on $n$ strands acting by Artin automorphisms
\cite{HN,FiniteCompanion}. It extends the Long--Moody construction
by equipping its output with a compatible action of $F_n$, so it can
be iterated for fixed $n$.

The construction and its application to
character-twisted Burau representations are established in
\cite{Long,BT,HN}. The invariant-form and signature results of
\cite{NKLM,Nsig} provide a criterion for positive-definite unitary
realizations. We first apply this criterion to make the parameter
domain explicit for the Burau family. With this preparation, the
present paper addresses two questions:
\begin{samepage}\begin{enumerate}
\renewcommand{\labelenumi}{(\theenumi)}
\setlength{\itemsep}{1pt}
\item What are the spectra and invariant summands, and which
quantum-gate properties do they imply?
\item How can all generators be expressed as short words in elementary
unitaries in one common coordinate system?
\end{enumerate}\end{samepage}
Items~(\ref{main:spectral}) and~(\ref{main:synthesis}) of
Section~\ref{subsec:main-results} state the corresponding results.
The spectral analysis gives non-Clifford and entangling properties;
the common coordinates give elementary factorizations with
approximation bounds.

\paragraph{Reading guide.}
Sections~\ref{subsec:intro-klm}--\ref{subsec:intro-forms} introduce
the objects and recall the prior results. Section~\ref{subsec:main-results}
separates the positive parameter domain obtained from those results
from the two main parts of the present analysis.
Tables~\ref{tab:results-burau} and~\ref{tab:results-other} collect
the conditions and conclusions, distinguishing applications of prior
results, proved statements, examples, and numerical observations.

Sections~\ref{sec:prelim}--\ref{sec:gauge} supply the general tools;
Proposition~\ref{prop:irreducible-input} gives the additional restriction
on seeds with irreducible output. Section~\ref{subsec:burau-signature}
applies the signature formula to the Burau family.
Section~\ref{sec:spectrum} then gives its spectra and invariant
summands, followed by gate consequences in Section~\ref{sec:gates}
and common-basis synthesis in Section~\ref{sec:synthesis}.

Appendix~\ref{app:quantum-background} explains the quantum-computing
terminology, with mathematical correspondences in
Tables~\ref{tab:quantum-dictionary} and~\ref{tab:circuit-dictionary}.
The principal notation is summarized in Table~\ref{tab:notation}
on page~\pageref{tab:notation}. Appendix~\ref{app:repro} lists the
reproducibility files. Appendix~\ref{app:additional-definite-regions}
continues the signature calculation beyond the domain used for the
main structure and synthesis results.

\setcounter{tocdepth}{1}
\tableofcontents

\subsection{Braid groups and the Katz--Long--Moody construction}
\label{subsec:intro-klm}\label{subsec:klm}

\begin{definition}[Braid groups]\label{def:braid-groups}
The \emph{braid group} $B_n$ is generated by
$\sigma_1,\ldots,\sigma_{n-1}$ with relations
\[
 \sigma_i\sigma_{i+1}\sigma_i=\sigma_{i+1}\sigma_i\sigma_{i+1}
 \quad(1\le i<n-1),\qquad
 \sigma_i\sigma_j=\sigma_j\sigma_i\quad(|i-j|\ge2).
\]
Geometrically, $\sigma_i$ exchanges the $i$th and $(i+1)$st strands;
group multiplication concatenates braids. The \emph{pure braid group}
is $\PB_n:=\Ker(B_n\to\mathfrak S_n)$, where
$\sigma_i\mapsto(i\ i+1)$ and $\mathfrak S_n$ is the symmetric group.
\end{definition}

\begin{definition}[Representations and irreducibility]\label{def:representation}
A finite-dimensional complex \emph{representation} of a group $G$
is a homomorphism $\rho:G\to\operatorname{GL}(V)$, where $V$ is
a finite-dimensional complex vector space. A subspace $W\subseteq V$
is \emph{invariant} if $\rho(\gamma)W\subseteq W$ for every
$\gamma\in G$. The representation is \emph{irreducible} if
$V\ne0$ and its only invariant subspaces are $0$ and $V$.
Its image is the subgroup $\rho(G)$.
\end{definition}
A representation of $B_n$ is thus specified by invertible matrices
satisfying the displayed braid relations.

The KLM construction takes a representation of a larger group as its
input. Write $F_n=\langle x_1,\ldots,x_n\rangle$ for the free group
and $\Gamma=F_n\rtimes B_n$ for the semidirect product defined by
\[
 \sigma_i x_i\sigma_i^{-1}=x_{i+1},\qquad
 \sigma_i x_{i+1}\sigma_i^{-1}=x_{i+1}^{-1}x_ix_{i+1},\qquad
 \sigma_i x_j\sigma_i^{-1}=x_j\quad(j\ne i,i+1).
\]
These are the Artin-action conventions of \cite[Section~2.1]{NKLM}.

\begin{definition}[Input representation]\label{def:input-representation}
An \emph{input} or \emph{seed} is a representation
$\rho:F_n\rtimes B_n\to\operatorname{GL}(V)$, with $\dim V=N$.
Following \cite[Definition~8]{NKLM}, write
\[
 g_j:=\rho(x_j),\qquad s_i:=\rho(\sigma_i).
\]
Thus the $s_i$ satisfy the braid relations and
\[
 s_i g_i s_i^{-1}=g_{i+1},\qquad
 s_i g_{i+1}s_i^{-1}=g_{i+1}^{-1}g_i g_{i+1},\qquad
 s_i g_j s_i^{-1}=g_j\quad(j\ne i,i+1).
\]
\end{definition}
In particular, arbitrary quantum gates do not by themselves constitute
an input: these compatibility relations must also hold.

\begin{definition}[Twisted Long--Moody operators]
\label{def:lm-operators}
Following \cite[Definitions~8 and~10]{NKLM}, let
$\lambda\in\C\setminus\{0,1\}$.
In the direct-sum basis of $V^{\oplus n}$,
define $G_i:=\rho^{\mathrm{LM}}_\lambda(x_i)$ to be the identity
except in block row $i$, where
\begin{equation}\label{eq:free-matrices}
 (G_i)_{i,\bullet}=
 \bigl(\lambda(g_1-I_N),\ldots,\lambda(g_{i-1}-I_N),\
       \lambda g_i,\ g_{i+1}-I_N,\ldots,g_n-I_N\bigr).
\end{equation}
The braid operators are
\begin{equation}\label{eq:braid-matrices}
\begin{aligned}
 S_i&:=\rho^{\mathrm{LM}}_\lambda(\sigma_i)
 =s_i^{\oplus n}
 \begin{pmatrix}
 I_{N(i-1)}&0&0\\
 0&R_i&0\\
 0&0&I_{N(n-i-1)}
 \end{pmatrix},\\[2pt]
 R_i&:=\begin{pmatrix}0&g_i\\I_N&I_N-g_{i+1}\end{pmatrix}.
\end{aligned}
\end{equation}
The factors $I_0$ at the ends are omitted. Equivalently, for
$v=(v_1,\ldots,v_n)\in V^{\oplus n}$, with each $v_j\in V$, and
$1\le i\le n-1$, all components are given by
\[
\begin{aligned}
S_i(v_1,\ldots,v_n)=\bigl(&s_i v_1,\ldots,s_i v_{i-1},\ s_i g_i v_{i+1},\\
 &s_i\bigl(v_i+(I_N-g_{i+1})v_{i+1}\bigr),\
 s_i v_{i+2},\ldots,s_i v_n\bigr).
\end{aligned}
\]
Empty initial or final strings of components are omitted.
These matrices define the representation
$\rho^{\mathrm{LM}}_\lambda:F_n\rtimes B_n\to
\operatorname{GL}(V^{\oplus n})$ of
\cite[Definitions~8 and~10]{NKLM}; we denote this module by $\LM_\lambda(V)$.
\end{definition}
We call the $j$th summand of $V^{\oplus n}$ the $j$th \emph{slot};
$v_j\in V$ is the component of $v=(v_1,\ldots,v_n)$ in that slot.
Thus $S_i$ mixes only slots $i$ and $i+1$, and acts by $s_i$ on
every other slot.

\begin{definition}[Katz--Long--Moody quotient]\label{def:klm-quotient}
Following \cite[Lemma~3 and Definition~11]{NKLM}, define subspaces of
$V^{\oplus n}$ by
\[
\begin{aligned}
K&:=\{(v_1,\ldots,v_n)\in V^{\oplus n}:\
                 v_j\in\Ker(g_j-I_N)\text{ for every }j\},\\
L_\lambda&:=\bigcap_{j=1}^n\Ker(G_j-I_{nN}),
\qquad G_j=\rho^{\mathrm{LM}}_\lambda(x_j).
\end{aligned}
\]
Thus $K$ is the direct sum of the local fixed spaces, with
$\Ker(g_j-I_N)$ in slot $j$, and $L_\lambda$ is the common fixed
space of the output free-group operators.
Put $P_s:=g_1\cdots g_s$. The same subspace $L_\lambda$ has the
component description
\[
L_\lambda=\left\{(v_1,\ldots,v_n)\in V^{\oplus n}:\
\begin{array}{l}
v_j=g_{j+1}v_{j+1}\quad(1\le j<n),\\
v_n=\lambda P_n v_n
\end{array}\right\}.
\]
These are invariant subspaces and $K\cap L_\lambda=0$ for
$\lambda\ne1$ \cite[Lemma~3, Definition~11, and proof of Theorem~16(iii)]{NKLM}.
The \emph{Katz--Long--Moody representation} is the induced action on
\begin{equation}\label{eq:quotient}
 Q_\lambda:=\KLM_\lambda(V)
 =\LM_\lambda(V)/(K+L_\lambda),\qquad
 \rho^{\mathrm{KLM}}_\lambda(\gamma)[v]
 =[\rho^{\mathrm{LM}}_\lambda(\gamma)v].
\end{equation}
Equivalently,
\begin{equation}\label{eq:L}
 \iota_\lambda:\Ker(\lambda P_n-I_N)\longrightarrow V^{\oplus n},\qquad
 v\longmapsto(g_2\cdots g_nv,\ g_3\cdots g_nv,\ldots,g_nv,v)
\end{equation}
has image $L_\lambda$. A parameter is \emph{exceptional} when
$\Ker(\lambda P_n-I_N)\ne0$, and \emph{nonexceptional} otherwise.
Following \cite[Section~6]{FiniteCompanion}, set
\[
 U_0:=V^{\oplus n}/K,\qquad
 Q_\lambda\cong U_0\big/\bigl((K+L_\lambda)/K\bigr),\qquad
 R:=\dim U_0=nN-\dim K.
\]
The subscript $0$ labels the fixed intermediate quotient; it is not
specialization at $\lambda=0$. At nonexceptional parameters,
$Q_\lambda$ is canonically identified with $U_0$.
The braid representation on the whole space $U_0$ is called the
\emph{full output}, in contrast to its individual invariant summands
and Long's reduced construction (Proposition~\ref{prop:long-reduced}).
\end{definition}
Thus KLM constructs compatible matrix actions on a new space from the
input matrices. To obtain a unitary realization, we must additionally
establish positive definiteness of an invariant form on that space.

\subsection{Character-twisted Burau inputs}
\label{subsec:intro-burau}
The Burau representation supplies the concrete input family studied here.
We label the strands of $B_{n+1}$ by $0,\ldots,n$ and its generators
by $\sigma_0,\ldots,\sigma_{n-1}$; the subgroup $B_n$ braids strands
$1,\ldots,n$.

\begin{definition}[Burau and reduced Burau representations]\label{def:burau}
For $q\in\C^\times$, the \emph{unreduced Burau representation}
$\rho_q:B_{n+1}\to\operatorname{GL}_{n+1}(\C)$ is given in the basis
$e_0,\ldots,e_n$ by
\[
 \rho_q(\sigma_j)
 =I_j\oplus\begin{pmatrix}1-q&q\\1&0\end{pmatrix}\oplus I_{n-j-1}
 \qquad(0\le j<n).
\]
The linear functional
$\ell_q(v_0,\ldots,v_n):=\sum_{j=0}^n q^jv_j$ is invariant.
The restriction of $\rho_q$ to the $n$-dimensional space
$V_{\mathrm{red}}:=\Ker\ell_q\subset\C^{n+1}$ is the
\emph{reduced Burau representation}, denoted by
$\mathrm{Bur}^{\mathrm{red}}_{n+1}(q)$.
The vector $\mathbf1_{n+1}=(1,\ldots,1)^T$ is fixed; whenever
$1+q+\cdots+q^n\ne0$, it gives
$\C^{n+1}=\C\mathbf1_{n+1}\oplus V_{\mathrm{red}}$.
The same formulas with $n+1$ replaced by $n$ (and the generators
relabelled as $\sigma_1,\ldots,\sigma_{n-1}$) define $\rho_q$ on $B_n$
and $\mathrm{Bur}^{\mathrm{red}}_n(q)$, of dimension $n-1$;
the relation between $V_{\mathrm{red}}$ and
$\mathrm{Bur}^{\mathrm{red}}_n(q)$ on $\mathcal T_n$ is stated in
Proposition~\ref{prop:decomp}.
\end{definition}
For the seed, the matrices above are written in the basis
$e_0,\ldots,e_n$. They need not be unitary with respect to the
standard inner product. We refer to this basis, before unitarization,
as the \emph{raw Burau basis}, and to its coordinates as raw Burau coordinates.
For $q=e^{2\pi ia}$ with $0<a<1/(n+1)$ the representation preserves
a positive definite Hermitian form, so one common change of basis
makes its matrices unitary \cite{Squier}.

\begin{definition}[Character-twisted Burau seed]\label{def:burau-seed}
Identify $F_n\rtimes B_n$ with the subgroup $B_{1,n}\subset B_{n+1}$
whose induced permutation fixes strand $0$, using
$x_1=\sigma_0^2$ and $x_{j+1}=\sigma_jx_j\sigma_j^{-1}$.
Let $\tau\in\C^\times$ be a scalar parameter. The input matrices are
\[
 g_j^{\mathrm{tB}}:=\tau\rho_q(x_j),\qquad s_i^{\mathrm{tB}}:=\rho_q(\sigma_i).
\]
The scalar $\tau$ defines a one-dimensional character
$\chi_\tau:F_n\rtimes B_n\to\C^\times$ by
$\chi_\tau(x_j)=\tau$ and $\chi_\tau(\sigma_i)=1$.
The input representation is the tensor product of this character
with the restriction of $\rho_q$.
The Artin relations ensure that it is a representation of
$F_n\rtimes B_n$ \cite{Long,BT}.
The superscript $\mathrm{tB}$ distinguishes this family from a general
input: in the general construction we substitute
$g_j=g_j^{\mathrm{tB}}$ and $s_i=s_i^{\mathrm{tB}}$.
\end{definition}
Here $q$ is the Burau parameter (often written $t$), $\tau$ is the
scalar twisting parameter, and $\lambda$ is the KLM
convolution parameter. In particular, $\tau$ is not a parameter for
projectivizing the output. For $|\tau|=1$ and the above range of $a$,
the entire seed preserves a common positive definite form.

Braid representations used in quantum computation also include Jones
representations at suitable roots of unity, which support universal
quantum computation \cite{FLW}, and Ising-anyon representations,
whose braiding gates lie in the Clifford group and are not universal
by braiding alone \cite{AGW}. Once a common positive-definite invariant
inner product and an orthonormal basis are chosen, the braid generators
act as unitary quantum gates \cite{BarencoGates}. Their computational
capabilities depend on the representation and the chosen encoding. Our purpose is to apply KLM to an explicit family of
compatible input operators and analyze the resulting gates.

\subsection{Invariant Hermitian forms and the signature formula}
\label{subsec:intro-forms}
There are two conventions for the word ``unitary'' in the relevant
literature. Long \cite[Definition~2.2]{Long} requires preservation of
a nondegenerate Hermitian form, without assuming positive definiteness.
N--KLM \cite[Definition~7]{NKLM} follows this convention.
We first recall the preservation result in that sense, and then state
the positive-definite convention used for quantum gates in this paper.

\begin{definition}[Hermitian forms and signatures]\label{def:hermitian}
In a chosen basis, a \emph{Hermitian form} is
$h(v,w)=v^\dagger H_0w$, with $H_0^\dagger=H_0$.
Its \emph{radical} is $\Ker H_0$; it is \emph{nondegenerate} if
$\Ker H_0=0$, and \emph{positive definite} if
$v^\dagger H_0v>0$ for all $v\ne0$.
It is \emph{negative definite} if $-H_0$ is positive definite,
and \emph{definite} if it has either sign.
It is \emph{positive semidefinite} if $v^\dagger H_0v\ge0$ for all $v$.
Its signature $\sig(H_0)=(n_+(H_0),n_-(H_0))$ counts positive and
negative eigenvalues with multiplicity.
A representation preserves $h$ if
$\rho(\gamma)^\dagger H_0\rho(\gamma)=H_0$ for every $\gamma$.
\end{definition}

\begin{theorem}[{Preservation of an invariant form,
\cite[Theorems~8--9]{NKLM}}]\label{thm:klm-form}
Suppose the input $\rho$ preserves a nondegenerate Hermitian form $H_0$.
Let $|\lambda|=1$, $\lambda\ne1$, and choose a square root of $\lambda$.
The block matrix $\Ht(\lambda)$ on $V^{\oplus n}$ with blocks
\begin{equation}\label{eq:blocks}
 \Ht_{jk}(\lambda)=
 \begin{cases}
 \lambda^{-1/2}H_0(g_j^{-1}-I_N)(g_k-I_N),&j<k,\\
 \lambda^{-1/2}H_0(g_j^{-1}-\lambda I_N)(g_j-I_N),&j=k,\\
 \lambda^{1/2}H_0(g_j^{-1}-I_N)(g_k-I_N),&j>k
 \end{cases}
\end{equation}
is Hermitian and invariant under the whole representation
$\rho^{\mathrm{LM}}_\lambda$. Its radical is exactly $K+L_\lambda$.
Consequently it induces a nondegenerate invariant form
$\Hq(\lambda)$ on $Q_\lambda$.
\end{theorem}
The ambient form is nondegenerate precisely when $K=L_\lambda=0$;
the quotient form is nondegenerate even at an exceptional parameter.
Neither assertion implies positive definiteness of the output,
even when the input form is positive definite.

\begin{definition}[Eigenangles and rank data]\label{def:eigenangles}
Write $\ee{x}:=\exp(2\pi ix)$. For a unit-modulus eigenvalue
$\ee{\theta}$, its \emph{eigenangle} is the representative
$\theta\in[0,1)$, counted with algebraic multiplicity.
Assume now that $H_0>0$ and choose an $H_0$-orthonormal basis, so
$H_0=I_N$. Write $\lambda=\ee l$ with $0<l<1$ and fix
$\lambda^{1/2}=\exp(\pi il)$.
Let $a_j^{(k)}$ and $\beta_k(l)$ be the eigenangles of $g_j$ and
$\lambda P_n$, respectively, and put
\[
 \begin{gathered}
 \kappa:=\dim K,\quad R=nN-\kappa,\quad
 r(l):=\dim\Ker(\lambda P_n-I_N),\\
 m(l):=Nl+\sum_{j,k}a_j^{(k)}-\sum_k\beta_k(l)\in\Z.
 \end{gathered}
\]
\end{definition}

\begin{theorem}[{Signature formula,
\cite[Theorem~1.1 and Corollary~1.2]{Nsig}}]\label{thm:import-sig}
For positive-definite input and every $l\in(0,1)$,
\[
 \sig\Hq(\ee l)=\bigl(R-m(l),\ m(l)-r(l)\bigr).
\]
Thus the canonical quotient form is positive definite if and only if
$m(l)=r(l)$, negative definite if and only if $m(l)=R$, and definite exactly when
$m(l)\in\{r(l),R\}$. At nonexceptional parameters, $r(l)=0$.
\end{theorem}
This determines when the canonical form yields a positive inner
product after a possible change of sign. For reducible output it is
a sufficient condition for the existence of a positive definite
invariant inner product, but need not be
necessary: different invariant summands may require different signs.

\paragraph{Positive input and finite images.}
Positive definiteness does not imply that the input or output image
is finite. A componentwise criterion for the existence of a
positive-definite invariant inner product is given in
\cite[Theorem~3.3]{FiniteCompanion}. Under the additional assumptions
that the full input image is finite and $\lambda\ne1$ is a root of
unity, finiteness of the full KLM image is equivalent to definiteness
of the canonical forms on all nonzero convolutions of irreducible
free-group input constituents, at every embedding of a coefficient
number field containing $\lambda$
\cite[Theorem~2.10 and Corollary~3.10]{FiniteCompanion}.
The signs may differ between constituents and embeddings.
This arithmetic condition concerns the joint image of the generators;
finite order of each generator alone does not suffice.
Our gate family is studied without a finite-image assumption.

\begin{definition}[Unitary and unitarizable representations]\label{def:unitary}
From now on, a representation on a space equipped with a
\emph{positive definite} form $H$ is called \emph{unitary} if
$\rho(\gamma)^\dagger H\rho(\gamma)=H$ for every $\gamma$.
It is \emph{unitarizable} if such an $H$ exists, equivalently if one
change of basis makes all its matrices unitary for the standard inner
product. We write
$U(N)=\{U\in\operatorname{GL}_N(\C):U^\dagger U=I_N\}$.
\end{definition}
If $A_i^\dagger H A_i=H$ and $H=\mathsf{C}^\dagger\mathsf{C}>0$,
then $U_i=\mathsf{C} A_i\mathsf{C}^{-1}$ is a unitary realization
in one common orthonormal basis. For fixed input the operators $S_i$
in \eqref{eq:braid-matrices} are independent of $\lambda$; hence so
are their induced operators $A_i$ on $U_0$.
At nonexceptional parameters, $\lambda$ varies the invariant form on this fixed
braid representation, whereas $q$ and $\tau$ vary the input.
This independence does not apply to the free-group action or to the
additional quotient at exceptional parameters.

\begin{remark}[KZ-type equations]
On the pure braid group $\PB_{n+1}\cong F_n\rtimes\PB_n$, the KLM
construction corresponds to Haraoka's multiplicative middle convolution
for Knizhnik--Zamolodchikov (KZ)-type equations
\cite[Theorem~6 and Corollary~7]{NKLM}; see also \cite{HN,Katz,DRMiddle}.
Thus the algebraic transformation also has an interpretation in terms
of monodromy and integral transformations of differential equations.
The gate results below use its matrix formulation.
\end{remark}

The construction of a unitary realization follows the sequence
\[
 \text{compatible seed}\ \longmapsto\ (A_i)
 \ \longmapsto\ H=\mathsf C^\dagger\mathsf C>0
 \ \longmapsto\ U_i=\mathsf C A_i\mathsf C^{-1},
\]
where $H$ is invariant under every $A_i$ and the same factor
$\mathsf C$ is used for all generators.

\subsection{Parameter domain and main results}\label{subsec:main-results}
Appendix~\ref{app:quantum-background} defines the quantum terminology:
projective order (Definition~\ref{def:projective-entangling}), entangling
operators and frames (Definitions~\ref{def:entangling-gates}
and~\ref{def:wh-frame}), and encodings and phase reflections
(Definitions~\ref{def:encodings}--\ref{def:phase-reflections}).

\paragraph{Parameter domain from the signature formula.}
Fix $n\ge2$. For the $(n+1)$-dimensional character-twisted Burau
seed of Definition~\ref{def:burau-seed}, take $q=\ee a$, $\tau=\ee c$,
and put
\[
 M=n(n+1),\qquad
 \mathcal T_n=\{(a,c):a>0,\ c>0,\ nc+(n+1)a<1\}.
\]
On $\mathcal T_n$, the seed eigenvalues give $K=0$ and $U_0=\C^M$.
Applying \cite[Cor.~6.10]{Nsig} to the boundary spectrum computed
in Lemma~\ref{lem:fulltwist} gives a positive canonical form for
\[
 0<l<1-nc-(n+1)a.
\]
Theorem~\ref{thm:twisted-burau} records this explicit specialization;
Proposition~\ref{prop:twisted-unitarization} applies the known
unitarization criterion. For fixed $(q,\tau)$, the braid matrices are
independent of $l$ and their unitary realizations are simultaneously
unitarily equivalent. These facts specify the family studied below.

\paragraph{Main results for this family.}
On $\mathcal T_n$, we establish the following structural and synthesis
statements.
\begin{enumerate}
\renewcommand{\labelenumi}{(\theenumi)}
\item \textbf{Invariant summands and spectral consequences.}\label{main:spectral}
The representation decomposes into two trivial summands, the two
$(n-1)$-dimensional summands $\mathrm{Bur}^{\mathrm{red}}_n(\tau)$ and
$\mathrm{Bur}^{\mathrm{red}}_n(q)$, and $\mathcal R_n(q,\tau)$ of dimension $n(n-1)$,
which is generically irreducible for $n=3,4$
(Propositions~\ref{prop:decomp} and~\ref{prop:generic-irred}).

The generator spectrum recovers the ordered pair $(q,\tau)$;
a generator has finite projective order exactly when $a,c\in\Q$
(Proposition~\ref{prop:charpoly}, Corollary~\ref{cor:injective}, and
Theorem~\ref{thm:nonclifford}).

If at least one of $a,c$ is irrational, every generator is non-Clifford
in every Weyl--Heisenberg frame (Theorem~\ref{thm:nonclifford}).
If $1,a,c$ are $\Q$-linearly independent, every generator is entangling
for every nontrivial unitary bipartite identification
(Theorem~\ref{thm:entangling}).

\item \textbf{Common-basis synthesis and approximation bounds.}\label{main:synthesis}
For $l$ in the positive interval above, one common orthonormal basis
gives small coordinate blocks (Theorem~\ref{thm:syn-blocks}).
Applying the standard reflection-length formula to the spectrum gives
minimal rank-one phase-reflection length $2n$, also modulo scalar phase
(Corollary~\ref{cor:syn-length}).

In the fixed binary encoding, we construct $O(M)$-letter words in
one-factor $U(2)$ matrices and all-to-all CNOT, without auxiliary
qubits and including direction preparation and its inverse
(Definition~\ref{def:elementary-generators} and Theorem~\ref{thm:syn-cost}).

For a fixed finite inverse-closed dense-generating local $SU(2)$
alphabet with exact CNOT, we give projective operator-norm
approximation bounds for generators and braid words, together with
conditioning estimates for constructing the common basis
(Corollary~\ref{cor:syn-finite-precision} and
Sections~\ref{subsec:syn-error-margins} and~\ref{app:syn-counts}).
\end{enumerate}

\paragraph{A further result for general inputs.}
Beyond the recalled augmentation criteria, we prove that irreducibility
of the nonzero full braid output at nonexceptional parameters forces
the free-group operators of a positive-definite unitary input to commute
(Proposition~\ref{prop:irreducible-input}). This restricts which seeds
can have irreducible full output; it does not assume a Burau seed.

Tables~\ref{tab:results-burau}--\ref{tab:results-other} summarize the
parameter regimes and the evidence for each conclusion. In Table~\ref{tab:results-burau},
the canonical-form column applies the signature formula of \cite{Nsig};
the structure and gate column records the subsequent analysis.
The rows can
overlap; they are not a disjoint classification of seeds. A dash means
that no conclusion is recorded here. Canonical definiteness and the
existence of an arbitrary positive invariant form are different
questions, as Example~\ref{ex:empty-window} shows.

In Table~\ref{tab:results-burau}, the seed is that of
Definition~\ref{def:burau-seed}, with $n\ge2$, $q=\ee a$,
$\tau=\ee c$, and $0<a<1/(n+1)$. Put $M=n(n+1)$ and
$\kappa=\dim K$. All windows in $l$ concern the canonical form at
nonexceptional $\lambda=\ee l$, $0<l<1$; there $Q_\lambda\cong U_0$.
Definitions~\ref{def:resonance} and~\ref{def:definite-window}
specify the exceptional parameters and definite windows.

\begingroup
\small
\setlength{\tabcolsep}{4pt}
\renewcommand{\arraystretch}{1.08}
\setlength{\LTcapwidth}{\linewidth}
\begin{longtable}{@{}>{\raggedright\arraybackslash}p{0.25\linewidth}
>{\raggedright\arraybackslash}p{0.29\linewidth}
>{\raggedright\arraybackslash}p{0.25\linewidth}
>{\raggedright\arraybackslash}p{0.15\linewidth}@{}}
\caption{Burau seed: parameter regimes, canonical forms, and consequences.}
\label{tab:results-burau}\\
\toprule
Parameters & Quotient and canonical form & Structure and gate consequences & Evidence \\
\midrule
\endfirsthead
\multicolumn{4}{@{}l}{\tablename~\thetable\ (continued)}\\
\toprule
Parameters & Quotient and canonical form & Structure and gate consequences & Evidence \\
\midrule
\endhead
\midrule
\multicolumn{4}{r@{}}{Continued on the next page}\\
\endfoot
\bottomrule
\endlastfoot
$\boldsymbol{(a,c)\in\mathcal T_n}$.
 & $K=0$, $\dim U_0=M$. Positive for
 $0<l<1-nc-(n+1)a$; never negative.
 & Decomposition: item~(\ref{main:spectral}). For $n=3,4$,
 $\mathcal R_n$ is generically irreducible. Finite projective generator
 order iff $a,c\in\Q$; otherwise non-Clifford in every frame.
 If $1,a,c$ are $\Q$-independent: entangling in every nontrivial
 unitary bipartition. Positive window: $O(M)$ words in a fixed encoding.
 & Form: application\par \ref{thm:twisted-burau}.\par
 Structure and\par gates:\par \ref{prop:decomp}, \ref{prop:generic-irred},
 \ref{thm:nonclifford}, \ref{thm:entangling}.\par
 Synthesis:\par \ref{thm:syn-cost}.
 \\\addlinespace[4pt]
\textbf{Beyond the positive triangle.}\par
$c>0$, $c+2a<1$,\par $nc+(n+1)a\ge1$.
 & $K=0$, $\dim U_0=M$. No positive canonical window.
 & May overlap Negative A below; absence of a positive canonical form
 does not decide unitarizability.
 & Signature\par application:\par
 \ref{prop:twisted-unitarization}, \ref{prop:additional-definite-regions}.
 \\\addlinespace[4pt]
\textbf{Additional positive window.}\par
$n=2$, $\tfrac14<a<\tfrac13$,\par $1-2a<c<\tfrac12$.
 & $K=0$, $\dim U_0=6$. Positive for $0<l<1-2c$.
 & Unitary realization. The main structure, gate, and synthesis
 results are not extended here.
 & Signature\par application:\par \ref{prop:additional-definite-regions}.
 \\\addlinespace[4pt]
\textbf{Negative A.}\par
$1-\tfrac1n<c<1-2a$.
 & $K=0$, $\dim U_0=M$. Negative for $n(1-c)<l<1$.
 & $-\Ht$ gives a unitary realization. Further conclusions: ---.
 & Signature\par application:\par \ref{prop:additional-definite-regions}.
 \\\addlinespace[4pt]
\textbf{Negative B.}\par
$1-\tfrac an<c<1$.
 & $K=0$, $\dim U_0=M$. Negative for $n+1-nc-a<l<1$.
 & $-\Ht$ gives a unitary realization. Further conclusions: ---.
 & Signature\par application:\par \ref{prop:additional-definite-regions}.
 \\\addlinespace[4pt]
\textbf{Untwisted specialization.}\par $c=0$, $\tau=1$.
 & $\kappa=n^2$, $\dim U_0=n$. Canonical window: ---.
 & $U_0\cong\mathbf1\oplus\mathrm{Bur}^{\mathrm{red}}_n(q)$;
 the augmentation identifies $U_0$ with $V_{\mathrm{red}}$.
 & Proved remark\par \ref{rem:chosen-region}(b).
 \\\addlinespace[4pt]
\textbf{Resonant line.}\par $c=1-2a$, $\tau q^2=1$.
 & $K\ne0$; the dimension changes. Canonical window: ---.
 & Excluded from the additional-window classification. Further
 conclusions: ---.
 & Scope of\par Appendix\par \ref{app:additional-definite-regions}.
 \\
\end{longtable}
\endgroup

Finite projective order of the generators does not imply finiteness
of the joint braid image. The additional windows classify
the canonical form only, without extending the main structure or
synthesis theorems beyond their stated domains.

\begingroup
\small
\setlength{\tabcolsep}{4pt}
\renewcommand{\arraystretch}{1.08}
\setlength{\LTcapwidth}{\linewidth}
\begin{longtable}{@{}>{\raggedright\arraybackslash}p{0.25\linewidth}
>{\raggedright\arraybackslash}p{0.29\linewidth}
>{\raggedright\arraybackslash}p{0.25\linewidth}
>{\raggedright\arraybackslash}p{0.15\linewidth}@{}}
\caption{Other seeds, comparison examples, and general results.}
\label{tab:results-other}\\
\toprule
Seed and parameters & Quotient and invariant forms & Structure and gate consequences & Evidence \\
\midrule
\endfirsthead
\multicolumn{4}{@{}l}{\tablename~\thetable\ (continued)}\\
\toprule
Seed and parameters & Quotient and invariant forms & Structure and gate consequences & Evidence \\
\midrule
\endhead
\midrule
\multicolumn{4}{r@{}}{Continued on the next page}\\
\endfoot
\bottomrule
\endlastfoot
\textbf{Two scalar constituents.}\par
$n=3$, $N=2$; equal diagonal $g_j$ with eigenangles $1/10,9/10$;
$s_i=I_2$.
 & $K=0$, $\dim U_0=6$. No definite canonical quotient form, even at
 walls. At $\lambda=-1$, changing one summand's sign gives
 $H_{\mathrm{pos}}>0$, invariant under the full KLM action.
 & Two three-dimensional unreduced Burau summands. Unitarizable;
 finiteness and gate properties: ---.
 & Explicit example\par \ref{ex:empty-window}.
 \\\addlinespace[4pt]
\textbf{Hecke specialization.}\par
$n\ge2$, $0<a<1$;\par $s_i=sI_N$, $s\in U(1)$;\par
$g_j=g\in U(N)$,\par $\spec(g)\subseteq\{1,\ee a\}$.
 & Let $\nu>0$ be the $\ee a$-eigenspace dimension. Then
 $\kappa=n(N-\nu)$ and $\dim U_0=n\nu$. The nonexceptional definite locus is nonempty iff
 $a\in(0,1/n)$ or $a\in((n-1)/n,1)$.
 & $\nu$ copies of the $n$-dimensional unreduced Burau representation
 at $\ee a$, with each braid generator multiplied by $s$.
 & Known specialization:\par Example\par \ref{ex:hecke};\par
 Definition\par \ref{def:klm-quotient}.
 \\\addlinespace[4pt]
\textbf{Finite Clifford reference example.}\par
$n=3$, $N=1$,\par $g_j=i$, $s_i=1$;\par $\lambda\ne0,1$.
 & $K=0$, $\dim U_0=3$. At exceptional $\lambda=i$,
 $\dim Q_i=2$; otherwise $\dim Q_\lambda=3$.
 Canonical window: ---.
 & At $\lambda=i$: a two-dimensional representation $\varphi$, with
 ordinary/projective braid image orders $96/24$ and projective one-qubit
 Clifford image. Nonexceptionally: $\mathbf1\oplus\varphi$, with orders $96/96$.
 & Exact reference example:\par Section\par \ref{subsec:clifford-benchmark}.
 \\\addlinespace[4pt]
\textbf{Jones--Temperley--Lieb endpoint sectors.}\par
Restrict $B_{n+1}$ sectors of dimension $\ge2$ to $B_{1,n}$;
$q=\ee{1/\ell}$. Tested $(n,\ell)$ as in Remark~\ref{rem:numerics}(b).
 & Tests concern the ambient Long--Moody matrices $S_i$.
 Quotient dimensions, canonical windows, and unitarizability: ---.
 & Finite-precision tests suggest Jordan blocks and show large norms
 of sampled words. An exact obstruction remains open.
 & Numerical observation:\par Remark\par \ref{rem:numerics}(b).
 \\\addlinespace[4pt]
\textbf{General unitary seeds.}\par
The full $F_n\rtimes B_n$ seed preserves a common positive-definite
form. Work on fixed $U_0$ at nonexceptional parameters.
 & $R=\dim U_0=nN-\dim K$.
 For a fixed seed, canonical unitary realizations with
 $l\in\Wdef$ are simultaneously unitarily equivalent, without an
 irreducibility assumption.
 & If $U_0\ne0$ and its braid restriction is irreducible, then
 $R\le N$, the $g_j$ commute, and $\Wdef=\e^{-1}(E^c)$.
 & Recalled:\par \ref{prop:augmentation}.\par
 Seed restriction:\par \ref{prop:irreducible-input}.\par
 Standard transport:\par \ref{thm:gauge}.
 \\
\end{longtable}
\endgroup

\begin{table}[htbp]
\centering\small
\caption{Principal notation. Indexed local matrices are distinguished
from scalar dimensions and subspaces.}\label{tab:notation}
\begin{tabular}{@{}p{0.26\linewidth}p{0.68\linewidth}@{}}
\toprule
Symbol & Meaning \\
\midrule
$g_j,\ s_i$ & General input operators $\rho(x_j)$ and $\rho(\sigma_i)$. \\
$g_j^{\mathrm{tB}},\ s_i^{\mathrm{tB}}$ & Character-twisted Burau input operators of Definition~\ref{def:burau-seed}. \\
$G_j,\ S_i$ & Long--Moody free-group and braid operators on $V^{\oplus n}$. \\
$R_i$ & The $2N\times2N$ block in \eqref{eq:braid-matrices}. \\
$K,\ L_\lambda$ & Subspaces removed in the KLM quotient. \\
$U_0,\ Q_\lambda$ & $V^{\oplus n}/K$ and $V^{\oplus n}/(K+L_\lambda)$. \\
$A_i,\ U_i$ & Braid operator on $U_0$ and its unitary realization in a common basis. \\
$H_0$ & The $N\times N$ invariant Hermitian form on the seed. \\
$\widetilde H,\ \overline H$ & Canonical forms on $V^{\oplus n}$ and $Q_\lambda$. \\
$H_l,\ \mathsf C_l$ & In the synthesis section, $H_l=\widetilde H(\ee l)>0$ and its upper triangular Cholesky factor, $\mathsf C_l^\dagger\mathsf C_l=H_l$. \\
$H_0^{\mathrm{end}}$ & The $M\times M$ limit $\lim_{l\downarrow0}H_l$ in Proposition~\ref{prop:syn-endpoint-rates}. \\
Slot-major order & Coordinate order defined in Section~\ref{sec:synthesis}. \\
$N,\ R,\ M$ & Seed dimension, $R=\dim U_0$, and $M=n(n+1)$, the output dimension in the main results. \\
$E,\ E^c,\ \Wdef$ & Exceptional and nonexceptional values of $\lambda$; $\Wdef$ is a set of angles $l$. \\
$q,\ \tau,\ \lambda$ & Burau parameter, scalar twisting parameter, and convolution parameter. \\
$\mathcal T_n,\ d,\ \omega$ & The chosen input region; $d=\tfrac32a+c$ and $\omega=\ee d=q^{3/2}\tau$ with $q^{1/2}=\ee{a/2}$. \\
\bottomrule
\end{tabular}
\end{table}

\subsection{Relation to prior work and scope}
\paragraph{Prior results used.}
The construction and character rescaling originate in \cite{BT,HN,Long}.
Reduced Burau inputs and their $n(n-1)$-dimensional reduced outputs
already occur in \cite[Example~2.1(b), Theorem~2.11]{Long}, whose
Example~4.3 gives explicit twelve-dimensional $B_4$ matrices.
The proof of that theorem uses the difference coordinates underlying
Proposition~\ref{prop:long-reduced}. We record them in our Artin-action
conventions and verify that, on $\mathcal T_n$, their restriction
identifies the augmentation kernel $\mathcal R_n$ with Long's reduced output.
The invariant form is from \cite{NKLM}, and \cite{Nsig} determines its
signature and definite windows for general unitary inputs.
The $\lambda$-independence and unitarization criterion recalled here
are also in \cite[Cor.~1.2, Prop.~6.11]{Nsig}.
The componentwise unitarizability, finite-image criteria, and
exceptional-quotient analysis used here are from \cite{FiniteCompanion}.

\paragraph{Applications of prior results.}
We compute the Burau boundary spectrum and substitute it into the known
signature formula to obtain the domain $\mathcal T_n$ and its positive
window (Section~\ref{subsec:burau-signature}). Appendix~\ref{app:additional-definite-regions}
continues this calculation in the stated Burau strip. These are
applications of the existing criterion. The generator characteristic
polynomial also specializes \cite[Theorem~3.1]{NHecke}; we give a direct
coordinate proof in Proposition~\ref{prop:charpoly}.
Section~\ref{sec:gauge} records the standard polar-decomposition
arguments needed to compare unitary realizations.

\paragraph{New results.}
Using these inputs, we determine the invariant summands and their
spectra, prove generic irreducibility of $\mathcal R_n$ for $n=3,4$,
and derive the parameter-recovery, non-Clifford, and entangling
conclusions.
The synthesis results construct small coordinate blocks simultaneously
for every generator, and turn this common-basis structure into
$O(M)$ elementary words with error and conditioning bounds.
The general seed restriction of Proposition~\ref{prop:irreducible-input}
is a separate conclusion beyond the recalled augmentation facts.

The synthesis uses unitary reflection length \cite[arXiv version, Example~2.15]{HLR},
Householder and sparse-isometry methods \cite{IKV,MIC,GoubaultHouseholder},
and state preparation and controlled gates
\cite{SBM,MottonenStatePreparation,BarencoGates}.
For a general $m$-qubit state with $s$ nonzero amplitudes,
Li--Luo \cite[Theorem~1]{LiLuo} give an ancilla-free
$O(ms/\log m+m)$ bound in the same one-qubit/CNOT model.
Related Toffoli-cost refinements appear in \cite{RupprechtWoelk}.

Here the common basis gives core reflection directions supported on
two slot labels and at most $2(n+1)$ coordinates.
This special support permits $O(n)$ preparation per core direction;
Section~\ref{subsec:syn-elementary} compares it with the general-support
bound. Including spectator directions, controlled phases, and inverse
preparation gives the $O(M)$ word bound of Theorem~\ref{thm:syn-cost}.
Finite-alphabet approximation uses one-qubit Solovay--Kitaev results
\cite{DawsonNielsenSK,KuperbergSK}.
Quantum braid implementations already occur in the Jones path model
\cite{AJL}; Burau unitarity and density are also studied via
non-semisimple topological quantum field theory (TQFT) \cite{GLPS}.
Squier-form Cholesky unitarization appears
in a two-dimensional $B_3$ quantum-switch model
\cite[Section~4]{KolpakovSwitch}. The ordered-block Cholesky argument
here supplies the small supports in a \emph{single} basis for this KLM
family. This coordinate-block property is distinct from tensor-local
Yang--Baxter localization \cite{RW}, and our preprocessing model differs
from sparse-unitary oracle access \cite{JordanWocjan}.

\paragraph{Scope.}
Unitary equivalence preserves spectra and the decomposition into
invariant summands up to isomorphism. It does not, in general, preserve
word length in a fixed elementary generating set. The spectral
conclusions concerning Clifford groups and entanglement quantify over
all relevant frames and tensor identifications. The synthesis bounds
use the common coordinates and encoding constructed in
Section~\ref{sec:synthesis}.

The full braid image preserves proper summands, so it is not
projectively dense (Section~\ref{sec:gates}). The local operations used
in Section~\ref{app:fixed-encoding} are independently supplied.

The elementary word lengths include the unitary preparation of
reflection directions. Classical preprocessing to obtain the matrices,
Cholesky factors, eigenvectors, and rotation angles is a separate
resource, as are conversion from a preassigned physical encoding and
hardware costs.
Numerical diagnostics and reference compilers are distinguished from
validated error enclosures and hardware comparisons. The circuit bounds
are not quantum-speedup or fault-tolerance results.

\section{The fixed braid action and the canonical Hermitian family}\label{sec:prelim}

We recall the braid action and invariant-form machinery of
\cite{HN,NKLM,Nsig,FiniteCompanion}, fixing the notation needed below.
The braid action on $U_0$ is independent of $\lambda$; the canonical
Hermitian form varies with $\lambda$.

\subsection{Conventions and exceptional parameters}
Neither the input nor the output is assumed to have finite image.
We use $\spec(T)$ for the eigenvalue multiset of $T$, and
$z^{\times r}$ means multiplicity $r$.

\begin{definition}[Unitary KLM seed and coordinate convention]\label{def:seed}
From this section onward, the seed $\rho:F_n\rtimes B_n\to
\operatorname{GL}(V)$ preserves one common positive definite inner
product on the whole group. We choose an orthonormal basis for it,
so the Hermitian matrix of the input invariant form in \eqref{eq:blocks} is $H_0=I_N$.
Except where raw coordinates are explicitly specified, the seed
operators $g_j=\rho(x_j)$ and $s_i=\rho(\sigma_i)$ are written in
this basis and belong to $U(N)$.
\end{definition}
The algebraic construction of Section~\ref{subsec:intro-klm} still
allows every $\lambda\in\C\setminus\{0,1\}$; the unit-circle
restriction below concerns the canonical Hermitian form.
We retain the notation $P_s$, $K$, $L_\lambda$, $U_0$, $Q_\lambda$,
$\kappa$, and $R$ introduced above.

\begin{definition}[Exceptional parameters and chambers]\label{def:resonance}
For the later Hermitian-form discussion, retain the parameter sets of
\cite[Section~2]{Nsig}:
\[
\begin{aligned}
 E&:=\{\lambda\in\C:|\lambda|=1,\ \lambda\ne1,\
                  \Ker(\lambda P_n-I)\ne0\},\\
 E^c&:=\{\lambda\in\C:|\lambda|=1,\ \lambda\ne1,\
                  \Ker(\lambda P_n-I)=0\}.
\end{aligned}
\]
Thus $E$ and $E^c$ are the exceptional and nonexceptional sets of
Definition~\ref{def:klm-quotient}, restricted to $|\lambda|=1$,
$\lambda\ne1$, in the terminology of \cite[Section~6]{FiniteCompanion}.
The angle values of exceptional parameters are the spectral walls.
The parametrization $l\mapsto\ee l$, $0<l<1$, identifies these sets
with their angular inverse images $\e^{-1}(E)$ and $\e^{-1}(E^c)$.
A \emph{chamber} is a connected component of $\e^{-1}(E^c)$.
Here $E^c$ is the complement of $E$ only among complex numbers of
modulus one other than $1$, not in $\C\setminus\{0,1\}$.
This notation does not restrict the algebraic construction below.
\end{definition}

\subsection{\texorpdfstring{$\lambda$}{lambda}-independence of the braid action and the augmentation map}
\begin{proposition}[{$\lm$-independence of the braid part, \cite[Thm.~0.1(1)]{HN}, \cite[Cor.~1.2]{Nsig}}]
\label{prop:lambda-free}
The matrices \eqref{eq:braid-matrices} do not involve $\lm$: the
restriction of $\LM_\lm$ to $B_n$ is the Long--Moody functor of
\cite{Long}. Since $K$ is invariant, $S_i$ induces an
operator $A_i$ on $U_0$. Consequently, for every
$\lambda\in\C\setminus\{0,1\}$ with $\Ker(\lambda P_n-I)=0$,
the braid representation on $Q_\lambda$, identified with $U_0$,
is the same fixed family $(A_i)$. If $\Ker(\lambda P_n-I)\ne0$,
$Q_\lambda$ is the further
quotient of $U_0$ by the image of $L_\lambda$. In particular, on
$|\lambda|=1$, these two cases correspond to $\lambda\in E^c$ and
$\lambda\in E$, respectively.
\end{proposition}

Throughout, $S_i$ denotes the Long--Moody operator on
$V^{\oplus n}$ and $A_i$ its induced operator on $U_0$, explicitly
$A_i[v]=[S_iv]$. Thus $g_j,s_i$ are seed operators, whereas $G_j,S_i$
are the operators of $\rho^{\mathrm{LM}}_\lambda$.
For matrix calculations on $U_0$, we choose one basis of this fixed
vector space, independently of $\lambda$.
We use $\mathsf{C}$ (or $\mathsf{C}_l$ when the parameter $l$ is shown)
for a unitarizing factor, so that $S_i$ is reserved for the LM braid action.
When $K=0$, we use the direct-sum basis of $V^{\oplus n}$ and
identify the two spaces, and hence the two operators, without further
comment.

\begin{definition}[Augmentation or multiplication map]\label{def:augmentation}
Following \cite[Proposition~6.1]{Nsig} and
\cite[Lemma~6.1]{FiniteCompanion}, put
\[
D:V^{\oplus n}\longrightarrow V,\qquad
D(v_1,\dots,v_n):=\sum_{j=1}^n(g_j-I)v_j.
\]
Since $D(K)=0$, it induces $\bar D:U_0\to J$, where
$J:=\sum_j\operatorname{Im}(g_j-I)$ and $\bar D[v]=Dv$.
The bar indicates passage to the quotient, not complex conjugation.
\end{definition}

\begin{lemma}[Equivariance of the augmentation map]\label{lem:augmentation-intertwining}
The augmentation map and its induced map on $U_0$ satisfy
\[
D S_i=s_iD,\qquad \bar D A_i=s_i\bar D
\qquad(1\le i<n).
\]
\end{lemma}
\begin{proof}
For the $i$th slot one uses
$(g_{i+1}-I)s_i=s_i(g_i-I)$.  For the $(i+1)$st slot,
\[
(g_i-I)s_i g_i+(g_{i+1}-I)s_i(I-g_{i+1})
 =s_i(g_{i+1}-I),
\]
and for every other slot $g_j$ commutes with $s_i$.  These are
exactly the Artin relations for the seed. The quotient identity follows
because $D$ vanishes on $K$.
\end{proof}

\subsection{The canonical Hermitian family}\label{subsec:form}
We now vary the invariant form while keeping the braid operators
$(A_i)$ fixed. Write $\lambda=\ee l$ with
$0<l<1$, and use the square root $\lambda^{1/2}=\exp(\pi il)$.
In the coordinates of Definition~\ref{def:seed}, the canonical matrix
is \eqref{eq:blocks} with $H_0=I_N$. Theorem~\ref{thm:klm-form}
identifies its radical with $K+L_\lambda$ and its induced form on
$Q_\lambda$ with $\Hq(\lambda)$. For $\lambda\in E^c$ this is a
nondegenerate invariant form on the fixed space $U_0$, satisfying
\[
 A_i^\dagger\Hq(\lambda)A_i=\Hq(\lambda)\qquad(1\le i<n).
\]
Since $r(l)=0$ there, Theorem~\ref{thm:import-sig} gives
\[
 \sig\Hq(\ee l)=\bigl(R-m(l),\,m(l)\bigr),
 \qquad \ee l\in E^c.
\]
The two entries count the positive and negative eigenvalues, respectively.

\begin{definition}[Definite window]\label{def:definite-window}
With $m(l)$, $r(l)$, and $R$ as in Theorem~\ref{thm:import-sig}, put
\[
 \Wdef:=\{l\in(0,1):\ee l\in E^c,\ m(l)\in\{0,R\}\}.
\]
By the signature formula, this is exactly the set of angles corresponding to nonexceptional parameters for which $\Hq(\ee l)$ is definite. The sets $E$ and $E^c$
contain complex parameter values, whereas $\Wdef$ contains angles.
\end{definition}

\begin{remark}[Walls]\label{rem:walls}
At $\lambda\in E$ the quotient $Q_\lambda$ has dimension $R-r(l)$, and the signature
jumps by $r(l)$ across the wall \cite[Prop.~6.8]{Nsig}; by
\cite[Prop.~6.11(ii)]{Nsig} the wall quotients of a unitarizable
representation are unitarizable as well. More precisely, the unitary
seed makes $P_n$ unitary and hence semisimple. Put
$\mu=\lambda^{-1}\ne1$ and $V_\mu=\Ker(P_n-\mu I)$.
Then \cite[Proposition~6.3]{FiniteCompanion}
gives a $B_n$-equivariant splitting
\[
 U_0\cong V_\mu\oplus Q_\lambda\qquad\text{as $B_n$-modules},
\]
where $V_\mu$ carries the restricted seed braid action. This is a
splitting of braid representations, not of the full
$F_n\rtimes B_n$-representation. Definiteness of a wall quotient is
not part of $\Wdef$: the statements below concern the fixed space $U_0$.
\end{remark}

\begin{remark}[{A linear-combination formula for the invariant forms, \cite[Prop.~6.1]{Nsig}}]\label{rem:pencil}
With the augmentation map $D$ defined above and
$B:=\Ht(-1)$ (the value at $l=\tfrac12$, with the notation of
\cite[Proposition~6.1]{Nsig}), one has
$\Ht(\ee l)=\cos(\pi l)\,D^{\dagger}D+\sin(\pi l)\,B$,
because each block of \eqref{eq:blocks} is a linear combination of
$\ee{\mp l/2}$ with $l$-independent coefficients. Both coefficient
forms are braid-invariant, so the canonical family lies in a real
subspace of dimension at most two of the space of invariant Hermitian
forms of $(A_i)$; this does not assert that all invariant forms lie in
that subspace. This identity will be used to derive the endpoint
estimates in Proposition~\ref{prop:syn-endpoint-rates}.
\end{remark}


\section{Definite canonical forms and unitary realizations}\label{sec:certificate}

This section recalls the definiteness criterion of \cite{Nsig}, then
proves the seed restriction in Proposition~\ref{prop:irreducible-input}.
The examples distinguish definiteness of the canonical family from
unitarizability of the representation.

\subsection{Definiteness and unitarization}\label{subsec:recalled-criterion}
We apply the invariant forms recalled in Section~\ref{subsec:form}
to realize the fixed braid action by unitary matrices for the standard
inner product.

\begin{theorem}[{Definiteness and unitarization, \cite[Prop.~6.11(ii),(iii)]{Nsig}}]\label{thm:certificate}
If $\Wdef\neq\emptyset$, the braid representation on $U_0$ is
unitarizable. More precisely, for $l\in\Wdef$ choose
$\varepsilon\in\{\pm1\}$ with $H_+:=\varepsilon\Hq(\ee l)>0$ and an
invertible $\mathsf{C}$ with $\mathsf{C}^{\dagger}\mathsf{C}=H_+$; then $U_i:=\mathsf{C}A_i\mathsf{C}^{-1}$ are
unitary braid generators.

If the representation of $B_n$ on $U_0$ is irreducible, the following
are equivalent:
\[
U_0\ \text{is unitarizable as a $B_n$-representation},\qquad
\Wdef\neq\emptyset,\qquad \Wdef=\e^{-1}(E^c) .
\]
Irreducibility here is irreducibility of the braid restriction, not
of the $F_n\rtimes B_n$-module.
\end{theorem}

\begin{proof}[Proof (recalled from \cite{Nsig})]
The forward claim follows from the invariance recalled above from
\cite[Theorems~8--9]{NKLM} and Theorem~\ref{thm:import-sig}, by
conjugation with $\mathsf{C}$. For the
converse let $H_{\mathrm{pos}}$ be a positive-definite invariant form on the
irreducible braid representation. For each $\ee l\in E^c$, invariance of
both forms gives that $H_{\mathrm{pos}}^{-1}\Hq(\ee l)$ commutes with every $A_i$, so
it is a scalar by Schur's lemma; since both forms are Hermitian and
$\Hq(\ee l)$ is non-degenerate, the scalar is real and non-zero.
Hence $\Hq(\ee l)$ is definite for every $\ee l\in E^c$.
\end{proof}

\begin{definition}[Unitary realizations]\label{def:certified}
For each $l\in\Wdef$, choose $\varepsilon\in\{\pm1\}$ so that
$\varepsilon\Hq(\ee l)$ is positive definite.
An invertible $\mathsf{C}$ satisfying $\mathsf{C}^\dagger \mathsf{C}=\varepsilon\Hq(\ee l)$ is a
\emph{unitarizing factor}; the family $U_i=\mathsf{C}A_i\mathsf{C}^{-1}$ is the associated
\emph{unitary realization}, and its $U_i$ are \emph{unitary generators}.
In the fixed coordinates on $U_0$, put
$H_+:=\varepsilon\Hq(\ee l)$ and
$\langle v,w\rangle_{H_+}:=v^\dagger H_+w$. Then $\mathsf C$ is an
isometric isomorphism
\[
 \mathsf C:(\C^R,\langle\ ,\ \rangle_{H_+})
 \longrightarrow(\C^R,\langle\ ,\ \rangle_{\mathrm{std}}),
 \qquad U_i\mathsf C=\mathsf C A_i,
\]
where $\langle v,w\rangle_{\mathrm{std}}=v^\dagger w$.
The same map $\mathsf C$ intertwines every generator. Its matrix is
unitary for the standard inner product on both sides only when $H_+=I_R$.
One choice of $\mathsf C$ is the upper triangular Cholesky factor of
$\varepsilon\Hq(\ee l)$ with positive diagonal.
\end{definition}

The definiteness condition above is sufficient but, without braid
irreducibility, need not be necessary. In contrast,
\cite[Theorem~3.3]{FiniteCompanion} characterizes unitarizability of
both the free-group restriction and the full KLM output by definiteness
of the nonzero convolutions of the irreducible free-group input
constituents; their definite signs need not agree. This componentwise
criterion is not a converse to the global-window condition
$\Wdef\ne\emptyset$; see Example~\ref{ex:empty-window}.

\subsection{An obstruction to irreducible output}
\label{subsec:seed-obstruction}

The following consequences of the augmentation construction
\cite[Prop.~1.10]{HN}, its equivariance
\cite[Lemma~6.1]{FiniteCompanion} (Lemma~\ref{lem:augmentation-intertwining}),
and the invariant-form criterion \cite[Prop.~6.11(iii)]{Nsig}
are recalled for use in Proposition~\ref{prop:irreducible-input}.
\begin{proposition}[Consequences of augmentation]
\label{prop:augmentation}
For $\bar D:U_0\to J$ in Definition~\ref{def:augmentation}, put
$t:=\dim J\le N$ and recall that
$R=\dim U_0=\sum_j\operatorname{rank}(g_j-I)$.
\begin{enumerate}
\item If $R>t>0$, then $\Ker\bar D$ is a non-zero proper invariant
subspace and the braid representation on $U_0$ is reducible. In
particular $U_0$ is reducible whenever $R>N$, e.g.\ whenever $K=0$,
$n\ge2$ and $N>0$.
\item If $U_0\neq0$ is irreducible, then $\bar D$ is injective, the
form $G_0([v],[w]):=\langle\bar D[v],\bar D[w]\rangle_V$ is a
positive-definite invariant form on $U_0$, and hence $U_0$ is
unitarizable and $\Wdef=\e^{-1}(E^c)$.
\end{enumerate}
\end{proposition}

The following restriction on the seed operators is the additional
conclusion proved here beyond these augmentation criteria.
\begin{proposition}[A restriction on seeds with irreducible output]
\label{prop:irreducible-input}
Let $\rho:F_n\rtimes B_n\to\operatorname{GL}(V)$ be a seed preserving
one common positive-definite form on the whole group, written in an
orthonormal basis as in Definition~\ref{def:seed}.
If $U_0\neq0$ is irreducible as a $B_n$-representation, then the seed
operators are of the form $g_j=I+(\zeta-1)\Pi_j$ with a common
$\zeta\neq1$ and orthogonal projectors $\Pi_j$ with
$\Pi_j\Pi_k=0$ for $j\ne k$; in particular the $g_j$ commute.
Thus non-commuting unitary seeds never have an irreducible braid
restriction at nonexceptional parameters.
\end{proposition}
\begin{proof}
Use the positive-definite invariant form $G_0$ from
Proposition~\ref{prop:augmentation}(2).
Both $G_0$ (pulled back to
$V^{\oplus n}$, where it is $D^{\dagger}D$) and
$B:=\Ht(-1)$ vanish on $K$ and are invariant, so by Schur's lemma
$B=\mu\,D^{\dagger}D$ on $V^{\oplus n}$ with $\mu\in\R$.
This uses the positive-definite form $G_0$ on $U_0$; it permits
$\mu=0$ and does not require $-1$ to be nonexceptional or $B$ to be
nondegenerate on $U_0$.

Writing $\delta_j:=g_j-I$, the blocks of $D^{\dagger}D$ are
$\delta_j^{\dagger}\delta_k$, while by Remark~\ref{rem:pencil} (and
\cite[Prop.~6.1]{Nsig}) the blocks of $B$ are $-i\delta_j^{\dagger}\delta_k$
for $j<k$ and $-i(g_j-g_j^{\dagger})$ for $j=k$. Comparing the blocks
$j<k$ gives $(\mu+i)\delta_j^{\dagger}\delta_k=0$, hence $\delta_j^{\dagger}\delta_k=0$:
the images of the $\delta_j$ are mutually orthogonal.

On an eigenvector of
$g_j$ with eigenangle $a\in(0,1)$ the diagonal blocks give
$2\sin(2\pi a)=\mu\cdot4\sin^2(\pi a)$, i.e.\ $\mu=\cot(\pi a)$; since
$a\mapsto\cot(\pi a)$ is injective on $(0,1)$, all non-trivial eigenvalues of all
$g_j$ equal one $\zeta=\ee a$, so $g_j=I+(\zeta-1)\Pi_j$ with $\Pi_j$
the orthogonal projector onto the $\zeta$-eigenspace, and
$\delta_j^{\dagger}\delta_k=|\zeta-1|^2\Pi_j\Pi_k=0$ gives $\Pi_j\Pi_k=0$.
\end{proof}

In the family of Section~\ref{sec:twisted}, $K=0$ and
$\dim U_0=n(n+1)>\dim V=n+1$, so Proposition~\ref{prop:augmentation}(1) already forces
reducibility of the full braid restriction. The relevant irreducibility
question is therefore that of its summand $\mathcal R_n$
(Proposition~\ref{prop:generic-irred}); the obstruction does not
exclude irreducible summands of outputs from noncommuting seeds.

\begin{corollary}[A signature obstruction to irreducibility]
\label{cor:reducibility}
If the signature pair $(n_+(l),n_-(l))$ of the non-degenerate
invariant form $\Hq(\ee l)$ on the fixed space $U_0$ takes at least
three distinct values as $l$ varies over
$\{l\in(0,1):\ee l\in E^c\}$, then the braid representation on
$U_0$ is reducible.
\end{corollary}
\begin{proof}
If it were irreducible, every invariant Hermitian form would be a
real multiple of one fixed non-degenerate invariant form (Schur), and
its non-zero multiples have only the two signatures $(p,q)$, $(q,p)$.
\end{proof}
This is a general Schur-lemma observation. Under the positive-definite
seed hypothesis of this paper, Proposition~\ref{prop:augmentation}(2)
gives a stronger test: a single indefinite canonical form at a
nonexceptional parameter already forces reducibility.

\subsection{Examples: empty windows and Hecke specializations}\label{subsec:definiteness-examples}

\begin{example}[An empty canonical definite locus]\label{ex:empty-window}
The following example shows that, without the irreducibility assumption,
the braid representation can be unitarizable even though its canonical
invariant Hermitian form is indefinite for every admissible parameter
$\lambda$.
Let $n=3$, $N=2$, and take the unitary seed
\[
g_1=g_2=g_3=\begin{pmatrix}\ee{1/10}&0\\0&\ee{9/10}\end{pmatrix},\qquad
s_1=s_2=I_2 ,
\]
which satisfies the Artin relations. Here $K=0$, $U_0=V^{\oplus3}=\C^6$,
and a fixed permutation of coordinates splits the braid
representation into two three-dimensional scalar summands with
$N=1$ and angles $a=\tfrac1{10}$, $a=\tfrac9{10}$.

For a scalar
summand of angle $a$, Theorem~\ref{thm:import-sig} specializes to
$\sig\Hq_a(\ee l)=(3-j,\,j-r_a(l))$, $j=\lfloor l+3a\rfloor$, with
$r_a(l)=1$ iff $l+3a\in\Z$. The walls are $l=\tfrac3{10},\tfrac7{10}$,
and adding the summand signatures gives, off the walls,
\[
\sig\Hq(\ee l)=\begin{cases}(4,2),&0<l<\tfrac3{10},\\
(3,3),&\tfrac3{10}<l<\tfrac7{10},\\ (2,4),&\tfrac7{10}<l<1,\end{cases}
\]
and $(3,2)$, $(2,3)$ at the two walls. So no canonical quotient form
is definite, even including wall parameters, and
Corollary~\ref{cor:reducibility} already predicts reducibility.

Nevertheless the six-dimensional braid representation is
unitarizable: at $l=\tfrac12$ the summand forms have signatures $(3,0)$
and $(0,3)$, so in the summand-ordered basis
\[
H_{\mathrm{pos}}:=\Hq_{1/10}(-1)\oplus\bigl(-\Hq_{9/10}(-1)\bigr)
\]
is a positive-definite invariant form. A sign change on one invariant
summand produces a positive form outside the canonical family.
The two-summand family in \cite[Rem.~7.8(c)]{Nsig} exhibits the
same phenomenon at an exceptional parameter.
The summands are the three-dimensional
unreduced Burau representations of $B_3$ at $\ee{1/10}$ and
$\ee{9/10}$, respectively, each in Squier's definite range
\cite{Squier}.

This is also an instance of the componentwise sign correction in
\cite[Theorem~3.3]{FiniteCompanion}: at $\lambda=-1$ the two
free-group constituents have definite canonical forms of opposite
signs. In this example the positive form $H_{\mathrm{pos}}$ is invariant under the
full KLM action at $\lambda=-1$, not only under its braid restriction.
This unitarizability statement alone does not imply finiteness of the
image: the arithmetic finite-image criteria also require the stated
arithmetic hypotheses and definiteness at all relevant coefficient-field
embeddings \cite[Section~2.2]{FiniteCompanion}.
\end{example}

\begin{example}[Hecke specialization]\label{ex:hecke}
The following specialization illustrates how the general signature
formula recovers the classical Burau--Squier parameter range in which
a definite canonical invariant Hermitian form exists.
Let $0<a<1$, $s_i=sI_N$ with a scalar $s\in U(1)$, and $g_i=g\in U(N)$
with $\spec(g)\subseteq\{1,\ee a\}$ and $\nu:=\dim\Ker(g-\ee a)>0$.
Then \cite[Cor.~7.7, Rem.~7.8]{Nsig} give $\dim U_0=n\nu$ and, for
$\ee l\in E^c$, $\sig\Hq=(\nu(n-j),\nu j)$ with $j=\lfloor l+na\rfloor$ (and
$(\nu(n-J),\nu(J-1))$ on the wall $l+na=J\in\Z$, where $r(l)=\nu$). Thus
$\Wdef\neq\emptyset$ iff $a\in(0,\tfrac1n)\cup(\tfrac{n-1}n,1)$, the
Burau--Squier interval. As a $B_n$-representation on $U_0$, the family
is $\nu$ copies of the $n$-dimensional unreduced Burau representation
at $\ee a$, with every braid generator multiplied by the scalar $s$
\cite{NHecke}. Accordingly the normalized generators $T_i:=s^{-1}A_i$
satisfy the Hecke relation
\[
 (T_i-I)(T_i+\ee a I)=0,
\]
because each Burau block satisfies this quadratic identity.
In particular $n=3$, $\nu=1$ gives a three-dimensional
unreduced $B_3$-representation. This is a consistency check with
\cite{Squier}, not a new representation.
\end{example}


\section{Uniqueness of unitary realizations up to unitary equivalence}\label{sec:gauge}

Different definite forms and unitarizing factors can give different
matrices $U_i$ from the same fixed operators $A_i$. We recall the
standard polar-decomposition argument for unitary equivalence and
apply it to the $\lambda$-independence of \cite[Prop.~6.11(i)]{Nsig}.
This gives one explicit unitary intertwiner for all braid generators
and, in Section~\ref{sec:mult}, for the multiplicativity isomorphism
of \cite[Thm.~0.2]{HN}.

\begin{definition}[Intertwiners and equivalence]\label{def:intertwiners}
For representations $\rho_1,\rho_2$ of the same group, an
\emph{intertwiner} is a linear map $T$ satisfying
$T\rho_1(\gamma)=\rho_2(\gamma)T$ for every group element $\gamma$.
The representations are \emph{equivalent} if such a $T$ is invertible,
and \emph{unitarily equivalent} if $T$ can be chosen unitary for their
specified positive inner products.
\end{definition}

\begin{definition}[Polar decomposition]\label{def:polar-decomposition}
Let $T:\mathcal H_1\to\mathcal H_2$ be an invertible linear map
between finite-dimensional complex Hilbert spaces. Put
\[
P=(T^\dagger T)^{1/2},\qquad
W=TP^{-1}=T(T^\dagger T)^{-1/2},
\]
where $\dagger$ denotes the adjoint for the specified inner products
and $P$ is the unique positive-definite square root of $T^\dagger T$.
Then $W$ is unitary. The factorization $T=WP$ is the
\emph{polar decomposition} of $T$, and $W$ is its
\emph{unitary polar part}.
\end{definition}

\begin{lemma}[Unitary intertwiners from polar decomposition]\label{lem:polar}
Let $(U^{(1)}_i)$, $(U^{(2)}_i)$ be unitary families on Hilbert spaces
of equal dimension and $T$ an invertible intertwiner,
$TU^{(1)}_i=U^{(2)}_iT$. Then $W:=T(T^{\dagger}T)^{-1/2}$ is a unitary
intertwiner. If $(U^{(1)}_i)$ is irreducible, $W$ is unique up to a
phase, and every invariant Hermitian form $H^{(2)}$ of
$(U^{(2)}_i)$ satisfies $T^{\dagger}H^{(2)}T=c\,H^{(1)}$ with $c\in\R$
for any nonzero invariant Hermitian form $H^{(1)}$ of $(U^{(1)}_i)$.
\end{lemma}
\begin{proof}
$T^{\dagger}T=(U^{(2)}_iTU^{(1)-1}_i)^{\dagger}(U^{(2)}_iTU^{(1)-1}_i)
=U^{(1)}_i(T^{\dagger}T)U^{(1)-1}_i$, so the positive operator
$T^{\dagger}T$ and its inverse square root commute with every
$U^{(1)}_i$; hence $W$ intertwines and
$W^{\dagger}W=(T^{\dagger}T)^{-1/2}(T^{\dagger}T)(T^{\dagger}T)^{-1/2}=I$.
No irreducibility is used for existence. Under irreducibility,
uniqueness up to phase is Schur's lemma applied to $W'W^{\dagger}$, and
$T^{\dagger}H^{(2)}T$ is an invariant Hermitian form of $(U^{(1)}_i)$,
hence a real multiple of $H^{(1)}$.
\end{proof}

\begin{theorem}[Unitary equivalence]\label{thm:gauge}
Let $l_1,l_2\in\Wdef$, put $\lambda_\nu:=\ee{l_\nu}$ for $\nu\in\{1,2\}$, and suppose
$\varepsilon_\nu\Hq(\lambda_\nu)>0$ with
$\mathsf{C}_\nu^{\dagger}\mathsf{C}_\nu=\varepsilon_\nu\Hq(\lambda_\nu)$. Write
$U_i^{(\nu)}:=\mathsf{C}_\nu A_i\mathsf{C}_\nu^{-1}$. Then
\[
U_i^{(2)}=W\,U_i^{(1)}\,W^{\dagger}\quad(i=1,\dots,n-1),\qquad
W:=\text{unitary polar part of } \mathsf{C}_2\mathsf{C}_1^{-1},
\]
with the unitary polar part as in Definition~\ref{def:polar-decomposition},
independently of the chambers and signs. Hence the unitary equivalence
class of the braid representation on $U_0$ is an invariant of the seed
alone; for any braid word $b$, $U^{(2)}(b)=W U^{(1)}(b)W^{\dagger}$.
The statement is restricted to $\Wdef$: a definite wall quotient has
dimension $R-r(l)<R$ and is not conjugate to the representation on $U_0$.
\end{theorem}
\begin{proof}
$\mathsf{C}_2\mathsf{C}_1^{-1}$ intertwines the two unitary families because both are
conjugates of the same $(A_i)$ (Proposition~\ref{prop:lambda-free});
apply Lemma~\ref{lem:polar}.
\end{proof}

Note that $T:=\mathsf{C}_2\mathsf{C}_1^{-1}$ itself is in general not unitary:
$T^{\dagger}T=\mathsf{C}_1^{-\dagger}(\mathsf{C}_2^{\dagger}\mathsf{C}_2)\mathsf{C}_1^{-1}
=\mathsf{C}_1^{-\dagger}\varepsilon_2\Hq(\lambda_2)\mathsf{C}_1^{-1}$; the conjugator
$W$ is its unitary polar part.

\begin{remark}
As $F_n\rtimes B_n$-modules the $\KLM_\lm(V)$ can vary with
$\lm$: two of them are non-isomorphic whenever the operators induced
by a free generator on the KLM quotients have different spectral
multisets (the ambient spectra of \eqref{eq:free-matrices} must first
be adjusted for the subspace removed in forming the quotient), and
monodromy of the associated KZ-type systems deforms with $\lm$. The collapse in Theorem~\ref{thm:gauge} is exactly the
forgetting of the $F_n$-part.
\end{remark}


\subsection{Compatibility with composition of KLM functors}\label{sec:mult}

We record the standard application of polar decomposition to
composition of KLM functors. This supplementary result is not used
in the Burau structure or synthesis arguments below. The number of
functor iterations is unrelated to circuit depth.

\begin{definition}[Almost trivial modules]\label{def:almost-trivial}
An \emph{almost trivial} $\C[F_n]$-module is one-dimensional with all
but possibly one of the standard generators acting trivially. Following
\cite[\S5.1]{HN}, let $\mathbf{Mod}^{\mathrm{NT}}$ consist of the
$F_n\rtimes B_n$-modules $V$ such that, for every almost trivial
$\C[F_n]$-module $\mathcal A$,
\[
\Hom_{\C[F_n]}(\mathcal A,V)=0=\Hom_{\C[F_n]}(V,\mathcal A),
\]
where $V$ is restricted to $F_n$ in these Hom spaces.
\end{definition}
For $V\in\mathbf{Mod}^{\mathrm{NT}}$ and
$\lambda_1,\lambda_2\in\C^\times$ with
$\lambda_1,\lambda_2,\lambda_1\lambda_2\ne1$,
\cite[Thm.~0.2]{HN} supplies the natural module isomorphism
\[
\Phi:\KLM_{\lambda_1}(\KLM_{\lambda_2}(V))\to\KLM_{\lambda_1\lambda_2}(V).
\]

\begin{proposition}[Unitary realization]\label{prop:mult}
Suppose the two modules above are equipped with positive-definite
braid-invariant forms $H^{(2)}$ and $H^{(1)}$ respectively, and let
$\mathsf{C}_\nu^{\dagger}\mathsf{C}_\nu=H^{(\nu)}$ for $\nu\in\{1,2\}$, $T:=\mathsf{C}_1\Phi \mathsf{C}_2^{-1}$,
$W:=T(T^{\dagger}T)^{-1/2}$. Then $W$ is a unitary intertwiner between
the unitarized braid representations; if they are irreducible, it is
unique up to a phase.
\end{proposition}
\begin{proof}
$\Phi$ intertwines the braid actions by functoriality, so $T$
intertwines the two unitary families; apply Lemma~\ref{lem:polar}.
\end{proof}

The canonical forms are constructed by Theorem~\ref{thm:klm-form}
and \eqref{eq:blocks}; Theorem~\ref{thm:import-sig} determines their
signatures. To use this signature criterion at successive stages,
each input must carry a positive-definite form invariant under the
whole group $F_n\rtimes B_n$. In particular, this must hold for the
intermediate module $\KLM_{\lambda_2}(V)$.

A sufficient condition is that its canonical quotient form be definite,
since that form is invariant under the whole group
\cite[Theorems~8--9]{NKLM}; changing its sign if necessary gives the
required positive inner product. This includes the nonexceptional
definite case, but $L_{\lambda_2}=0$ is not necessary:
the quotient $Q_{\lambda_2}$ at an exceptional parameter can also be used whenever it
carries a positive definite form invariant under the whole group.
Unitarizability of the intermediate braid restriction alone is not enough.
To obtain $H^{(\nu)}$ from a canonical output form, its definiteness
must also be checked, with a sign change if necessary.

For the canonical forms at each stage,
the complex parameters must satisfy $|\lambda_1|=|\lambda_2|=1$ and
$\lambda_1,\lambda_2,\lambda_1\lambda_2\ne1$, in accordance with
the parameter domain used here. Proposition~\ref{prop:mult} itself
only assumes positive-definite braid-invariant forms on its two
modules.


\section{The character-twisted Burau family}
\label{sec:twisted}

Use the Burau representation and character-twisted seed of
Definitions~\ref{def:burau} and~\ref{def:burau-seed}, with
$q=\ee a$, $\tau=\ee c$ and $V=\C^{n+1}$. Thus $N=n+1$.
In the general KLM formulas we take $g_j=g_j^{\mathrm{tB}}$ and
$s_i=s_i^{\mathrm{tB}}$; the character multiplies only the free-group
operators.
If $1+q+\cdots+q^n\ne0$, the input decomposes as a direct sum of
invariant subspaces
\[
V=\C(1,\dots,1)^T\oplus V_{\mathrm{red}},\qquad
V_{\mathrm{red}}=\Ker\ell_q,
\]
where $\ell_q$ is the invariant functional of Definition~\ref{def:burau}.
This condition holds whenever $0<a<1/(n+1)$; in that range the input
is unitarizable for $|\tau|=1$ \cite{Squier}.
Fix one unitarizing change of basis for the entire seed on $V$.
This change acts on the input space; the output factor $\mathsf C$
of Definition~\ref{def:certified} acts on $U_0$.
Characteristic polynomials are invariant under this simultaneous
conjugation and may be computed in the raw Burau basis, while the
coordinate supports used for synthesis refer to the particular common
basis of Section~\ref{sec:synthesis}.
\begin{definition}[Character-twisted reduced Burau seed]
\label{def:reduced-burau-seed}
The \emph{character-twisted reduced Burau seed} is the restriction
of Definition~\ref{def:burau-seed} to $V_{\mathrm{red}}=\Ker\ell_q$,
with operators
\[
g_j^{\mathrm{tB,red}}:=g_j^{\mathrm{tB}}|_{V_{\mathrm{red}}},\qquad
s_i^{\mathrm{tB,red}}:=s_i^{\mathrm{tB}}|_{V_{\mathrm{red}}}.
\]
After a simultaneous change of basis, $V_{\mathrm{red}}$ denotes
the image of this subspace in the new coordinates.
\end{definition}

Long's construction and its one-parameter extension are recalled in
\cite[Thm.~2.4, Cor.~2.6]{Long}; the rescaling displayed there
multiplies the braid generators by $s$.
Here we use the free-generator rescaling described in
\cite[\S4, proof of Theorem~4.1 and Theorems~4.2--4.3]{BT},
multiplying the free-group input operators by $\tau$ and leaving
the braid-group input operators unchanged. We first compute the spectral data
needed to apply the signature formula of \cite{Nsig}; this specifies
the unitary family used in the rest of the paper. We then determine
its generator spectra and invariant decomposition, and prove generic
irreducibility of the remaining summand for $n=3,4$. Throughout we write
\[
\mathcal T_n:=\{(a,c):\ a>0,\ c>0,\ nc+(n+1)a<1\},\qquad
d:=\tfrac32a+c,\qquad \omega:=\ee d ,
\]
and note that on $\mathcal T_n$ one has $0<a,c<d<\tfrac12$ and
$(n+1)a<1$. We set $M:=n(n+1)$ for the dimension of the KLM braid space in this region.

\subsection{The positive window: boundary spectrum and signature}\label{subsec:burau-signature}
The boundary-loop calculation below supplies the eigenangles in
\cite[Cor.~6.10]{Nsig}. Theorem~\ref{thm:twisted-burau} is its explicit
specialization to the character-twisted Burau seed.

\begin{lemma}[Spectrum of the boundary loop]\label{lem:fulltwist}
Let $X_n:=x_1\cdots x_n\in B_{1,n}$ be the boundary loop. For
the unreduced Burau representation of $B_{n+1}$ and every
$q\in\C^\times$,
\[
 \det\bigl(zI_{n+1}-\rho_q(X_n)\bigr)
 =(z-1)(z-q)^{n-1}(z-q^{n+1}).
\]
\end{lemma}
\begin{proof}
\emph{Boundary-loop identity.}
Define the half twists
\[
 \theta_{n+1}:=\prod_{r=n-1,\ldots,0}(\sigma_0\cdots\sigma_r),
 \qquad
 \theta_n:=\prod_{r=n-1,\ldots,1}(\sigma_1\cdots\sigma_r),
\]
where the factors are ordered from left to right in decreasing $r$.
These act on strands $0,\ldots,n$ and $1,\ldots,n$, respectively.
Their squares are the full twists
$\theta_{n+1}^2=(\sigma_0\cdots\sigma_{n-1})^{n+1}$ and
$\theta_n^2=(\sigma_1\cdots\sigma_{n-1})^n$.

We first verify the boundary-loop identity for arbitrary $n$.
The definition of $x_j$ gives
\[
x_j=(\sigma_{j-1}\cdots\sigma_1)\sigma_0^2
(\sigma_1^{-1}\cdots\sigma_{j-1}^{-1}).
\]
The Artin relations then give
\[
 \sigma_0x_j\sigma_0^{-1}
 =(\sigma_{j-1}\cdots\sigma_2)\sigma_1^2
   (\sigma_2^{-1}\cdots\sigma_{j-1}^{-1})
 \qquad(j\ge2).
\]
Indeed, $\sigma_0$ commutes with $\sigma_k$ for $k\ge2$, and
$\sigma_0\sigma_1\sigma_0^2\sigma_1^{-1}\sigma_0^{-1}
=\sigma_1^2$. Writing $\operatorname{sh}$ for the shift
$\sigma_i\mapsto\sigma_{i+1}$, we obtain
$X_n=\sigma_0\operatorname{sh}(X_{n-1})\sigma_0$.
Starting from $X_1=\sigma_0^2$, induction yields
\[
 X_n=C\bar C,\qquad
 C:=\sigma_0\cdots\sigma_{n-1},\quad
 \bar C:=\sigma_{n-1}\cdots\sigma_0.
\]
Put $\Delta:=\theta_{n+1}$ and $\delta:=\theta_n$.
The Artin relations also give
$\Delta=\delta C=\bar C\delta$ and
$\Delta\sigma_i=\sigma_{n-1-i}\Delta$.
Thus $\Delta^2$ is central, and
\[
 \Delta^2=\bar C\delta^2C,
 \qquad
 \Delta^2=C\Delta^2C^{-1}=C\bar C\delta^2=X_n\delta^2.
\]
Consequently $X_n=\theta_{n+1}^2\theta_n^{-2}$.

\emph{Full-twist matrix.}
In the present Burau convention the full twist on
$m$ strands has matrix
\begin{equation}\label{eq:burau-full-twist}
 T_m:=\rho_q\bigl((\sigma_0\cdots\sigma_{m-2})^m\bigr)
 =q^mI_m+(1-q)\mathbf1_m(1,q,\ldots,q^{m-1}),
\end{equation}
where $\mathbf1_m=(1,\ldots,1)^T$.
To see this for arbitrary $m\ge2$, the matrix
$C_m:=\rho_q(\sigma_0\cdots\sigma_{m-2})$ has first row
$(1-q,q(1-q),\ldots,q^{m-2}(1-q),q^{m-1})$
and entries $1$ immediately below the diagonal, with all other
entries zero. Hence
\[
 \det(zI_m-C_m)
 =(z-1)(z^{m-1}+qz^{m-2}+\cdots+q^{m-1}).
\]
If $1+q+\cdots+q^{m-1}\ne0$, then
\[
 \C^m=\C\mathbf1_m\oplus\ker(1,q,\ldots,q^{m-1}).
\]
Both summands are $C_m$-invariant, and the characteristic polynomial
of its restriction to the second summand is
$z^{m-1}+qz^{m-2}+\cdots+q^{m-1}$.
The Cayley--Hamilton theorem on this summand, followed by multiplication
by $C_m-qI$, gives $C_m^m=q^mI$ there, whereas $C_m^m=I$ on the
fixed line. The projection onto that line is
$\mathbf1_m(1,q,\ldots,q^{m-1})/(1+q+\cdots+q^{m-1})$,
which proves~\eqref{eq:burau-full-twist}; the identity extends to
all $q\ne0$ by polynomial continuation.

\emph{Boundary spectrum.}
Put $[n]_q:=1+q+\cdots+q^{n-1}$. For $[n]_q\ne0$ and
$1+q+\cdots+q^n\ne0$, use the direct sum
$\C^{n+1}=L\oplus W\oplus L'$, where
\[
 \begin{aligned}
 L&=\C(1,\ldots,1)^T,\\
 W&=\bigl\{(0,w_1,\ldots,w_n)^T:
                 \textstyle\sum_{j=1}^n q^{j-1}w_j=0\bigr\},\\
 L'&=\C(-q[n]_q,1,\ldots,1)^T.
 \end{aligned}
\]
Here $W\oplus L'=V_{\mathrm{red}}$, while $W$ is the reduced
Burau space for the subgroup $B_n$ and $L'$ is fixed by that subgroup.
Formula~\eqref{eq:burau-full-twist} shows that
$\rho_q(\theta_{n+1}^2)$ acts on $L,W,L'$ by
$1,q^{n+1},q^{n+1}$, respectively, and
$\rho_q(\theta_n^2)$ acts by $1,q^n,1$.
Thus $\rho_q(X_n)$ acts by $1,q,q^{n+1}$ on these spaces,
whose dimensions are $1,n-1,1$. This gives the claimed
characteristic polynomial for generic $q$.
For each fixed $n$, both sides are Laurent polynomials in $q$
with polynomial dependence on $z$, so the identity holds for every
$q\ne0$, including values where the direct sum above does not exist.
\end{proof}
For the character-twisted seed, the corresponding matrix product is
\[
 P_n=g_1^{\mathrm{tB}}\cdots g_n^{\mathrm{tB}}=\tau^n\rho_q(X_n).
\]

\begin{theorem}[{Definite window; an instance of \cite[Cor.~6.10]{Nsig}}]\label{thm:twisted-burau}
Let $(a,c)\in\mathcal T_n$. Then for the character-twisted Burau seed:
\begin{enumerate}
\item $\spec(g_i^{\mathrm{tB}})=\{\tau^{\times n},\tau q^2\}$, hence $K=0$,
$R=n(n+1)$ and $U_0=V^{\oplus n}=\C^{n(n+1)}$;
\item $\spec(\lm P_n)=\{\ee{l+nc},\ \ee{l+nc+a}^{\times(n-1)},\
\ee{l+nc+(n+1)a}\}$, so
$\e^{-1}(E)=\{1-nc-(n+1)a,\ 1-nc-a,\ 1-nc\}$ with $r(l)=1,\,n-1,\,1$, respectively;
\item for $\ee l\in E^c$,
$\sig\Ht(\lm)=(n(n+1)-m(l),\,m(l))$ with
\[
m(l)=
\begin{cases}
0,&0<l<1-nc-(n+1)a,\\
1,&1-nc-(n+1)a<l<1-nc-a,\\
n,&1-nc-a<l<1-nc,\\
n+1,&1-nc<l<1.
\end{cases}
\]
Put $l_{\mathrm{first}}:=1-nc-(n+1)a$.
In particular $\Ht(\lm)$ is positive definite
exactly for $0<l<1-nc-(n+1)a$, i.e.\ $\Wdef=(0,1-nc-(n+1)a)$, and it
is never negative definite.
\end{enumerate}
\end{theorem}
\begin{proof}
(1) $\rho_q(\sigma_0^2)$ has eigenvalues $1^{\times n}$, $q^2$, and the
$g_i^{\mathrm{tB}}$ are conjugate to $\tau\rho_q(\sigma_0^2)$; $\tau\ne1$,
$\tau q^2\ne1$ since $0<c<c+2a\le nc+(n+1)a<1$.

(2) This follows from $P_n=\tau^n\rho_q(x_1\cdots x_n)$ and
Lemma~\ref{lem:fulltwist}.

(3) We apply
\cite[Cor.~6.10]{Nsig}: with $\gamma_k$ the eigenangles of $P_n$,
namely $nc$, $nc+a$ ($n-1$ times) and $nc+(n+1)a$, all $<1$ on
$\mathcal T_n$, and with $\sum_{j,k}a^{(k)}_j=n(nc+c+2a)$ by (1),
the constant $w=\sum_{j,k}a^{(k)}_j-\sum_k\gamma_k$ equals
$n(n+1)c+2na-\bigl(n(n+1)c+2na\bigr)=0$, and
$m(l)=\#\{k:\,1-\gamma_k\le l\}$. When $\ee l\in E^c$ no equality occurs, so
this agrees with the strict inequalities in the displayed step function.
At a wall the weak inequality is retained and the quotient signature is
$(R-m(l),m(l)-r(l))$ by Theorem~\ref{thm:import-sig}.
Since the largest eigenangle is $\gamma_{\max}=nc+(n+1)a$,
$l_{\mathrm{first}}=1-\gamma_{\max}$ is the first spectral wall.
Thus $m(l)=0$ precisely for $0<l<l_{\mathrm{first}}$, so the form is
positive definite exactly on that interval (and the corresponding
wall quotient is positive definite at $l=l_{\mathrm{first}}$).
Moreover the displayed step function gives $m(l)\le n+1<R=n(n+1)$
for $n\ge2$, so negative definiteness never occurs.
\end{proof}

\begin{proposition}[Unitarization and the canonical positive window]
\label{prop:twisted-unitarization}
For the character-twisted Burau seed, the following hold.
\begin{enumerate}
\item If $(a,c)\in\mathcal T_n$, the fixed braid representation on
$\C^{n(n+1)}$ is unitarizable. For every $l\in\Wdef$, any factor
$\mathsf C^\dagger\mathsf C=\Ht(\ee l)$ gives a unitary realization;
its unitary equivalence class depends only on $(q,\tau)$.
\item If $a,c>0$, $(n+1)a<1$, and
$c+2a<1\le nc+(n+1)a$, the canonical form is not positive definite
for any $l\in(0,1)$ with $\ee l\in E^c$.
\end{enumerate}
\end{proposition}
\begin{proof}
(1) Theorem~\ref{thm:twisted-burau} gives a nonempty positive-definite
window. Apply Theorem~\ref{thm:certificate} to obtain the unitary
realizations, and Theorem~\ref{thm:gauge} for their unitary equivalence.
The braid operators are independent of $\lambda$ by
Proposition~\ref{prop:lambda-free}.

(2) The assumption $(n+1)a<1$ still gives a unitary Burau seed.
The inequality $c+2a<1$ keeps its eigenangles equal to $c$ and
$c+2a$, without a wrap. Lemma~\ref{lem:fulltwist} and
Definition~\ref{def:eigenangles} therefore give
\[
 m(l)=\lfloor l+nc\rfloor+(n-1)\lfloor l+nc+a\rfloor
       +\lfloor l+nc+(n+1)a\rfloor\ge1,
\]
since $nc+(n+1)a\ge1$. For $\ee l\in E^c$, the signature formula
therefore excludes positive definiteness.
\end{proof}

\begin{remark}\label{rem:chosen-region}\label{rem:chosen-certified-region}
(a) Appendix~\ref{app:additional-definite-regions} explains the
restriction to $\mathcal T_n$ and describes the additional definite
regions outside it.

(b) For the untwisted seed ($c=0$, $\tau=1$) with
$0<a<1/(n+1)$, one has $\kappa=n^2$ and $\dim U_0=n$, and $U_0\cong V_{\mathrm{red}}|_{B_n}
\cong\mathbf1\oplus\mathrm{Bur}^{\mathrm{red}}_n(q)$, the unreduced
Burau representation of $B_n$ at $q$.
Here each $g_j^{\mathrm{tB}}-I$ has rank one: the middle convolution
of this pseudo-reflection system is again a pseudo-reflection system.

Indeed $\Img(g_1^{\mathrm{tB}}-I)$ is spanned by $v_1:=-qe_0+e_1$.
Conjugating $x_1$ by $\sigma_{j-1}\cdots\sigma_1$ gives $x_j$, so
$\Img(g_j^{\mathrm{tB}}-I)$ is spanned by
\[
 v_j:=\rho_q(\sigma_{j-1}\cdots\sigma_1)v_1
 =-qe_0+(1-q)(e_1+\dots+e_{j-1})+e_j.
\]
This follows by induction, using
$\rho_q(\sigma_j)e_j=(1-q)e_j+e_{j+1}$ and $\rho_q(\sigma_j)e_i=e_i$
for $i\notin\{j,j{+}1\}$. The $v_j$ are linearly independent, so
$\operatorname{rank}D=n=\dim U_0$, the map $\bar D$ of
Proposition~\ref{prop:augmentation} is an isomorphism onto
$J=V_{\mathrm{red}}$, and $U_0\cong V_{\mathrm{red}}|_{B_n}$. The
character twist makes the kernel summand survive and lets the full
$n(n+1)$-dimensional module survive.
\end{remark}

\subsection{Generator spectra, invariant summands, and Long's reduced construction}\label{sec:spectrum}

The generator polynomial specializes the general determinant formula
of \cite[Thm.~3.1]{NHecke}; we include a direct block calculation.
It then determines the spectra of the summands in
Proposition~\ref{prop:decomp}, before the generic irreducibility
argument for $\mathcal R_n$.

\begin{proposition}[Characteristic polynomial]\label{prop:charpoly}
Let $n\ge2$ and $q,\tau\in\C^\times$. Take as input the
character-twisted unreduced Burau representation of
Definition~\ref{def:burau-seed} on $V=\C^{n+1}$, with
$g_j^{\mathrm{tB}}=\tau\rho_q(x_j)$ and $s_i^{\mathrm{tB}}=\rho_q(\sigma_i)$, where
$\rho_q$ is the unreduced Burau representation of $B_{n+1}$.
For every $1\le i\le n-1$ and every
$\lambda\in\C\setminus\{0,1\}$, the Long--Moody braid matrix
$S_i$ of \eqref{eq:braid-matrices} on $V^{\oplus n}$ has
characteristic polynomial
\[
\det(z-S_i)=(z-1)^{n(n-1)}(z+q)^{n-1}(z+\tau)^{n-1}
(z^2-q^{3}\tau^{2}) .
\]
For $(a,c)\in\mathcal T_n$ one has $K=0$, so the same formula holds
for the braid operator $A_i$ on $U_0$ and for every unitary
realization $U_i$ at a nonexceptional parameter; in this region the eigenvalue
multiset is
\[
\{\,1^{\times n(n-1)},\ (-q)^{\times(n-1)},\ (-\tau)^{\times(n-1)},\ \omega,\ -\omega\,\},
\qquad \omega=q^{3/2}\tau=\ee d ,
\]
with eigenangles $0,\ a+\tfrac12,\ c+\tfrac12,\ d,\ d+\tfrac12$; since
$0<a,c<d<\tfrac12$ on $\mathcal T_n$, the five values are distinct
when $a\ne c$, and for $a=c$ the multiplicities of $-q=-\tau$ add.
(Outside $\mathcal T_n$ the quotient may be smaller: for $n=2$,
$\tau=1$, $q=\ee{1/10}$ one has $\dim K=4$ and
$\det(z-A_1|_{U_0})=(z-1)(z+q)$.)
\end{proposition}
\begin{proof}
Compute in the raw Burau basis with $i=1$ (for $n\ge3$ all $\sigma_i$
are conjugate in $B_n$, and for $n=2$ there is one generator). By
\eqref{eq:braid-matrices}, $S_1$ acts on the slots $j\ge3$ by
$s_1^{\mathrm{tB}}=\rho_q(\sigma_1)$, contributing $n-2$ copies of
$\det(z-\rho_q(\sigma_1))=(z-1)^n(z+q)$.

On the two slots $1,2$ the
operators $s_1^{\mathrm{tB}}$, $g_1^{\mathrm{tB}}$, $g_2^{\mathrm{tB}}=s_1^{\mathrm{tB}}g_1^{\mathrm{tB}}(s_1^{\mathrm{tB}})^{-1}$ act
non-trivially only on the raw Burau coordinates $0,1,2$; on each of the
remaining $n-2$ raw Burau coordinates they act by $1$, $\tau$, $\tau$, so
the two-slot block restricted to such a coordinate is
$B_\tau=\begin{pmatrix}0&\tau\\1&1-\tau\end{pmatrix}$, with
$\det(z-B_\tau)=(z-1)(z+\tau)$.

We denote by $s_1^{\mathrm{tB},\flat}$, $g_1^{\mathrm{tB},\flat}$, and $g_2^{\mathrm{tB},\flat}$
the restrictions of $s_1^{\mathrm{tB}}$, $g_1^{\mathrm{tB}}$, and $g_2^{\mathrm{tB}}$, respectively,
to the invariant coordinate subspace
$\operatorname{span}\{e_0,e_1,e_2\}$ in the raw Burau basis.
On this subspace in the two slots, the block is the $6\times6$ matrix
\[
\begin{gathered}
C(q,\tau)=\begin{pmatrix}0&s_1^{\mathrm{tB},\flat}g_1^{\mathrm{tB},\flat}\\ s_1^{\mathrm{tB},\flat}&s_1^{\mathrm{tB},\flat}(I-g_2^{\mathrm{tB},\flat})\end{pmatrix},\\
g_2^{\mathrm{tB},\flat}=s_1^{\mathrm{tB},\flat}g_1^{\mathrm{tB},\flat}(s_1^{\mathrm{tB},\flat})^{-1}.
\end{gathered}
\]
\[
s_1^{\mathrm{tB},\flat}=\begin{pmatrix}1&0&0\\0&1-q&q\\0&1&0\end{pmatrix},\qquad
g_1^{\mathrm{tB},\flat}=\tau\begin{pmatrix}1-q&q&0\\1&0&0\\0&0&1\end{pmatrix}^{2},
\]
which does not depend on $n$, and a direct computation (a Schur
complement reduces it to a $3\times3$ determinant) gives
$\det(z-C)=(z-1)^2(z+q)(z+\tau)(z^2-q^3\tau^2)$.

Hence
$S_1\sim C(q,\tau)\oplus B_\tau^{\oplus(n-2)}\oplus
\rho_q(\sigma_1)^{\oplus(n-2)}$ as a single matrix (this is not a
decomposition of the braid representation), and multiplying the three
characteristic polynomials gives the formula. On $\mathcal T_n$ the
five values are distinct iff $a\ne c$ because $0<a,c<d<\tfrac12$.
\end{proof}

\medskip\noindent
The polynomial can also be obtained by specializing the general
formula $\det(z-S_i)=\det(z-s_i)^{n-1}\det(z+s_ig_i)$ of
\cite[Thm.~3.1]{NHecke}: for the present seed,
$\det(z-s_i^{\mathrm{tB}})=(z-1)^n(z+q)$ and
$\det(z+s_i^{\mathrm{tB}}g_i^{\mathrm{tB}})=(z+\tau)^{n-1}(z^2-q^3\tau^2)$. The block
argument above is a direct proof of this specialization.

We next describe subspaces invariant under all braid generators.
Their direct sum is a decomposition of the representation; the
tensor-product identifications used for gates are specified separately
in Section~\ref{sec:gates}. The decomposition combines additivity of
the Long--Moody construction, Burau branching, and the augmentation
map with complete reducibility on $\mathcal T_n$.

\begin{proposition}[An exact decomposition]\label{prop:decomp}
For $(a,c)\in\mathcal T_n$, the unitary braid representation on
$U_0=\C^{n(n+1)}$ decomposes as
\[
U_0\;\cong\;\mathbf 1^{\oplus2}\oplus\mathrm{Bur}^{\mathrm{red}}_n(\tau)
\oplus\mathrm{Bur}^{\mathrm{red}}_n(q)\oplus\mathcal R_n(q,\tau),
\qquad \dim\mathcal R_n=n(n-1),
\]
where $\mathrm{Bur}^{\mathrm{red}}_n(\cdot)$ is the reduced Burau
representation of $B_n$ and $\mathcal R_n$ is the kernel of the
augmentation map $D$ on the Long--Moody module of the reduced
Burau seed. The generator spectrum of $\mathcal R_n$ is
$\{1^{\times(n-1)(n-2)},(-q)^{\times(n-2)},(-\tau)^{\times(n-2)},
\omega,-\omega\}$. No irreducibility is required for this decomposition;
generic irreducibility for $n=3,4$ is proved below. For $n=2$,
$\mathcal R_2$ is a two-dimensional representation of the cyclic
group $B_2$, hence reducible.
\end{proposition}
\begin{proof}
By Lemma~\ref{lem:augmentation-intertwining}, the augmentation
$D:V^{\oplus n}\to V$ is a braid intertwiner. The unitary seed splits as $\mathbf 1\oplus
\mathrm{Bur}^{\mathrm{red}}_{n+1}(q)$ (both with the same character
twist), and the Long--Moody construction is additive.

The trivial-seed
summand is the scalar-seed summand with $g=\tau$, $s=1$, whose output
has the local block $\left(\begin{smallmatrix}0&\tau\\1&1-\tau\end{smallmatrix}\right)$.
This is the unreduced Burau representation in the Long--Moody convention
of \cite[Proposition~2.1]{BT}, up to simultaneous change of basis from
the raw convention of Section~\ref{sec:twisted}.
For this scalar identification, note that
\cite[\S\S~1--2, Proposition~2.1]{BT} uses the relation
$g_{i+1}\sigma_i=\sigma_i g_i$ and the block
$R_i=\left(\begin{smallmatrix}0&g_i\\1&1-g_i\end{smallmatrix}\right)$,
with $\sigma_i$ multiplying every matrix entry on the left.
Specializing their free generators to $\tau$ and their braid
generators to $1$ gives the displayed block.
Since $0<c<1/n$ on $\mathcal T_n$, we have $\tau^n\ne1$ and hence
$1+\tau+\cdots+\tau^{n-1}\ne0$. Thus this unreduced Burau
representation splits as $\mathbf1\oplus\mathrm{Bur}^{\mathrm{red}}_n(\tau)$.

On the reduced-seed summand $(V_{\mathrm{red}})^{\oplus n}$, $g_1^{\mathrm{tB,red}}-I$ is
invertible on $\mathcal T_n$ (its eigenvalues are $\tau-1$ and
$\tau q^2-1$), so $D$ is surjective onto $V_{\mathrm{red}}$ with
kernel $\mathcal R_n$ of dimension $n^2-n$, and by the intertwining
relation the quotient $(V_{\mathrm{red}})^{\oplus n}/\mathcal R_n$
is the restriction of
$\mathrm{Bur}^{\mathrm{red}}_{n+1}(q)$ to $B_n$, which is
$\mathbf1\oplus\mathrm{Bur}^{\mathrm{red}}_n(q)$.
Indeed, $0<a<1/(n+1)$ makes both geometric sums in the proof of
Lemma~\ref{lem:fulltwist} nonzero. Its decomposition
$V_{\mathrm{red}}=W\oplus L'$ identifies $W$ with
$\mathrm{Bur}^{\mathrm{red}}_n(q)$ and $L'$ with a fixed line of $B_n$.

Since the whole representation is unitarizable on $\mathcal T_n$
(Proposition~\ref{prop:twisted-unitarization}(1)), every invariant subspace has an
invariant complement, so all these exact sequences split.

The spectrum of $\mathcal R_n$ follows from
Proposition~\ref{prop:charpoly} by dividing out the characteristic
polynomials of the other summands, $(z-1)^2$, $(z-1)^{n-2}(z+\tau)$
and $(z-1)^{n-2}(z+q)$.
\end{proof}

\begin{remark}[An explicit equivariant splitting]\label{rem:splitting}
Let
\[
\begin{gathered}
F:V\longrightarrow V^{\oplus n},\\
F(v)=(g_2^{\mathrm{tB}}\cdots g_n^{\mathrm{tB}}v,\ g_3^{\mathrm{tB}}\cdots g_n^{\mathrm{tB}}v,\ \dots,\ g_n^{\mathrm{tB}}v,\ v),
\end{gathered}
\]
the map in \eqref{eq:L} without the eigenvalue condition.
By \cite[Lemma~6.1]{FiniteCompanion}, $F$ and $D$ are braid
intertwiners and $DF=P_n-I$. On $\mathcal T_n$, $P_n-I$ is
invertible by Theorem~\ref{thm:twisted-burau}(2). Therefore
\[
J_0=F(P_n-I)^{-1},\qquad DJ_0=I,\qquad
V^{\oplus n}=F(V)\oplus\Ker D.
\]
This splitting is algebraic and does not require a positive form.
The further Burau splittings follow from the nonvanishing geometric
sums in the proof of Proposition~\ref{prop:decomp}; positivity makes
the resulting summands unitary.
\end{remark}

Long's reduced construction and the difference-coordinate map in
\cite[Theorem~2.11 and its proof]{Long} give the module and map below.
We record them in our Artin-action conventions and verify that, on
$\mathcal T_n$, the map restricts to an isomorphism from the
augmentation-kernel summand $\mathcal R_n$ onto the reduced construction.

\begin{proposition}[Identification with Long's reduced construction]
\label{prop:long-reduced}
On $\mathcal T_n$, the summand $\mathcal R_n$ is isomorphic to Long's
reduced construction \cite[Thm.~2.11]{Long} applied to the
character-twisted reduced Burau seed.
\end{proposition}
\begin{proof}
In this proof, let $V=V_{\mathrm{red}}$, so that $\dim V=n$,
and use the restricted operators $g_j^{\mathrm{tB,red}}$ and
$s_i^{\mathrm{tB,red}}$ of Definition~\ref{def:reduced-burau-seed}.
Here $D$, $F$, and $P_n$ are those of the general construction with
$g_j=g_j^{\mathrm{tB,red}}$ and $s_i=s_i^{\mathrm{tB,red}}$. Write
$\mathcal R_n=\Ker D\subset V^{\oplus n}$ and define
\[
 \Psi:V^{\oplus n}\longrightarrow V^{\oplus(n-1)},\qquad
 (\Psi v)_j=v_j-g_{j+1}^{\mathrm{tB,red}}v_{j+1}\quad(1\le j<n).
\]
Then $\Psi$ is surjective and $\Ker\Psi=F(V)$, where $F$ is the equivariant
map of Remark~\ref{rem:splitting}. Since $D F=P_n-I$ and
$P_n-I$ is invertible on $\mathcal T_n$, one has
$\Ker D\cap\Ker\Psi=0$. Thus $\Psi|_{\mathcal R_n}$ is injective.
Since
\[
\dim\Ker D=n(n-1)=\dim V^{\oplus(n-1)},
\]
it is an isomorphism from $\mathcal R_n=\Ker D$ onto $V^{\oplus(n-1)}$.

Writing $u=\Psi v$, the Artin relations give $\Psi S_i=A_i^{\mathrm{red}}\Psi$, where, omitting
nonexistent endpoint coordinates,
\[
(A_i^{\mathrm{red}}u)_j=
\begin{cases}
s_i^{\mathrm{tB,red}}(u_{i-1}+g_i^{\mathrm{tB,red}} u_i),&j=i-1,\\
-s_i^{\mathrm{tB,red}} g_i^{\mathrm{tB,red}} u_i,&j=i,\\
s_i^{\mathrm{tB,red}}(u_i+u_{i+1}),&j=i+1,\\
s_i^{\mathrm{tB,red}} u_j,&\text{otherwise}.
\end{cases}
\]
These are Long's reduced matrices \cite[Thm.~2.11]{Long}; in the
present conventions his $\rho(g_{i+1}\sigma_i)$ is $s_i^{\mathrm{tB,red}} g_i^{\mathrm{tB,red}}$.
Thus $\Psi|_{\mathcal R_n}$ is the desired braid intertwiner.
\end{proof}

The decomposition leaves open whether $\mathcal R_n$ has further
invariant summands. For $n=3,4$, the following result shows that it
is generically irreducible, of dimension $6$ and $12$, respectively.

\begin{proposition}[Generic irreducibility of $\mathcal R_n$ for $n=3,4$]
\label{prop:generic-irred}
For $n\in\{3,4\}$ there is a non-zero Laurent polynomial
$f_n(q,\tau)$ such that $\mathcal R_n(q,\tau)$ is irreducible for all
$(a,c)\in\mathcal T_n$ with $f_n(\ee a,\ee c)\neq0$; in particular on a
dense open subset of $\mathcal T_n$.
\end{proposition}
\begin{proof}
On $\mathcal T_n$, Proposition~\ref{prop:decomp} gives
\[
\dim\operatorname{End}_{B_n}(U_0)
\ge 4+1+1+\dim\operatorname{End}_{B_n}(\mathcal R_n)\ge7.
\]
Here $4$ comes from the two trivial summands, and the next two
terms from the scalar endomorphisms of the two reduced Burau summands.
We prove generic irreducibility by showing that this lower bound is
attained wherever a suitable minor is nonzero.

The commutant of the braid representation on $U_0$ is the solution
space of the linear system $S_iX-XS_i=0$
($1\le i\le n-1$), whose coefficient matrix
$\mathcal C_n(q,\tau)$ has entries in $\Z[q,q^{-1},\tau]$ (raw Burau
basis). Its $M^2$ unknowns are the entries of $X$, and its kernel is
exactly $\operatorname{End}_{B_n}(U_0)$. An exact computation over $\mathbb F_{1009}$ at the
specialization $q=2$, $\tau=3$ exhibits a non-zero minor of size
$M^2-7$ (ranks $137$ of $144$ for $n=3$ and $393$ of $400$
for $n=4$; the selected determinants modulo $1009$ are $786$ and
$962$, respectively, and their row and column indices are recorded by
the ancillary script). The kernel dimensions at these specializations
are therefore $144-137=7$ and $400-393=7$.

The corresponding minor of $\mathcal C_n(q,\tau)$ is a
Laurent polynomial $f_n$ that is non-zero modulo $1009$, hence
non-zero, so $\dim\operatorname{End}_{B_n}(U_0)\le7$ wherever
$f_n(q,\tau)\ne0$. The upper and lower bounds agree wherever $f_n\ne0$,
forcing $\operatorname{End}_{B_n}(\mathcal R_n)=\C$. Since $\mathcal R_n$ is
unitarizable on $\mathcal T_n$, it is completely reducible, and a
one-dimensional commutant forces irreducibility. A non-zero Laurent
polynomial does not vanish on any non-empty open subset of the torus,
so $\{f_n\ne0\}\cap\mathcal T_n$ is dense and open.
\end{proof}

\begin{remark}
The argument does not cover $n\ge5$, the whole of $\mathcal T_n$, or
the specialization $q=2$, $\tau=3$ itself (where complete reducibility
is not available). The same specialization technique is used for
other Long--Moody families in \cite[arXiv version, Theorem~5.8, Lemma~5.9]{ST}. The identification with Long's reduced construction is established in
Proposition~\ref{prop:long-reduced}. The present generic argument does not
locate the entire exceptional set.
\end{remark}

The next section uses the generator spectra to obtain properties
independent of a chosen unitary realization. Section~\ref{sec:synthesis}
then uses the matrix structure in one common coordinate realization
to construct elementary factorizations of all generators.

\Needspace{20\baselineskip}
\section{Spectral consequences: projective order, Clifford frames, and entanglement}\label{sec:gates}

We derive quantum-gate properties from the spectrum computed in
Section~\ref{sec:spectrum}, for $n\ge2$ and $(a,c)\in\mathcal T_n$.
The general notions of entanglement, Weyl--Heisenberg frames, Clifford
groups, and projective order are recalled in
Appendix~\ref{app:quantum-background}; Table~\ref{tab:quantum-dictionary}
summarizes them. The results below apply these notions and the
product-vector preserver theorem \cite[Thm.~4.1]{Johnston} and the
universality theorems cited below to this representation family.
Here $V=\C^{n+1}$, $K=0$, and
$U_0=V^{\oplus n}$ has dimension $M:=n(n+1)$.
The construction supplies the vector-space identification
\[
 V^{\oplus n}\longrightarrow\C^n\otimes V,
 \qquad (v_1,\ldots,v_n)\longmapsto\sum_{j=1}^n e_j\otimes v_j,
\]
where $e_1,\ldots,e_n$ is the standard basis of $\C^n$; the first
factor records the $n$ slots.
A chosen positive invariant inner product on $U_0$ need not be a product
inner product under this identification. After unitarization, we
consider any unitary identification
\[
 \varphi:\C^M\longrightarrow\C^{d_1}\otimes\C^{d_2},
 \qquad d_1d_2=M,\quad d_1,d_2\ge2,
\]
and write $U_\varphi=\varphi U\varphi^{-1}$ for $U\in U(M)$.
We call such a $\varphi$ a (nontrivial) \emph{bipartition}.
Theorems~\ref{thm:nonclifford} and~\ref{thm:entangling} range over all
Weyl--Heisenberg frames and all nontrivial unitary bipartite
identifications, respectively, making their conclusions invariant
under the common unitary conjugation of Theorem~\ref{thm:gauge}.

\begin{corollary}[The spectrum determines the seed]\label{cor:injective}
For every $n\ge2$ and $(a,c)\in\mathcal T_n$, the eigenvalue multiset
of one generator $U_i$ determines the ordered pair $(q,\tau)$, hence
$(a,c)$.

Consequently distinct seeds in $\mathcal T_n$ give pairwise
non-conjugate unitary braid representations, while by
Theorem~\ref{thm:gauge} all choices of $l\in\Wdef$ for one seed give conjugate
ones.
\end{corollary}
\begin{proof}
Since $0<a,c<d<\tfrac12$, the pair $\{\omega,-\omega\}$ is the unique antipodal
pair in the spectrum: the other non-trivial eigenangles are
$a+\tfrac12$ and $c+\tfrac12$, neither antipodal to $\omega$ (as $d>a,c$)
nor to each other (their difference is less than $\tfrac12$ in
absolute value), and $-1$ is not an eigenvalue. Removing $\{\omega,-\omega\}$
and the multiplicity $n(n-1)$ of $1$ leaves the polynomial
$(z+q)^{n-1}(z+\tau)^{n-1}$, which determines the unordered pair
$\{q,\tau\}$, also when $q=\tau$. With $p:=q\tau$ and $v:=\omega^2=q^3\tau^2$
one then has $q=v/p^2$ and $\tau=p/q$. The argument uses neither
$n-1>1$ nor $a\ne c$.
\end{proof}

The projective finiteness of the Clifford normalizers, proved after
Definition~\ref{def:wh-frame}, gives the following spectral criterion.

\begin{theorem}[Projective order]\label{thm:nonclifford}
Let $n\ge2$ and $(a,c)\in\mathcal T_n$. A unitary generator $U_i$ on
$\C^{n(n+1)}$ has finite projective order if and only if
$a,c\in\Q$.

If $a\notin\Q$ or $c\notin\Q$, then for every finite-dimensional
Weyl--Heisenberg frame on $\C^M$, $U_i$ does not normalize its finite
Weyl--Heisenberg subgroup, hence lies in no associated generalized
Clifford group. More generally, it lies in no subgroup of $U(M)$
that is finite modulo phases.

For $n=3$, $a=c=\tfrac1{20}$, all eigenangles are multiples of
$\tfrac1{40}$ and $U_i^{40}=I$.
\end{theorem}
\begin{proof}
Since $1$ is an eigenvalue of $U_i$, any scalar power
$U_i^k=zI$ must have $z=1$. Thus ordinary and projective order
coincide for these full-space generators.

The spectrum contains $1$, $-q$ and $-\tau$
(Proposition~\ref{prop:charpoly}); the ratios $-q=\ee{a+\frac12}$ and
$-\tau=\ee{c+\frac12}$ are roots of unity only if $a,c\in\Q$, so
irrational $a$ or $c$ excludes a scalar power. Conversely, for
$a,c\in\Q$ all eigenangles $0,a+\frac12,c+\frac12,d,d+\frac12$ are
rational and $U_i$, being diagonalizable, has finite order.
The non-Clifford conclusion follows from the projective finiteness
of the normalizer proved after Definition~\ref{def:wh-frame}.
\end{proof}

\begin{theorem}[Not a product operator for any bipartition]
\label{thm:entangling}
Let $n\ge2$ and $(a,c)\in\mathcal T_n$ with $1,a,c$ linearly
independent over $\Q$. Then for every unitary identification
$\varphi:\C^{n(n+1)}\to\C^{d_1}\otimes\C^{d_2}$ with
$d_1,d_2\ge2$ and every $i$, the operator
$U_{i,\varphi}:=\varphi U_i\varphi^{-1}$ is not of the form
$X\otimes Y$.

Consequently, for each $\varphi$ and $i$ there
exist nonzero $x\in\C^{d_1}$ and $y\in\C^{d_2}$ such that
$U_{i,\varphi}(x\otimes y)$ is not decomposable.
\end{theorem}
\begin{proof}
\emph{Reduction to a tensor-product spectrum.}
Fix $\varphi$ and write $U_i$ for its matrix in these tensor
coordinates throughout the proof.
A unitary mapping every product vector to a product vector is a
product operator, or a product operator composed with the swap when
$d_1=d_2$ \cite[Thm.~4.1]{Johnston}; for equal factor dimensions,
see also \cite[arXiv version, Theorem~1.4]{BrylinskiBrylinski}.
The swap alternative would require
the total dimension $n(n+1)=d_1^2$ to be a square, which is excluded by
$n^2<n(n+1)<(n+1)^2$. So it suffices to exclude $U_i=X\otimes Y$, and
by rescaling the factors reciprocally we may take $X,Y$ unitary.

Let $\Xi$, $\mathrm H$ be the sets of eigenangles of $X$, $Y$; the set of
eigenangles of $U_i$ is $\Xi+\mathrm H$ (mod $1$) with multiplicities
$m_U(s)=\sum_{\xi+\eta=s}m_X(\xi)m_Y(\eta)$.

\emph{Case 1: one factor has a single eigenvalue.}
If $\Xi$ or $\mathrm H$ is a singleton, every $m_U(s)$ is a multiple
of $d_1$ or $d_2$, contradicting the simple eigenvalues $\pm\omega$
(Proposition~\ref{prop:charpoly}).

\emph{Case 2: both factors have at least two eigenvalues.}
Pick $\xi\ne\xi'$,
$\eta\ne\eta'$; the four sums satisfy
$(\xi+\eta)+(\xi'+\eta')\equiv(\xi+\eta')+(\xi'+\eta)$ with
$\{\xi+\eta,\xi'+\eta'\}\ne\{\xi+\eta',\xi'+\eta\}$ as multisets, a
relation $\vartheta_i+\vartheta_j\equiv \vartheta_k+\vartheta_l\pmod\Z$ between unordered pairs of the
five affine forms $\vartheta_0=0$, $\vartheta_1=a+\tfrac12$, $\vartheta_2=c+\tfrac12$,
$\vartheta_3=d$, $\vartheta_4=d+\tfrac12$. Since $1,a,c$ are $\Q$-linearly independent
such a relation must hold identically in $(a,c)$; among the fifteen
unordered pair sums, equal $(a,c)$-coefficients occur only in the
groups $\{\vartheta_0+\vartheta_3,\vartheta_0+\vartheta_4\}$, $\{\vartheta_1+\vartheta_3,\vartheta_1+\vartheta_4\}$,
$\{\vartheta_2+\vartheta_3,\vartheta_2+\vartheta_4\}$, $\{2\vartheta_3,\vartheta_3+\vartheta_4,2\vartheta_4\}$, and the only integral
difference is $2\vartheta_3-2\vartheta_4=-1$. Hence the relation is
$2\vartheta_3\equiv2\vartheta_4$. Up to exchanging rows or columns,
the four eigenangles must therefore be arranged as
\[
 \begin{pmatrix}
 \xi+\eta&\xi+\eta'\\
 \xi'+\eta&\xi'+\eta'
 \end{pmatrix}
 \equiv
 \begin{pmatrix}
 \vartheta_3&\vartheta_4\\
 \vartheta_4&\vartheta_3
 \end{pmatrix}\pmod1.
\]

\emph{Multiplicity contradiction.}
Each of $\vartheta_3,\vartheta_4$ is thus realized by two distinct
eigenvalue pairs of the factors. Both multiplicities are at least two,
contradicting simplicity of $\omega$ and $-\omega$.
\end{proof}

\begin{corollary}[Entangling non-Clifford unitary gates]
\label{cor:quantum-gates}
If $(a,c)\in\mathcal T_n$ and $1,a,c$ are linearly independent over
$\Q$, then every unitary generator $U_i$ is simultaneously of
infinite projective order, outside every generalized Clifford group
associated with a finite-dimensional Weyl--Heisenberg frame, and
entangling with respect to every nontrivial bipartite tensor
identification of the ambient Hilbert space.

In particular, for $n=2$, after any chosen identification
$\C^6\cong\C^2\otimes\C^3$, the construction gives a unitary
entangling non-Clifford qubit--qutrit gate.
\end{corollary}
\begin{proof}
The $\Q$-linear independence implies that at least one of $a,c$ is
irrational, so Theorem~\ref{thm:nonclifford} applies; the entangling
claim is Theorem~\ref{thm:entangling}.
\end{proof}

\begin{remark}[Parameter exceptions and scope]\label{rem:entangling-scope}
(a) In Theorem~\ref{thm:entangling}, the $\Q$-independence hypothesis
can be weakened to $(a,c)\in\mathcal T_n\setminus\mathcal E_n$.
Using the five affine forms from the proof, put
\[
 \mathcal I=\{(j,k):0\le j\le k\le4\},\qquad
 h_{(j,k)}=\vartheta_j+\vartheta_k,
\]
and set
\[
 \mathcal E_n=
 \bigcup_{\substack{p,p'\in\mathcal I,\ p\ne p'\\
             \{p,p'\}\ne\{(3,3),(4,4)\}}}
 \{(a,c)\in\mathcal T_n:h_p(a,c)-h_{p'}(a,c)\in\Z\}.
\]
Only the comparison between $(3,3)$ and $(4,4)$ is omitted; all
other comparisons involving either pair are retained. After empty
loci are discarded, this is a finite union of intersections of
$\mathcal T_n$ with affine lines. This weakening concerns entanglement;
the projective-order and non-Clifford conclusions use
Theorem~\ref{thm:nonclifford}. For example,
$n=2$, $(a,c)=(1/100,3/100)$ lies outside $\mathcal E_2$ and its
generator has finite projective order $200$.

(b) The quantifier is one product input for each unitary bipartite
identification of the $n(n+1)$-dimensional output and each generator.
The input may depend on the identification and the generator.
Classical simulability is a separate question, discussed in
Section~\ref{sec:outlook}.

(c) Under the same parameter hypothesis as in
Theorem~\ref{thm:entangling}, Proposition~\ref{prop:decomp} shows that the same
spectral argument for entanglement applies to the summand $\mathcal R_n$
for $n\ge3$: its five eigenvalues have the same pair-sum relations,
$\pm\omega$ remain simple, and
$(n-1)^2<n(n-1)<n^2$ excludes the swap alternative.
Thus it is entangling for every nontrivial unitary bipartite
identification. For $n=2$, $\mathcal R_2$ has spectrum
$\{\omega,-\omega\}$ and projective order $2$ for every seed;
its prime dimension admits no nontrivial bipartition. The full-space
conclusion of Theorem~\ref{thm:nonclifford} does not extend to this summand.
\end{remark}

\begin{proposition}[Non-density of the full braid image]\label{prop:nondensity}
\label{par:nonuniversality}
For any fixed unitary realization, the braid image on the full
space $U_0$ is not projectively dense in $PU(U_0)$, despite
Corollary~\ref{cor:quantum-gates}.
Taking an orthogonal
invariant decomposition $U_0=\bigoplus_\alpha W_\alpha$, its image
therefore lies in the stabilizer
\[
 \{T\in U(U_0):T(W_\alpha)=W_\alpha\text{ for every }\alpha\}
 \cong\prod_\alpha U(W_\alpha).
\]
This is a proper closed subgroup containing all scalar unitaries,
so its projective image is also a proper closed subgroup.
Thus braid words alone cannot approximate all full-space gates.
\end{proposition}
\begin{proof}
Every braid preserves the
proper summands in Proposition~\ref{prop:decomp}.
\end{proof}

Density on a selected summand $\mathcal R_n$, or universality after
additional operations have been supplied, is a different question.
Generic irreducibility of $\mathcal R_n$ in the stated ranks does not
settle density: a proper compact subgroup can also act irreducibly.

For equal local dimensions, Brylinski--Brylinski call a gate
\emph{primitive} when it preserves product vectors and prove that an
imprimitive gate, together with arbitrary one-qudit gates, is exactly
universal \cite[arXiv version, Theorems~1.3--1.4]{BrylinskiBrylinski}.
Harrow's corresponding result allows unequal factor dimensions
\cite[arXiv version, Corollary~2]{HarrowExact}. In the coordinates above, put
\[
 K_\varphi
 =\{\varphi^{-1}(X\otimes Y)\varphi:X\in U(d_1),\ Y\in U(d_2)\}.
\]
For $U$ entangling with respect to $\varphi$, these results give
$\langle U,K_\varphi\rangle=U(M)$ when arbitrary operations in
$K_\varphi$ are independently available. This is exact universality
with the additional local gates (the elements of $K_\varphi$); it does not supply those
operations as braid words in the present representation.
Section~\ref{app:fixed-encoding} applies this distinction to a fixed
qubit--qutrit encoding.

\subsection{A finite Clifford reference example and exceptional parameters}
\label{subsec:clifford-benchmark}
The scalar specialization in \cite[Appendix~A.2.1]{FiniteCompanion}
provides a useful reference example for the distinction between ordinary and
projective gate images.
Take $n=3$, scalar input
$g_1=g_2=g_3=i$, $s_1=s_2=1$, and the exceptional parameter $\lambda=i$.
After passage to the two-dimensional quotient and a unitary change of
coordinates, the braid generators are the matrices $A,B$ below (these symbols are local to this
reference example):
\begin{equation}\label{eq:clifford-benchmark}
 A=\frac12\begin{pmatrix}1-i&1+i\\1+i&1-i\end{pmatrix},\qquad
 B=\begin{pmatrix}1&0\\0&-i\end{pmatrix}.
\end{equation}
This is the classical finite Burau specialization of
\cite{FunarKohno}; its ordinary braid image has order $96$.
With $H_{\rm Had}$ and $S_{\rm ph}$ as in
Appendix~\ref{app:gates-basics}, one has $A=H_{\rm Had}S_{\rm ph}^{\dagger}H_{\rm Had}$ and
$B=S_{\rm ph}^{\dagger}$.  Hence the projective image is the one-qubit
projective Clifford group, isomorphic to $\mathfrak S_4$ and of order
$24$; the ordinary image has scalar center $\mu_4$.  The same matrices occur, with the inverse convention, in Ising-anyon
braiding \cite[arXiv version, Eq.~(12)]{GeorgievIsing}; see also \cite{AGW}.  Thus this example is a finite Clifford
reference example, not a universality example.

For the same fixed input, the generic braid representation is
$\mathbf1\oplus\varphi$, where $\varphi:B_3\to U(2)$ is defined by
$\varphi(\sigma_1)=A$ and $\varphi(\sigma_2)=B$, whereas at
$\lambda=i$ it is $\varphi$.  Both ordinary images have order $96$, but
their projective images have orders $96$ and $24$, respectively
\cite[Propositions~6.3 and~6.7, Appendix~A.2.1]{FiniteCompanion}. Indeed, the map
$g\mapsto[\operatorname{diag}(1,g)]$ is injective on the ordinary group:
if $\operatorname{diag}(1,g)$ is scalar, its first entry forces that
scalar to be $1$.  A phase that is global on the
two-dimensional quotient at an exceptional parameter becomes a relative phase after the
trivial summand $\mathbf1$ is restored.

Thus exceptional parameters can change the projective gate image even
when the ordinary image order is unchanged. The kernel formula at
exceptional parameters is proved in
\cite[Proposition~6.7]{FiniteCompanion}; the equivalence of ordinary
and projective braid finiteness for finite full input is proved in
\cite[Proposition~6.9]{FiniteCompanion}. These results are not repeated
here. For this scalar input, the full $F_3\rtimes B_3$ image is finite
only at $\lambda=i$, although the ordinary braid image is finite for
every $\lambda\ne0,1$ \cite[Appendix~A.3.2]{FiniteCompanion}.
More generally, the parameter-independence of ordinary braid finiteness
in \cite[Theorem~6.6]{FiniteCompanion} assumes finite seed braid image;
it is not a finiteness assertion for arbitrary unitary input.

The phase convention matters at the level of the ordinary image.  For
example
\[
 ABA=e^{-\pi i/4}H_{\rm Had}.
\]
The Hadamard matrix itself is not in the displayed order-$96$ lift
(all its matrix entries lie in $\mathbb Q(i)$, whereas
$1/\sqrt2\notin\mathbb Q(i)$), although its projective class is
reachable.  After adjoining the trivial summand $\mathbf1$, this distinction is no
longer removable by a global phase.  This illustrates why the quantum
interpretation should keep ordinary and projective images separate.

\subsubsection{Open questions: irreducibility and image closure}
The exact decomposition leaves questions about irreducibility and
image closure. The following numerical observations concern these
questions and, separately, a different class of KLM seeds.

\begin{remark}[Exploratory numerical observations]
\label{rem:numerics}
(a) The decomposition of Proposition~\ref{prop:decomp} was first found
numerically for $(n,a,c)=(3,0.0713,0.1137)$, $(3,0.11,0.05)$ and
$(4,0.0521,0.1313)$: the commutant of the unitarized representation
has dimension $7$ in each case, the orthogonal decomposition into
summands of dimensions $1,1,n-1,n-1,n(n-1)$ is exact to $10^{-13}$,
and the commutant of $\mathcal R_n$ is one-dimensional (a
floating-point computation, reproduced by the ancillary script;
Proposition~\ref{prop:generic-irred} is the exact statement). The
closure of the image of $\mathcal R_n$ is not studied here; note that
its determinant is $-q^{n+1}\tau^n$ on each generator, so for rational
$(a,c)$ the determinant image is finite.

(b) We also made an exploratory
test with the standard Jones--Temperley--Lieb path-model representations
\cite{AJL,JonesHecke,Wenzl}.  For $q=\ee{1/\ell}$, let $\mathcal P_{m,\ell}$
be the Hilbert space with orthonormal basis the length-$m$ paths on the
$A_{\ell-1}$ graph starting at its left endpoint; the usual path-model
Temperley--Lieb operators $E_i$ satisfy
\[
 E_i^2=\delta_\ell E_i,\qquad
 \delta_\ell=2\cos(\pi/\ell).
\]
The implementation uses the braid normalization
\[
 \rho(\sigma_i)=\alpha_\ell I+\beta_\ell E_i,\qquad
 \alpha_\ell=i e^{-\pi i/(2\ell)},\qquad
 \beta_\ell=\alpha_\ell^{-1},
\]
with no additional scalar twist; thus $q=\alpha_\ell^{-4}$.
These matrices give a unitary representation of $B_m$ on this space.
The terminal vertex of a path is preserved, so
$\mathcal P_{m,\ell}$ is the orthogonal direct sum of its
\emph{endpoint sectors}; these are the sectors used below.  This is the
standard path realization of the Jones--Wenzl sectors, equivalently the
$SU(2)$ quantum-group/modular-functor sectors \cite{FLW}, with the
phase convention fixed above. We restrict each endpoint sector of
the $B_{n+1}$ representation to
$F_n\rtimes B_n\cong B_{1,n}$ as in Section~\ref{subsec:klm}, and use
that restriction as the KLM seed.

The ancillary script checks every endpoint sector of dimension at least
two for $n=3$, $\ell\in\{4,5,6,7,8,10,12,16\}$, and for $n=4$,
$\ell\in\{5,7\}$. It tests the ambient LM braid matrices
\eqref{eq:braid-matrices} only in sectors for which its numerical kernel
test gives $K=0$; otherwise it reports the kernel dimension and skips the
sector, without forming a quotient. Finite-precision rank tests
suggest the presence of non-trivial Jordan blocks, and sampled random
words of length $300$ have large operator norms, of orders
$10^{3}$--$10^{15}$. These finite-precision observations motivate
an exact Jordan obstruction; by themselves they establish neither
non-semisimplicity nor unboundedness.
\end{remark}

\begin{samepage}
\begin{problem}
Determine the remaining irreducibility locus of $\mathcal R_n(q,\tau)$,
within $\mathcal T_n$. Determine the closure
of the image on its irreducible sectors, and in the generic irreducible
case whether it contains $SU(n(n-1))$. Also prove the Jordan obstruction
for the Jones--Temperley--Lieb endpoint-sector seeds in
Remark~\ref{rem:numerics}(b).
\end{problem}
\end{samepage}


\subsection{The case \texorpdfstring{$n=2$}{n=2} in a spectral encoding: controlled gates}
\label{app:fixed-encoding}
This subsection gives exact circuit identities for $n=2$ in one
spectral encoding, with local gates supplied independently of
braid words. Since $B_2$ is cyclic, the construction diagonalizes a
single generator; Section~\ref{sec:synthesis} uses a different common
encoding for all generators.

\subsubsection{Spectral encoding of \texorpdfstring{$\C^6$}{C6} and local phase factors}
Fix $(a,c)\in\mathcal T_2$ and one unitary generator $U_1$.
Let $e_0,e_1$ and $f_0,f_1,f_2$ be the standard bases of $\C^2$
and $\C^3$, respectively, and order the product basis as
\[
 e_0\otimes f_0,\ e_0\otimes f_1,\ e_0\otimes f_2,\
 e_1\otimes f_0,\ e_1\otimes f_1,\ e_1\otimes f_2.
\]
Choose an orthonormal eigenbasis of $U_1$ with eigenvalues ordered
as $1,1,-q,\omega,-\omega,-\tau$, using
Proposition~\ref{prop:charpoly}. Map the displayed product basis to
that eigenbasis. This defines a unitary map
$\mathcal J:\C^2\otimes\C^3\to U_0$, which is kept fixed for the
chosen seed and realization. We call $\mathcal J$ the
\emph{spectral encoding} of $U_1$. Thus the identification used at the
start of this section is $\varphi=\mathcal J^\dagger$.
In these coordinates,
\begin{equation}\label{eq:encoded-D}
 D_{a,c}:=\mathcal J^\dagger U_1\mathcal J
 =\operatorname{diag}(1,1,-q,\omega,-\omega,-\tau).
\end{equation}
This eigenvalue ordering is part of the encoding. A different global
permutation of eigenvectors need not be local. In particular, the new
factors need not coincide with the raw slot--seed factors in the
identification $V^{\oplus n}\cong\C^n\otimes V$ at the start of this section.
As the parameters vary, the displayed diagonal matrices live on one
fixed model space, although their comparison maps $\mathcal J$ to
specified KLM realizations may vary.

\begin{proposition}[Local phase reduction]\label{prop:local-phase}
Put $P_b=e_be_b^\dagger$ for $b=0,1$, and set
\[
 L_2=\operatorname{diag}(1,\omega),\qquad
 L_3=\operatorname{diag}(1,1,-q),\qquad
 \eta_a=e^{-5\pi ia}=q^{-5/2},
\]
where $q^{1/2}=e^{\pi ia}$. Then
\begin{equation}\label{eq:encoded-local}
 D_{a,c}=(L_2\otimes L_3)C_a,\qquad
 C_a=P_0\otimes I_3+P_1\otimes V_a,\quad
 V_a=\operatorname{diag}(1,-1,\eta_a).
\end{equation}
Thus $C_a=\operatorname{diag}(1,1,1,1,-1,\eta_a)$, and $c$ can
be removed by local phase operations in this encoding.
\end{proposition}
\begin{proof}
The first five diagonal entries agree directly; the last is
$-\omega q\eta_a=-\tau$, since $\omega=q^{3/2}\tau$.
\end{proof}
The two summands $e_0\otimes\C^3$ and $e_1\otimes\C^3$ are
invariant, and $C_a$ restricts to $I_3$ and $V_a$, respectively.
Equivalently, it is block diagonal for the eigenspace decomposition
of $(P_0-P_1)\otimes I_3$. This is a controlled diagonal qutrit
unitary \cite[arXiv version, Section~IV.B and Fig.~8]{DiWei};
see also \cite{JiangControlled}.
The reduction is compatible with Corollary~\ref{cor:injective}:
conjugacy and equivalence under independent local multiplication on
the left and right are different equivalence relations. We remove
$c$ here; we do not claim a complete classification of the remaining
$a$ under unrestricted local equivalence.

\subsubsection{Exact encoded CZ, CNOT, and controlled phases}
We use the Schmidt coefficients and entanglement entropy of
Definition~\ref{def:schmidt-data}.
Let
\[
 \mathcal C_{01}=\C^2\otimes\operatorname{span}(f_0,f_1),\qquad
 \mathcal C_{02}=\C^2\otimes\operatorname{span}(f_0,f_2).
\]
These are four-dimensional invariant subspaces for both $C_a$ and
$D_{a,c}$, on which the restrictions give exact logical two-qubit gates.

\begin{proposition}[Encoded controlled gates]\label{prop:encoded-controlled}
In the respective ordered tensor bases,
\begin{equation}\label{eq:encoded-CZ}
 C_a|_{\mathcal C_{01}}=\operatorname{CZ}
  =\operatorname{diag}(1,1,1,-1),\qquad
 C_a|_{\mathcal C_{02}}=\operatorname{diag}(1,1,1,e^{-5\pi ia}).
\end{equation}
Let $H_{01}\in U(3)$ act as $H_{\mathrm{Had}}$ on
$\operatorname{span}(f_0,f_1)$ and fix $f_2$. Then
\[
 (I_2\otimes H_{01})C_a(I_2\otimes H_{01})
 =P_0\otimes I_3+P_1\otimes
 (f_0f_1^\dagger+f_1f_0^\dagger+\eta_a f_2f_2^\dagger).
\]
Its restriction to $\mathcal C_{01}$ is exactly CNOT and has no
leakage from that subspace. These identities hold throughout
$\mathcal T_2$, including rational parameters.
\end{proposition}
\begin{proof}
Select the relevant diagonal entries for the two restrictions.
Conjugating $\operatorname{diag}(1,-1)$ by $H_{\mathrm{Had}}$
gives $X$, while $f_2$ stays fixed.
\end{proof}
For example the product input
$(e_0+e_1)\otimes(f_0+f_1)/2$ is sent by $C_a$ to
\[
 \frac12\bigl(e_0\otimes(f_0+f_1)
                  +e_1\otimes(f_0-f_1)\bigr).
\]
Its two Schmidt coefficients are $1/\sqrt2$, so the output has
one bit of entanglement entropy, the maximum for a pure state of a
qubit--qutrit pair. The same input achieves this maximum for
$D_{a,c}$, whose remaining factor is local. This maximum is for the
declared encoding, not a uniform statement about entanglement entropy
in all encodings.

The second restriction in \eqref{eq:encoded-CZ} is a standard
controlled phase. With $c=1/20$, for example, $a=1/10$ gives
controlled-$S_{\mathrm{ph}}^\dagger$, and $a=1/20$ gives
controlled-$T_{\mathrm{ph}}^\dagger$, with $T_{\mathrm{ph}}$ as in
Appendix~\ref{app:clifford}.
Both seeds lie in $\mathcal T_2$ and realize these familiar logical
gates by specialization.
In the logical two-qubit Pauli frame (the Weyl--Heisenberg frame of
Definition~\ref{def:wh-frame} with $d_1=d_2=2$), a controlled phase with
angle $\phi$ is non-Clifford when $\phi\notin\pi\Z$: conjugation of
$X\otimes I$ has off-diagonal blocks
$\operatorname{diag}(1,e^{i\phi})$ and its adjoint, which cannot form
a Pauli tensor in that case. This verifies the two examples, but the
local corrections $L_2^\dagger,L_3^\dagger$ need not be Clifford;
no saving in non-Clifford resources is claimed. Nor is a
six-dimensional gate called a two-qubit gate without specifying the
four-dimensional subspace.

\subsubsection{An explicit word and its cost}
Define the level-selective controlled phase
\[
 \operatorname{CP}_j(\phi)
 =I_6+(e^{i\phi}-1)P_1\otimes f_jf_j^\dagger
 \qquad(j=1,2).
\]
Then the two factors commute and
\begin{equation}\label{eq:two-controlled-phases}
 C_a=\operatorname{CP}_1(\pi)\operatorname{CP}_2(-5\pi a).
\end{equation}
Thus two such operations and the local layer $L_2\otimes L_3$
synthesize $D_{a,c}$ in this gate library.

For a comparison in a controlled-two-level-swap library, let
$J_{0j}$ interchange $f_0,f_j$ and fix the third basis vector, and put
\[
 F_{0j}=P_0\otimes I_3+P_1\otimes J_{0j},\qquad
 Z_{0j}=f_0f_0^\dagger-f_jf_j^\dagger,\qquad
 R_{0j}(t)=\exp(itZ_{0j}).
\]
The $F_{0j}$ are generalized controlled-$X$ (GCX) gates
\cite[arXiv version, Section~II]{DiWei}. Since
$J_{0j}Z_{0j}J_{0j}=-Z_{0j}$, direct multiplication gives
\begin{equation}\label{eq:controlled-rotation}
 \begin{split}
 B_{0j}(t)&:=(I_2\otimes R_{0j}(t/2))F_{0j}
               (I_2\otimes R_{0j}(-t/2))F_{0j}\\
          &=P_0\otimes I_3+P_1\otimes R_{0j}(t).
 \end{split}
\end{equation}
Set $\phi_1=\pi$, $\phi_2=-5\pi a$, and
\[
 \gamma=(\phi_1+\phi_2)/3,\qquad
 \alpha=\gamma-\phi_1,\qquad \beta=\gamma-\phi_2.
\]
Since $V_a=e^{i\gamma}R_{01}(\alpha)R_{02}(\beta)$, we obtain
\begin{equation}\label{eq:four-GCX}
 C_a=(\operatorname{diag}(1,e^{i\gamma})\otimes I_3)
       B_{01}(\alpha)B_{02}(\beta).
\end{equation}
This uses four GCX gates, four local qutrit rotations and a local
control phase, before merging adjacent local gates. Together
with \eqref{eq:encoded-local} it is an exact circuit for $D_{a,c}$.
The four-GCX count agrees with the $N=3$ case of the general
$2(N-1)$ construction for controlled one-qudit unitaries
\cite[arXiv version, Section~IV.B]{DiWei}; the construction
restricts to $\C^2\otimes\C^N$ by retaining two control levels.

At $a=1/5$, the spectrum of $V_a$ is $(1,-1,-1)$ and one GCX
suffices. Explicitly set
\[
 W_{02}=\begin{pmatrix}
 1/\sqrt2&0&-1/\sqrt2\\0&1&0\\1/\sqrt2&0&1/\sqrt2
 \end{pmatrix}.
\]
Then $W_{02}J_{02}W_{02}^\dagger=\operatorname{diag}(-1,1,1)$, and hence
\[
 C_{1/5}=(\operatorname{diag}(1,-1)\otimes W_{02})
              F_{02}(I_2\otimes W_{02}^\dagger).
\]
This holds for any $0<c<1/5$, as required by $\mathcal T_2$.
The gate is entangling, so zero GCX gates plus local gates cannot
suffice. The one-GCX count is therefore optimal at this
specialization in precisely this resource model.

\subsubsection{Universality given independent local gates}
With $D_{a,c}$ in the declared encoding, the local phase corrections
and two local $H_{01}$ gates supply an exact logical CNOT on
$\mathcal C_{01}$. If arbitrary logical one-qubit operations are also
independently available, these gates are universal on that logical
pair by the standard CNOT universality theorem \cite{BarencoGates}.

More strongly, if all $U(2)\otimes U(3)$ operations are supplied
independently, $D_{a,c}$ is entangling for every
$(a,c)\in\mathcal T_2$ in this encoding, so it and those local
gates are exactly universal on $\C^6$ by
\cite[arXiv version, Corollary~2]{HarrowExact}. That theorem explicitly allows
unequal local dimensions. These consequences of independently supplied
local gates are distinct from braid-only nonuniversality in
Section~\ref{sec:gates}. The example connects a unitary realization obtained from a positive invariant form
to familiar logical gates through a verifiable circuit.

\section{Common-basis reflection synthesis}\label{sec:synthesis}

We use the standard reflection-length formula, Householder methods,
and state-preparation and controlled-phase decompositions
\cite{HLR,IKV,MIC,SBM,BarencoGates,GoubaultHouseholder}.
The specific step for this family is to construct one orthonormal
coordinate system in which \emph{every} braid generator has the small
blocks of Theorem~\ref{thm:syn-blocks}. Their supports give the
$O(M)$ elementary-word bound of Theorem~\ref{thm:syn-cost}; the
finite-alphabet bound then uses one-qubit approximation.

Throughout this section we use the \emph{slot-major order} on
$V^{\oplus n}\cong\C^n\otimes V$: the basis vectors $e_j\otimes f_k$
are listed slot by slot, with zero-based coordinate index
$(j-1)N+k$, where $1\le j\le n$, $0\le k\le n$, and $N=n+1$.
For the seed basis constructed in Section~\ref{subsec:syn-seed},
the blocks in Theorem~\ref{thm:syn-blocks} occupy contiguous
coordinate intervals in this order, as required by
Lemma~\ref{lem:syn-cholesky}.

\noindent\fbox{\begin{minipage}{\dimexpr\linewidth-2\fboxsep-2\fboxrule\relax}
\small
\textbf{Construction recipe in common orthonormal coordinates.}
Choose $n\ge2$, $a,c>0$ with $nc+(n+1)a<1$, and
$0<l<1-nc-(n+1)a$. Set $q=e^{2\pi ia}$, $\tau=e^{2\pi ic}$,
$\lambda=e^{2\pi il}$, $N=n+1$, and $M=nN$.
\begin{enumerate}
\setlength{\itemsep}{2pt}\setlength{\parskip}{0pt}
\item Form the upper bidiagonal $T_q$ in Lemma~\ref{lem:syn-seed}.
For the standard columns $E_j\in\C^n$ put $u_j=T_qE_j$ and
$\rho_q(\sigma_{j-1})=1\oplus[I_n-(1+q)u_ju_j^\dagger]$.
\item Put $s_i^{\mathrm{tB}}=\rho_q(\sigma_i)$, $w_1=f_1$,
$w_{j+1}=s_j^{\mathrm{tB}}w_j$ and
$g_j^{\mathrm{tB}}=\tau[I_N+(q^2-1)w_jw_j^\dagger]$. These are the unitary
input matrices in the fixed seed basis $f_0,\ldots,f_n$.
\item In slot-major order form $A_i$ from
\eqref{eq:braid-matrices}: its two-slot block is
$\left(\begin{smallmatrix}0&s_i^{\mathrm{tB}} g_i^{\mathrm{tB}}\\
s_i^{\mathrm{tB}}&s_i^{\mathrm{tB}}(I-g_{i+1}^{\mathrm{tB}})\end{smallmatrix}\right)$,
and it acts by $s_i^{\mathrm{tB}}$ on other slots. Form
$H_l=\widetilde H(\lambda)$ from \eqref{eq:blocks}, with the fixed
branch $\lambda^{1/2}=e^{\pi il}$.
\item Compute \emph{one} upper Cholesky factor $\mathsf{C}_l$ with
$\mathsf{C}_l^\dagger \mathsf{C}_l=H_l$ and set $U_i=\mathsf{C}_lA_i\mathsf{C}_l^{-1}$ for every $i$.
Theorem~\ref{thm:twisted-burau} guarantees $H_l>0$.
\end{enumerate}
Embed the $n$ slots and $n+1$ seed coordinates by the fixed isometry
$J_n$ in Section~\ref{subsec:syn-elementary}, and extend by the identity
on unused labels. The small blocks in this realization give a word
of $O(M)$ letters in the elementary generating set $\mathcal G_m$
(Definition~\ref{def:elementary-generators}) for each generator.
\end{minipage}}

Throughout this section fix $n\ge2$, $(a,c)\in\mathcal T_n$, $N=n+1$,
$M=nN$, and a parameter $0<l<1-nc-(n+1)a$. Theorem~\ref{thm:twisted-burau}
then gives a positive canonical KLM form on $U_0=V^{\oplus n}$.
In one common orthonormal coordinate system and a fixed encoding,
we express all braid generators as short elementary words, including
the preparation of reflection directions. This encoding differs from
that of Section~\ref{app:fixed-encoding}; Table~\ref{tab:circuit-dictionary}
gives the terminology.

\subsection{Reflection length}
We use the phase reflections $R(v,\theta)$ and the reflection lengths
$\ell_{\mathrm{ref}}$ and $\ell_{\mathrm{ref}}^{\mathrm{proj}}$ defined
in Definition~\ref{def:phase-reflections}. The following standard fact is recalled for completeness;
see \cite[arXiv version, Example~2.15]{HLR} for its unitary-group antecedent.

\begin{lemma}[Standard reflection length]\label{lem:syn-length}
For $U\in U(d_0)$,
\[
 \ell_{\mathrm{ref}}(U)=\operatorname{rank}(U-I),\qquad
 \ell_{\mathrm{ref}}^{\mathrm{proj}}(U)
 =d_0-\max_{\zeta\in\spec(U)}\dim\Ker(U-\zeta I).
\]
\end{lemma}
\begin{proof}
An orthonormal eigenbasis gives one phase reflection for every
nonidentity eigenvalue, counted with multiplicity. Conversely,
$AB-I=(A-I)B+(B-I)$ implies
$\operatorname{rank}(AB-I)\le\operatorname{rank}(A-I)+
\operatorname{rank}(B-I)$, so fewer factors cannot suffice.
Minimize $\operatorname{rank}(e^{-i\phi}U-I)$ over $\phi$ for the
projective assertion.
\end{proof}

\begin{corollary}[Reflection count for the unitary family]
\label{cor:syn-length}
Each full-space unitary braid generator satisfies
\[
 \ell_{\mathrm{ref}}(U_i)
 =\ell_{\mathrm{ref}}^{\mathrm{proj}}(U_i)=2n.
\]
\end{corollary}
\begin{proof}
By Proposition~\ref{prop:charpoly}, the identity eigenspace has dimension
$n(n-1)$ and is of maximum multiplicity. This remains true at $a=c$,
where the two multiplicities $n-1$ add (with a tie possible for $n=2$).
Apply Lemma~\ref{lem:syn-length}.
\end{proof}
The computed spectrum thus determines the phase-reflection count.
The next step is to prepare its directions in a \emph{common}
coordinate system and count the elementary factors. An arbitrary phase
reflection is not counted as one elementary factor.

\Needspace{10\baselineskip}
\subsection{Common coordinates for the Burau seed}\label{subsec:syn-seed}
In the raw Burau coordinates of Section~\ref{sec:twisted}, put
$h_j=-qe_{j-1}+e_j$ for $1\le j\le n$. These columns form a basis of
$V_{\mathrm{red}}$. Let $B_j^{\mathrm{root}}$ be the reduced matrix of
$\sigma_{j-1}$ in this basis. It is the identity except in row $j$,
whose adjacent entries are $1,-q,q$, with missing endpoint entries omitted.
For the standard column $E_j\in\C^n$ set
\[
 (J_q)_{jj}=1,\qquad
 (J_q)_{j,j+1}=-\frac{q}{1+q},\qquad
 (J_q)_{j+1,j}=-\frac1{1+q}.
\]
All other entries are zero, and direct calculation gives
\[
 B_j^{\mathrm{root}}=I-(1+q)E_jE_j^\dagger J_q,
 \qquad (B_j^{\mathrm{root}})^\dagger J_qB_j^{\mathrm{root}}=J_q.
\]
These formulas fix the normalization and coordinates of the Burau
invariant form \cite{Squier}.
Writing $\theta_q=\pi a$, a diagonal unitary makes $J_q$ real
tridiagonal with off-diagonal entries $-1/(2\cos\theta_q)$.
Its eigenvalues are
\[
 1-\frac{\cos(k\pi/(n+1))}{\cos\theta_q},\qquad 1\le k\le n,
\]
and are positive since $0<a<1/(n+1)$.

\begin{lemma}[Bidiagonal seed factor]\label{lem:syn-seed}
The upper Cholesky factor $T_q$, with $T_q^\dagger T_q=J_q$ and
positive diagonal, is bidiagonal. Its nonzero entries satisfy
\[
 t_j^2:=(T_q)_{jj}^2
  =\frac{\sin((j+1)\theta_q)}{2\cos\theta_q\sin(j\theta_q)},
 \qquad (T_q)_{j,j+1}=-\frac{e^{i\theta_q}}{2\cos\theta_q\,t_j}.
\]
Consequently
\[
 T_qB_j^{\mathrm{root}}T_q^{-1}=I-(1+q)u_ju_j^\dagger,
 \qquad u_j=T_qE_j,\quad \|u_j\|=1.
\]
The support of $u_j$ is contained in coordinates $j-1,j$, and that
of $u_1$ is contained in coordinate $1$.
\end{lemma}
\begin{proof}
Cholesky elimination of a positive tridiagonal matrix is bidiagonal.
The diagonal recurrence is $t_1=1$ and
$t_{j+1}^2=1-(4\cos^2\theta_q\,t_j^2)^{-1}$; the displayed sine
formula satisfies it. Use $J_qT_q^{-1}=T_q^\dagger$ to conjugate the
reflection formula, and $\|T_qE_j\|^2=(J_q)_{jj}=1$.
\end{proof}
Adjoin the trivial summand first and denote the orthonormal seed basis
by $f_0,f_1,\ldots,f_n$. Then $s_i^{\mathrm{tB}}=\rho_q(\sigma_i)$ is the
identity except on $\operatorname{span}(f_i,f_{i+1})$, for $1\le i<n$.
The free generators are given in the same coordinates by
\[
 g_1^{\mathrm{tB}}=\tau\rho_q(\sigma_0)^2,\qquad
 g_{j+1}^{\mathrm{tB}}=s_j^{\mathrm{tB}}g_j^{\mathrm{tB}}(s_j^{\mathrm{tB}})^\dagger.
\]
In particular,
\begin{equation}\label{eq:syn-free-generators}
 g_j^{\mathrm{tB}}=\tau[I+(q^2-1)w_jw_j^\dagger],\qquad
 w_1=f_1,\quad w_{j+1}=s_j^{\mathrm{tB}}w_j.
\end{equation}
No generator-dependent seed basis is used.

\subsection{Block structure after the common unitarization}
\begin{lemma}[Ordered block preservation]\label{lem:syn-cholesky}
Suppose $A$ is block diagonal for a partition into contiguous coordinate
intervals, $H>0$, and $A^\dagger HA=H$. If $\mathsf{C}$ is upper triangular
and $\mathsf{C}^\dagger \mathsf{C}=H$, then $\mathsf{C}A\mathsf{C}^{-1}$ is block diagonal for the same
partition, with diagonal blocks $\mathsf{C}_{jj}A_j\mathsf{C}_{jj}^{-1}$. Identity blocks
remain identity. The same $\mathsf{C}$ works for a family preserving $H$, even
when its members have different partitions into contiguous intervals.
\end{lemma}
\begin{proof}
Both $\mathsf{C}$ and $\mathsf{C}^{-1}$ are block upper triangular for every such partition.
Thus $\mathsf{C}A\mathsf{C}^{-1}$ is block upper triangular with the stated diagonal
blocks; it is also unitary. A block upper triangular unitary is block
diagonal, since its first coordinate block and its orthogonal complement
are invariant; induction finishes the proof. This does not assert that
$\mathsf{C}$ itself is sparse.
\end{proof}

Use slot-major coordinates
\[
 e_1\otimes f_0,\ldots,e_1\otimes f_n,
 e_2\otimes f_0,\ldots,e_n\otimes f_n.
\]
The two-slot block of \eqref{eq:braid-matrices} is contiguous on
slots $i,i+1$. In every other slot it gives $s_i^{\mathrm{tB}}$, hence one
contiguous two-coordinate block and singleton identity blocks.

\begin{theorem}[Common-basis small-block realization]\label{thm:syn-blocks}
Form the canonical positive KLM matrix $H_l=\Ht(e(l))$ in the seed
coordinates of Lemma~\ref{lem:syn-seed}, and let $\mathsf{C}_l$ be its upper
Cholesky factor. In the common unitary realization
$U_i=\mathsf{C}_lA_i\mathsf{C}_l^{-1}$, each generator consists of one $2(n+1)$-dimensional
\emph{core} coordinate block $C_i$, $n-2$ two-dimensional
\emph{spectator} blocks $D_{ij}$, and
singleton identity blocks. The possible active coordinates are
\[
 (\{i,i+1\}\times\{0,\ldots,n\})
 \ \cup\ 
 ((\{1,\ldots,n\}\setminus\{i,i+1\})\times\{i,i+1\}),
\]
whose number is $4n-2$. The block polynomials are
\begin{align*}
 \det(z-C_i)&=(z-1)^n(z+q)(z+\tau)^{n-1}(z^2-q^3\tau^2),\\
 \det(z-D_{ij})&=(z-1)(z+q).
\end{align*}
In particular at least $(n-1)(n-2)$ individual computational-basis
vectors are fixed by each generator.
\end{theorem}
\begin{proof}
Theorem~\ref{thm:twisted-burau} supplies $H_l>0$ and its invariance.
Apply Lemma~\ref{lem:syn-cholesky} to each interval partition just
described, always using the same $\mathsf{C}_l$. Each spectator block has
spectrum $\{1,-q\}$. Divide the polynomial of
Proposition~\ref{prop:charpoly} by these $n-2$ polynomials and by
$(z-1)^{(n-1)(n-2)}$ for the singleton blocks to get the core polynomial.
Counting the active coordinates proves the remaining assertions.
\end{proof}
The partitions depend on $i$, but the coordinates do not. This is
coordinate-block sparsity, not tensor localization by a fixed
Yang--Baxter $R$-matrix in the sense of \cite{RW}.

\begin{corollary}[Small supports for an optimal reflection product]
\label{cor:syn-directions}
The optimal $2n$ factors of Corollary~\ref{cor:syn-length} can be chosen
as $n+2$ directions supported in the core of size $2(n+1)$ and one
direction in each of the $n-2$ spectator blocks of size $2$.
\end{corollary}
\begin{proof}
Choose orthonormal eigenvectors separately in the coordinate blocks.
The core has $n+2$ nonidentity eigenvalues counted with multiplicity,
and each spectator has one. The phase reflections on these mutually
orthogonal directions multiply to $U_i$. The argument includes $a=c$.
\end{proof}

\subsection{Elementary words of length \texorpdfstring{$O(M)$}{O(M)}}
\label{subsec:syn-elementary}
\begin{definition}[Elementary generators and words]
\label{def:circuits}\label{def:elementary-generators}
For an integer $m\ge1$, on $\mathcal H_m=(\C^2)^{\otimes m}$ write
\[
 e_x=e_{x_1}\otimes\cdots\otimes e_{x_m},\qquad
 x=(x_1,\ldots,x_m)\in\mathbb F_2^m,\qquad
 [x]_2=\sum_{j=1}^m2^{m-j}x_j.
\]
We also index $e_x$ by the integer $[x]_2$: the first tensor factor
is the most significant bit. Write $e_{\mathbf0}=e_0^{\otimes m}$.
For $1\le j\le m$ and $j\ne k$, define
\[
 \iota_j(A)=I_2^{\otimes(j-1)}\otimes A\otimes I_2^{\otimes(m-j)},
 \quad A\in U(2),\qquad
 \operatorname{CNOT}_{jk}e_x=e_{x+x_j\epsilon_k},
\]
where $\epsilon_k$ is the $k$th coordinate vector and addition is
in $\mathbb F_2^m$. The elementary generating set is
\[
 \mathcal G_m=\bigcup_{j=1}^m\iota_j(U(2))
 \ \cup\ \{\operatorname{CNOT}_{jk}:1\le j,k\le m,\ j\ne k\}.
\]
A \emph{word} is a finite sequence $\mathfrak c=(g_1,\ldots,g_s)$
over the declared generating set, with evaluation and size
\[
\begin{gathered}
 \operatorname{ev}(\mathfrak c)=g_s\cdots g_1,\qquad
 |\mathfrak c|=s,\\
 \ell_{\mathcal G_m}(W)
 =\min\{|\mathfrak c|:\operatorname{ev}(\mathfrak c)=W,
 \ \mathfrak c\text{ is a word over }\mathcal G_m\}.
\end{gathered}
\]
The empty word evaluates to the identity. In quantum computing such
a word is called a \emph{circuit}, and its size is the \emph{gate count}.
Its \emph{depth} counts successive layers, allowing letters on disjoint
factors in one layer. The minimum word length $\ell_{\mathcal G_m}$
is distinct from the size of a particular word.
\end{definition}
The set $\mathcal G_m$ is infinite and closed under adjoints.
Allowing CNOT on every ordered pair is the all-to-all convention
\cite{BarencoGates,SBM}.

Put $p=\lceil\log_2 n\rceil$, $b=\lceil\log_2(n+1)\rceil$, and
$m=p+b$. With $E_1,\ldots,E_n$ the standard slot basis, fix
\begin{equation}\label{eq:syn-encoding}
\begin{gathered}
J_n:\C^M\longrightarrow\mathcal H_m,\qquad
J_n(E_{s+1}\otimes f_t)=e_{\operatorname{label}(s,t)},\qquad
\operatorname{label}(s,t)=2^bs+t,\\
0\le s<n,\quad 0\le t\le n,\qquad
\widehat U_i=J_nU_iJ_n^\dagger+(I-J_nJ_n^\dagger).
\end{gathered}
\end{equation}
The first $p$ factors encode the slot and the last $b$ the seed
coordinate (the slot-major order written in binary). The same $J_n$ is used for all generators and their words;
$M\le2^m<4M$. An ancilla-free word uses these $m$ factors alone;
padding labels are not additional working factors
(Definition~\ref{def:circuit-synthesis}).

For a unit vector $v\in\operatorname{im}J_n$, preparation supplies
$W_ve_{\mathbf0}=e^{i\beta}v$. The phase cancels in
\begin{equation}\label{eq:syn-phase-preparation}
 R(v,\theta)=W_v\,[I+(e^{i\theta}-1)
 e_{\mathbf0}e_{\mathbf0}^\dagger]W_v^\dagger,
\end{equation}
so preparation and its inverse implement the reflection
\cite[Lemmas~8--9]{MIC}.

\Needspace{15\baselineskip}
\begin{theorem}[Linear-in-dimension elementary-word bound]
\label{thm:syn-cost}
For every padded generator from Theorem~\ref{thm:syn-blocks}, there is
a word $\mathfrak c_i$ over $\mathcal G_m$ with
\[
\operatorname{ev}(\mathfrak c_i)=\widehat U_i,\qquad
|\mathfrak c_i|=O(n^2)=O(M).
\]
This word is exact and ancilla-free, including preparation and
unpreparation of the reflection directions. In particular,
$\ell_{\mathcal G_m}(\widehat U_i)=O(M)$.
A given braid word of length $L$ therefore yields an elementary word
of size $O(LM)$.
\end{theorem}
\begin{proof}
\emph{Core directions.}
A core direction uses the consecutive slot labels $x=i-1$ and $x+1=i$.
Identify them with their $p$-bit strings and put
$d=x\mathbin\oplus(x+1)$, so $d_p=1$. On slot labels define
\[
 T_x:z\longmapsto z\mathbin\oplus x,\qquad
 C=\prod_{\substack{1\le j<p\\d_j=1}}\operatorname{CNOT}_{p,j},
 \qquad P=C\circ T_x.
\]
The same symbols denote the induced permutation matrices, acting
trivially on seed factors. Translation uses at most $p$ X letters;
the commuting factors of $C$ number at most $p-1$.
Thus $P(x)=\mathbf0$, $P(x+1)=\epsilon_p$, and
\[
 Pv=e_0^{\otimes(p-1)}\otimes v_{\mathrm{core}},\qquad
 v_{\mathrm{core}}\in(\C^2)^{\otimes(b+1)}.
\]
Standard preparation $Q$ on those last $b+1$ factors uses
$O(2^{b+1})=O(n)$ letters \cite[\S~IV, Theorem~9]{SBM}.
The matrix $P^{-1}(I\otimes Q)$ therefore prepares $v$ with total
size $O(p+n)=O(n)$.

\emph{Spectator directions.}
For a spectator direction, translate its fixed
slot to zero and apply the same construction to its two consecutive
seed labels. One-factor preparation followed by the inverse
permutation then costs $O(m)$, including the support labels.

\emph{Controlled phase and total cost.}
Conjugating the control wires by X gates changes zero-controls to
one-controls. The all-zero phase is therefore an $(m-1)$-controlled
one-qubit phase and admits ancilla-free $O(m^2)$ synthesis \cite[arXiv version, Corollary~7.6]{BarencoGates}.
Including preparation and its inverse for the directions in
Corollary~\ref{cor:syn-directions} therefore costs
\[
 O\bigl((n+2)n+(n-2)m+2nm^2\bigr)
 =O(n^2+n\log^2n)=O(n^2).
\]
Each completed reflection fixes the orthogonal complement of its
direction, hence fixes the padding pointwise. Its preparation word
and individual elementary factors need not do so.
Since $\mathcal G_m$ is closed under adjoints, reversing a word and
taking adjoints of its letters gives the inverse braid generator with
the same size.
\end{proof}
The implementation in Section~\ref{subsec:resource-counts} uses
$m-2$ clean ancillary qubits, returned to zero;
Section~\ref{app:matched-benchmark} implements an ancilla-free version.

\paragraph{Compilation over a fixed finite gate alphabet.}
\label{par:syn-finite-alphabet}
Fix a finite set $\mathcal G_{\mathrm{fin}}\subset SU(2)$, closed under
adjoints and generating a dense subgroup of $SU(2)$. With placements
$\iota_j$ from Definition~\ref{def:elementary-generators}, put
\[
 \mathcal F_m=\bigcup_{j=1}^m\iota_j(\mathcal G_{\mathrm{fin}})
 \ \cup\ \{\operatorname{CNOT}_{jk}:1\le j,k\le m,\ j\ne k\}.
\]
The local alphabet is independent of $n$, the
seed, $l$, and accuracy, while its placements depend on $m$.
For two unitaries on the full padded register, put
\[
 d_{\mathrm{proj}}(V,W)
 :=\min_{\phi\in\R}\|V-e^{i\phi}W\|_{\mathrm{op}},
 \qquad
 \alpha_{\mathrm{SK}}:=\frac{\log 5}{\log(3/2)}=3.96936\ldots.
\]
The accuracy below concerns words over this finite alphabet,
not noise in their physical implementation.
Compressed descriptions of repeated words are not counted
as single letters.

\Needspace{24\baselineskip}
Applying the one-qubit Solovay--Kitaev theorem
\cite{DawsonNielsenSK} to the exact words above gives the following
approximation bound.

\begin{samepage}
\begin{corollary}[Finite-alphabet synthesis with prescribed accuracy]
\label{cor:syn-finite-precision}
Under the hypotheses of Theorem~\ref{thm:syn-blocks}, for every
$0<\varepsilon<1$ the following hold.
\begin{enumerate}
\item[(i)] \textup{\textbf{One generator.}}
For every generator there is a word
$\widetilde{\mathfrak c}_i$ over $\mathcal F_m$ with matrix
$\widetilde U_i=\operatorname{ev}(\widetilde{\mathfrak c}_i)$ such that
\[
 \begin{split}
 d_{\mathrm{proj}}(\widehat U_i,\widetilde U_i)&\le\varepsilon,\\
 |\widetilde{\mathfrak c}_i|
 &=O\!\left(M[1+\log(M/\varepsilon)]^{\alpha_{\mathrm{SK}}}\right).
 \end{split}
\]
\item[(ii)] \textup{\textbf{A braid word.}}
More generally, a braid word $w$ of length $L\ge1$ has an approximating
word $\widetilde{\mathfrak c}(w)$ over $\mathcal F_m$, with matrix
$\widetilde U(w)=\operatorname{ev}(\widetilde{\mathfrak c}(w))$, satisfying
\[
 \begin{split}
 d_{\mathrm{proj}}(\widehat U(w),\widetilde U(w))&\le\varepsilon,\\
 |\widetilde{\mathfrak c}(w)|
 &=O\!\left(LM[1+\log(LM/\varepsilon)]^{\alpha_{\mathrm{SK}}}\right).
 \end{split}
\]
\item[(iii)] \textup{\textbf{Ancillas, CNOT counts, and constants.}}
Both constructions are ancilla-free on the same $m$ tensor factors.
The CNOT count remains $O(M)$ per generator and $O(LM)$ per word.

The hidden constants may depend on
$\mathcal G_{\mathrm{fin}}$, but are independent of $n,L,\varepsilon$
and of the admissible seed and parameter $l$.
\end{enumerate}
\end{corollary}
\end{samepage}
\begin{proof}
\emph{One generator.}
Theorem~\ref{thm:syn-cost} gives an exact word over $\mathcal G_m$ with
$N_{\mathrm{gate}}=O(M)$ letters. For each one-qubit factor $V\in U(2)$ choose
$V=e^{i\beta}V_0$ with $V_0\in SU(2)$. By the one-qubit
Solovay--Kitaev theorem in \cite[arXiv version, Sections~2--3, especially Eq.~(8)]{DawsonNielsenSK},
$V_0$ admits an approximation from $\mathcal G_{\mathrm{fin}}$ with
operator-norm error at most $\delta$ and word length
$O([1+\log(1/\delta)]^{\alpha_{\mathrm{SK}}})$. Tensoring with identities
preserves the norm error, and
$\iota_j(V)=e^{i\beta}\iota_j(V_0)$ shows that the removed scalar is
a global phase on the whole register. CNOT factors are left unchanged.
In particular, any
controlled operation has already been decomposed into one-qubit gates
and CNOT before scalar phases are removed.

For unitary factors $B_j,\widetilde B_j$ and phases $\phi_j$ satisfying
$\|B_j-e^{i\phi_j}\widetilde B_j\|_{\mathrm{op}}\le\delta_j$, the
usual telescoping expansion gives
\[
 \left\|B_{N_{\mathrm{gate}}}\cdots B_1
       -e^{i\sum_j\phi_j}\widetilde B_{N_{\mathrm{gate}}}\cdots\widetilde B_1
 \right\|_{\mathrm{op}}
 \le \sum_{j=1}^{N_{\mathrm{gate}}}\delta_j.
\]
Choosing $\delta=\varepsilon/N_{\mathrm{gate}}$ for the one-qubit factors proves the
generator estimate. This uses Solovay--Kitaev only in the fixed group
$SU(2)$, not in a group whose dimension grows with $M$.

\emph{A braid word.}
For a word of length $L$, compile each required generator to error
$\eta=\varepsilon/L$, and use the adjoint word for its inverse;
inverse-closure preserves the gate alphabet and the error bound.
Use the same compiled generator for repeated occurrences. The same
telescoping inequality bounds the word error by $L\eta$, and multiplying
the per-generator size by $L$ gives the asserted bound. No ancillary
qubits or new CNOT gates are introduced by the one-qubit replacements.
\end{proof}

\paragraph{Improved exponents.}
The standard exponent $\alpha_{\mathrm{SK}}$ can be improved using
stronger one-qubit approximation results. If a fixed one-qubit alphabet
has worst-case approximation length $F_{\mathcal G}(\delta)$, the same
proof gives size $O(N_{\mathrm{gate}}[1+F_{\mathcal G}(\varepsilon/N_{\mathrm{gate}})])$, where
$N_{\mathrm{gate}}=O(M)$ or $O(LM)$. In particular, Kuperberg's improved
Solovay--Kitaev theorem \cite[Theorem~1.1]{KuperbergSK} allows
$\alpha_{\mathrm{SK}}$ in Corollary~\ref{cor:syn-finite-precision} to
be replaced by any fixed
\[
 \alpha>\frac{\log 2}{\log((1+\sqrt5)/2)}=1.44042\ldots.
\]
In this version the constants may also depend on $\alpha$.

\paragraph{What the bound covers.}
The distance is projective on the \emph{entire} padded register: phases
relative to its padding are included, while one overall phase is not.
Thus a controlled use of the whole compiled operator requires a separate
relative-phase convention and is not guaranteed by this corollary.
The approximants need not obey the braid relations exactly, preserve the
exact reflection length, or fix the padding pointwise. If $P$ is the
orthogonal projection onto the valid labels, then
\[
 \|(I-P)\widetilde U(w)P\|_{\mathrm{op}}\le\varepsilon,
\]
since $\widehat U(w)$ preserves those labels; the resulting leakage
probability is at most $\varepsilon^2$ for a unit input in their span.

\paragraph{Classical preprocessing cost.}
For the chosen seed and $l$, constructing $H_l$, its Cholesky factor
(an $O(M^3)$ arithmetic procedure), block eigenvectors, and rotation
angles is classical preprocessing. Conversion from a preassigned
physical encoding is a separate operation. We bound neither the bit
complexity of obtaining these data nor conditioning uniformly near a
boundary. An effective finite-alphabet compiler requires target and
alphabet entries to certified precision: splitting each local error
budget between a certified unitary target approximation and its
synthesis preserves the proof up to constants.
Section~\ref{subsec:syn-error-margins} supplies conditional error budgets
and positivity margins; Section~\ref{app:syn-counts} measures words with
arbitrary-angle letters.

After preprocessing, the same blocks apply one generator to a
classically stored state vector in $O((2N)^2+4(n-2))=O(M)$ arithmetic
operations. A word of length $L$ thus takes $O(LM)$ operations and
$O(M)$ state-vector memory.

\paragraph{Comparison with prior synthesis results.}
The common-block construction above supplies the supports to which
we apply the established preparation and reflection methods.

General sparse-state bounds, including ancilla tradeoffs, are given in
\cite[Theorems~1, 2, 7]{LiLuo}. Its ancilla-free Theorem~1 uses the same
one-qubit/CNOT model as Theorem~\ref{thm:syn-cost}.
The Toffoli-cost bounds of \cite{RupprechtWoelk} use a different
cost measure. A \emph{Toffoli gate} is the
three-qubit gate that flips its target bit exactly when both control
bits equal $1$. For example, the general
ancilla-free preparation bound $O(ms/\log m+m)$ in
\cite[Theorem~1]{LiLuo} for $s$ nonzero amplitudes
would give $O(n\log n/\log\log n)$ at $s=O(n)$ here. Our two-slot
support permits $O(n)$ preparation without ancillas by an explicit
$O(m)$ affine compression. With more ancillary space, general sparse
methods can also attain this order. Here the common block structure
of the KLM family supplies the special support used by those primitives.

Likewise, sparse-unitary oracle algorithms \cite{JordanWocjan} and
quantum implementations of Jones path-model braids \cite{AJL} predate
this work. The former assume sparsity and entry/position access with
specified costs; our proof constructs explicit elementary words.
The bound $O(M)$ is exponential in
the number of data qubits and is not polynomial-time braid simulation
in the same sense as the latter. Relative to treating the whole matrix
as dense, the upper bound improves the quadratic-in-dimension order
of generic synthesis \cite{SBM,KrolAlArs}. These are upper-bound
comparisons; Section~\ref{app:matched-benchmark} reports a matched-resource
experiment against an independently implemented one-qubit-stopping
\emph{quantum Shannon decomposition} (QSD) baseline. QSD recursively
factors a multiqubit unitary into controlled rotations and smaller
unitaries \cite{SBM}; the variant used in that experiment is specified
there.

\subsection{Reflection error budgets and positivity margins}
\label{subsec:syn-error-margins}
The finite-alphabet bound starts with specified unitary targets. We now
separate perturbations of their reflection data from the numerical
conditioning of the common positive form using elementary norm
estimates.
All operator norms in this subsection are spectral norms.

\begin{lemma}[Phase and direction errors]\label{lem:syn-reflection-error}
For orthogonal projections $P,\widetilde P$ on the same space and real
$\theta,\widetilde\theta$, put
$R(P,\theta)=I+(e^{i\theta}-1)P$. Then
\[
 \|R(P,\theta)-R(\widetilde P,\widetilde\theta)\|_{\mathrm{op}}
 \le 2\|P-\widetilde P\|_{\mathrm{op}}
       +|e^{i\theta}-e^{i\widetilde\theta}|.
\]
For unit vectors $v,\widetilde v$ and their rank-one projections,
\[
 \|vv^\dagger-\widetilde v\widetilde v^\dagger\|_{\mathrm{op}}
   =\sqrt{1-|v^\dagger\widetilde v|^2}.
\]
Consequently products of $r$ such factors, with corresponding projection
errors at most $\eta_j$ and phase chord errors at most $\zeta_j$, differ
in operator norm by at most $\sum_{j=1}^r(2\eta_j+\zeta_j)$.
The approximate projections need not be mutually orthogonal.
\end{lemma}
\begin{proof}
Subtract the factors as
\[
 (e^{i\theta}-1)(P-\widetilde P)
 +(e^{i\theta}-e^{i\widetilde\theta})\widetilde P
\]
and use $|e^{i\theta}-1|\le2$ and $\|\widetilde P\|_{\mathrm{op}}\le1$.
The rank-one formula follows by restricting to the span of the two
vectors; the nonzero eigenvalues of the difference are
$\pm\sqrt{1-|v^\dagger\widetilde v|^2}$. Every perturbed factor is
unitary, so telescoping the product gives the last assertion.
\end{proof}

For the $2n$ factors in Corollary~\ref{cor:syn-directions}, uniform
budgets $\eta_j\le\eta$ and $\zeta_j\le\zeta$ therefore give an error
at most $2n(2\eta+\zeta)$ per generator. A phase error of at most
$\zeta$ radians, measured modulo $2\pi$, suffices for the chord bound.
If a word of length $L$ uses these perturbed products (and their exact
adjoints for inverse generators), and finite-alphabet compilation of
that word adds projective error at most $\varepsilon_{\mathrm{comp}}$,
then
\[
 d_{\mathrm{proj}}(\widehat U(w),\widetilde U(w))
 \le 2nL(2\eta+\zeta)+\varepsilon_{\mathrm{comp}}.
\]
For example, allocating $\eta\le\varepsilon/(16nL)$,
$\zeta\le\varepsilon/(8nL)$, and
$\varepsilon_{\mathrm{comp}}\le\varepsilon/2$ makes the total at most
$\varepsilon$. To retain the small-support synthesis bound, the
approximate unit directions must stay in their prescribed coordinate
blocks; the norm estimate itself does not require that restriction.
An independently measured base-factorization error must also be added
when numerical factors, rather than an exact spectral product, are used.

The projection formulation removes the arbitrary phase of an eigenvector.
At a repeated eigenvalue, one may instead apply the same estimate to its
whole spectral projection. Indeed $R(P,\theta)$ for a higher-rank
projection is the product of the rank-one phase factors in any orthonormal
basis of its range. No continuity or uniqueness of individual numerical
eigenvectors is assumed. The lemma does not, by itself, bound projection
errors from rounded seed data or matrix entries.

\begin{definition}[Fixed-coordinate positive-definiteness margin]\label{def:pd-margin}
Keep the unitary-seed, slot-major coordinates in
Theorem~\ref{thm:syn-blocks}, before the final conjugation by $\mathsf{C}_l$.
For $H=H_l>0$ define
\[
 h_{\min}=\lambda_{\min}(H),\quad h_{\max}=\lambda_{\max}(H),\quad
 \mu_{\mathrm{pos}}(H)=\frac{h_{\min}}{h_{\max}},\quad
 \kappa_2(H)=\frac1{\mu_{\mathrm{pos}}(H)}.
\]
Thus $\mu_{\mathrm{pos}}$ is the smallest eigenvalue after normalizing
$H$ to spectral norm one. It is invariant under positive scalar
rescaling, but not under arbitrary nonunitary coordinate changes.
\end{definition}
After conjugation the form is $I$; its condition number would not measure
the conditioning of the preprocessing problem. General perturbation
analysis for Hermitian Cholesky factors is developed in \cite{DOVCholesky};
we use only the following elementary estimates.

\begin{proposition}[Positive-form and residual margins]
\label{prop:syn-positivity-margin}
For $E=E^\dagger$ and $e=\|E\|_{\mathrm{op}}<h_{\min}$,
\[
 H+E>0,\qquad
 \lambda_{\min}(H+E)\ge h_{\min}-e,\qquad
 \kappa_2(H+E)\le\frac{h_{\max}+e}{h_{\min}-e}.
\]
In particular, a relative error $e/h_{\max}<\mu_{\mathrm{pos}}(H)$
preserves positivity. For an a posteriori version, let
$\widehat H=\widehat H^\dagger$, let $\widehat{\mathsf{C}}$ be nonsingular, and
suppose bounds are available with
\[
 \|H-\widehat H\|_{\mathrm{op}}\le e_{\mathrm{data}},\quad
 \|\widehat{\mathsf{C}}^\dagger\widehat{\mathsf{C}}-\widehat H\|_{\mathrm{op}}
       \le\rho_{\mathrm{fac}},\quad
 0<s_\star\le\sigma_{\min}(\widehat{\mathsf{C}}).
\]
Then
\[
 \lambda_{\min}(H)
 \ge s_\star^2-\rho_{\mathrm{fac}}-e_{\mathrm{data}}.
\]
A strictly positive right-hand side certifies positivity, even without
assuming it in advance. A Frobenius-norm upper bound may replace the
spectral-norm factor residual.
\end{proposition}
\begin{proof}
For every unit vector $x$, $|x^\dagger Ex|\le e$. Taking the minimum
and maximum Rayleigh quotients proves the first assertions. For the
last assertion, write $H=\widehat{\mathsf{C}}^\dagger\widehat{\mathsf{C}}+
(H-\widehat H)-(\widehat{\mathsf{C}}^\dagger\widehat{\mathsf{C}}-\widehat H)$ and apply
the same argument.
\end{proof}

A certified bound for exact transcendental input data requires
reliable enclosures for $e_{\mathrm{data}}$,
$\rho_{\mathrm{fac}}$, and $s_\star$, including rounding errors in their
evaluation; successful floating-point Cholesky alone does not supply
these enclosures. Section~\ref{app:syn-margins} reports numerical
diagnostics and higher-precision cross-checks.

\paragraph{Unitarity defects from factor and invariance residuals.}
For the displayed matrices, in exact arithmetic, put
\[
 E_{\mathrm{fac}}=\widehat{\mathsf{C}}^\dagger\widehat{\mathsf{C}}-\widehat H,\quad
 E_{\mathrm{inv}}=\widehat A^\dagger\widehat H\widehat A-\widehat H,
 \qquad V=\widehat{\mathsf{C}}\widehat A\widehat{\mathsf{C}}^{-1}.
\]
Then
\[
 V^\dagger V-I=\widehat{\mathsf{C}}^{-\dagger}
 (E_{\mathrm{inv}}+\widehat A^\dagger E_{\mathrm{fac}}\widehat A
                  -E_{\mathrm{fac}})\widehat{\mathsf{C}}^{-1},
\]
so
\[
 \|V^\dagger V-I\|_{\mathrm{op}}
 \le\frac{\|E_{\mathrm{inv}}\|_{\mathrm{op}}+
 (\|\widehat A\|_{\mathrm{op}}^2+1)\|E_{\mathrm{fac}}\|_{\mathrm{op}}}
 {\sigma_{\min}(\widehat{\mathsf{C}})^2}.
\]
Thus the inverse margin matters in addition to the residuals. Forming
$V$ numerically introduces further triangular-solve and multiplication
errors. This estimate measures a unitarity defect, not distance to the
intended generator, and independent polar corrections would not ensure
the braid relations.

\begin{proposition}[Endpoint rates for a fixed seed]
\label{prop:syn-endpoint-rates}
Fix $n\ge2$, $(a,c)\in\mathcal T_n$ and the unitary-seed coordinates
used above, and put $l_{\mathrm f}=l_{\mathrm{first}}$ as defined in
Theorem~\ref{thm:twisted-burau}. Extend the family continuously to the
endpoints and write $H_0^{\mathrm{end}}:=\lim_{l\downarrow0}H_l$. Then
\[
 H_0^{\mathrm{end}}=D^\dagger D,\qquad
 \operatorname{rank}H_0^{\mathrm{end}}=n+1,\qquad
 \operatorname{rank}H_{l_{\mathrm f}}=M-1,
\]
and $H_0^{\mathrm{end}}+H_{l_{\mathrm f}}>0$. For $0<l<l_{\mathrm f}$,
\[
 H_l=
 \frac{\sin(\pi(l_{\mathrm f}-l))}{\sin(\pi l_{\mathrm f})}H_0^{\mathrm{end}}
 +\frac{\sin(\pi l)}{\sin(\pi l_{\mathrm f})}H_{l_{\mathrm f}}.
\]
Consequently, with constants allowed to depend on the fixed seed,
dimension and coordinates,
\begin{align*}
 \lambda_{\min}(H_l)&=\Theta(l),&
 \kappa_2(H_l)&=\Theta(l^{-1})&& (l\downarrow0),\\
 \lambda_{\min}(H_l)&=\Theta(l_{\mathrm f}-l),&
 \kappa_2(H_l)&=\Theta((l_{\mathrm f}-l)^{-1})
 &&(l\uparrow l_{\mathrm f}).
\end{align*}
The continuous positive margin $\mu_{\mathrm{pos}}(H_l)$ attains its
maximum at some point strictly inside this interval.
\end{proposition}
\begin{proof}
\emph{Interpolation.}
The endpoint ranks follow from surjectivity of $D$ and the
one-dimensional radical at the first wall. Eliminate $B$ from
Remark~\ref{rem:pencil} evaluated at $l$ and at $l_{\mathrm f}$ to
obtain the interpolation formula. Both coefficients are positive.
The endpoint forms are positive semidefinite; a common kernel vector
would therefore belong to $\Ker H_l$, contradicting $H_l>0$.
Thus $G_{\mathrm e}:=H_0^{\mathrm{end}}+H_{l_{\mathrm f}}>0$.

\emph{Endpoint bounds.}
Write the two interpolation coefficients as $\alpha_l,\beta_l$.
The lower bound
\[
 \lambda_{\min}(H_l)\ge
 \min(\alpha_l,\beta_l)\lambda_{\min}(G_{\mathrm e})
\]
holds throughout the interval. A unit vector in $\Ker H_0^{\mathrm{end}}$ gives
$\lambda_{\min}(H_l)\le\beta_l\|H_{l_{\mathrm f}}\|_{\mathrm{op}}$;
a unit vector in $\Ker H_{l_{\mathrm f}}$ gives
$\lambda_{\min}(H_l)\le\alpha_l\|H_0^{\mathrm{end}}\|_{\mathrm{op}}$.
Now $\beta_l=\Theta(l)$ at the left endpoint and
$\alpha_l=\Theta(l_{\mathrm f}-l)$ at the right endpoint, while
$\lambda_{\max}(H_l)$ tends to a positive limit at each endpoint.
These estimates prove the assertions on the eigenvalue and condition
number.

\emph{Interior maximizer.}
The margin extends continuously by zero at both endpoints
and is positive in the interior, so compactness gives an interior
maximizer.
\end{proof}

The two degenerations have different meanings: $l=0$ is the excluded
parameter $\lambda=1$, with kernel dimension $n^2-1$, whereas
$l=l_{\mathrm f}$ is a genuine spectral wall with a one-dimensional
kernel. These rates describe the conditioning of this positive-form
preprocessing, not loss of unitarity of the exact generators in the
open interval. Their simultaneous unitary-equivalence class remains
fixed by Theorem~\ref{thm:gauge}. Maximizing the margin is a possible
criterion for selecting $l$, not a unique choice or a guarantee
that the midpoint is optimal. The parameter $l$ and common encoding
must be fixed before compilation; no uniform-in-seed or uniform-in-$n$
conditioning bound is asserted.

\section{Numerical records: gate counts, conditioning, and a matched comparison}
\label{app:syn-counts}
This section evaluates words in one-qubit unitaries and CNOT gates,
using the conventions of Definition~\ref{def:elementary-generators}.
We count their CNOT letters; this is distinct from both the total number
of factors and the minimum length over all words representing the target
(Table~\ref{tab:circuit-dictionary}).

The two comparisons use different resource models.
Section~\ref{subsec:resource-counts} records an unoptimized implementation
on $m$ data qubits and $m-2$ clean ancillary qubits, returned to zero:
these are words over $\mathcal G_{2m-2}$. They are compared with published
ancilla-free upper bounds. Section~\ref{app:matched-benchmark} runs two
reference compilers on the same padded matrices in the same coordinates,
using words over $\mathcal G_m$ on the $m$ data qubits alone.
Section~\ref{app:syn-margins} measures conditioning of the positive form.
Exported angles, matrices, timings and error measurements are numerical
data; higher-precision checks are not validated interval enclosures.
We do not benchmark the finite-alphabet bound of
Corollary~\ref{cor:syn-finite-precision} or report a $T$-count
(Definition~\ref{def:finite-alphabet}). File inventories and reproduction
procedures are collected in Appendix~\ref{app:repro}.

\subsection{CNOT counts of the constructed words}
\label{subsec:resource-counts}
The one-qubit rotations used by the state-preparation routine are
\[
 R_y(\theta)=\begin{pmatrix}
 \cos(\theta/2)&-\sin(\theta/2)\\
 \sin(\theta/2)&\cos(\theta/2)
 \end{pmatrix},\qquad
 R_z(\theta)=\operatorname{diag}(e^{-i\theta/2},e^{i\theta/2}).
\]
A \emph{uniformly controlled} rotation applies a separate angle for
each computational-basis value of the control register; its matrix is
$\sum_r e_re_r^\dagger\otimes R_y(\theta_r)$, or the corresponding
$R_z$ expression. The \emph{magnitude tree} recursively uses these
$R_y$ rotations to distribute the squared amplitudes between binary
branches. The \emph{phase tree} uses controlled $R_z$ rotations to
set the relative phases. This magnitude--phase construction uses
the uniformly controlled rotations of
\cite{MottonenStatePreparation}; see also \cite[\S~IV, Theorem~9]{SBM}.
Our implementation retains the separate controlled-rotation circuits,
without the CNOT cancellations between the magnitude and phase
stages described in \cite{MottonenStatePreparation}.
On $k$ qubits it uses $2^{k+1}-4$ CNOTs, retaining all angles.
We also write \emph{CX} for CNOT. Define
\[
 h_s(i)=\operatorname{popcount}((i-1)\mathbin\oplus i)-1,\qquad
 h_b(i)=\operatorname{popcount}(i\mathbin\oplus(i+1))-1,
\]
where $\oplus$ is bitwise XOR and $\operatorname{popcount}$ counts
one bits. To implement an all-zero phase on $m$ bits, compute the
conjunction of $m-1$ controls with $m-2$ Toffoli gates, apply a
controlled phase using two CNOTs, and uncompute. Selecting zero
controls uses X gates. Six CNOTs per exact Toffoli gives
$12(m-2)+2$ CNOTs for this step. The resulting count for generator $i$
is
\[
\begin{split}
 N_{\mathrm{CX}}(i)={}&2(n+2)(2^{b+2}-4+h_s(i))\\
 &+2(n-2)h_b(i)+2n(12(m-2)+2).
\end{split}
\]
Both preparation and unpreparation, and all controlled phases, are
included. The exported circuit retains adjacent inverse pairs and
zero-angle rotations; its count is before routing to a restricted
connectivity architecture.

The following are exported CNOT counts, over all generators of each
listed rank. The last two columns are \emph{published generic upper
bounds}: the QSD formula of \cite[Eq.~(19), Table~I]{SBM} and the later
\emph{block-ZXZ} (BZXZ) bound of \cite[Section~V]{KrolAlArs}.
BZXZ is an alternative recursive decomposition into block factors and
controlled rotations. These columns evaluate the published formulas
at the register size, rather than compiling the present matrices.
For the same $m$ data bits they are respectively
\[
 \frac{23}{48}4^m-\frac32 2^m+\frac43,\qquad
 \frac{22}{48}4^m-\frac32 2^m+\frac53.
\]
\begin{center}\small
\begin{tabular}{@{}rrcrrr@{}}
\toprule
$n$ & $M$ & Data/anc. & Produced CX & QSD bound & BZXZ bound\\
\midrule
2 & 6 & 3/1 & 152 & 20 & 19\\
3 & 12 & 4/2 & 278--286 & 100 & 95\\
4 & 20 & 5/3 & 644--652 & 444 & 423\\
5 & 30 & 6/4 & 898--920 & 1868 & 1783\\
6 & 42 & 6/4 & 1056--1080 & 1868 & 1783\\
8 & 72 & 7/5 & 2204--2232 & 7660 & 7319\\
10 & 110 & 8/6 & 2936--2992 & 31020 & 29655\\
12 & 156 & 8/6 & 3476--3540 & 31020 & 29655\\
\bottomrule
\end{tabular}
\par\smallskip
\begin{minipage}{0.95\linewidth}\footnotesize
Data/anc.: numbers of data and clean ancillary qubits;
CX: CNOT; QSD: quantum Shannon decomposition;
BZXZ: block-ZXZ decomposition. Produced CX counts the CNOT factors
in the exported words. The QSD 2006 and BZXZ 2024 columns are generic
ancilla-free upper bounds, evaluated at the same number of data qubits.
The QSD bound includes the optimal three-CNOT two-qubit base case
and both additional optimizations in the Appendix of \cite{SBM}.
\end{minipage}
\end{center}
The generic bounds require no ancillas, whereas the exported counts
use the clean ancillas shown. The comparisons therefore do not
establish a matched-resource advantage.
For $n=2,3,4$ the unoptimized CNOT counts exceed both generic
bounds. Sparse-state methods \cite{LiLuo,RupprechtWoelk} are not
benchmarked by this table. The ancilla-free result of
\cite[Theorem~1]{LiLuo} uses the same one-qubit/CNOT model; the
family-specific distinction is the prescribed two-slot support of the
reflection directions. The bounds in \cite{RupprechtWoelk} use Toffoli
cost instead (Section~\ref{subsec:syn-elementary}).

\paragraph{Parameters and verification scope.}
The main cases use
\[
 a=\frac{\sqrt2}{16(n+1)},\qquad c=\frac{\sqrt3}{16n},\qquad
 l=\frac{1-nc-(n+1)a}{3}.
\]
The independent script checks all 42 generators for
$n=2,3,4,5,6,8,10,12$. It also checks the coincident-eigenvalue cases
$a=c$ at $n=2,3,4,6$, and verifies the identities for the reduced Burau Hermitian matrix and
reflections symbolically for $n=2,\ldots,6$.
The main-case maximum Frobenius residual for the predicted block
pattern is $5.19\times10^{-15}$, and for the reflection product it is $7.11\times10^{-15}$.
All six complete elementary circuits for $n\le4$ are simulated on
every input basis column, including unused padding labels. Their
maximum Frobenius error is $6.35\times10^{-14}$, and the norm of amplitude in
nonzero ancillary states is at most $2.81\times10^{-15}$.
Larger cases verify the matrix factorization and the preparation of
each direction, not every column of the complete ancillary circuit.

\subsection{Conditioning and endpoint diagnostics}
\label{app:syn-margins}
We measure the conditioning of $H_l$ in the fixed unitary-seed,
slot-major coordinates, before the final Cholesky conjugation
by $\mathsf C_l$. The numerical assembly uses the sine formula in
Lemma~\ref{lem:syn-seed} and the free-generator formulas in
\eqref{eq:syn-free-generators}.
The trivial seed coordinate is first, followed by the reduced seed
coordinates, as in Theorem~\ref{thm:syn-blocks}. The implementation
is \texttt{anc/error\_margins/verify\_error\_margins.py}.
The script keeps the Hermitian part of the floating-point assembly,
records its initial Hermiticity defect, and applies upper Cholesky.

In the following tables, floating-point quantities approximate their
exact counterparts. Write
\[
 \rho_F=\frac{\|\widehat{\mathsf{C}}^\dagger\widehat{\mathsf{C}}-\widehat H\|_F}
                    {\|\widehat H\|_{\mathrm{op}}},\qquad
 d_U=\max_i\|V_i^\dagger V_i-I\|_{\mathrm{op}},\quad
 V_i=\widehat{\mathsf{C}}\widehat A_i\widehat{\mathsf{C}}^{-1}.
\]
The values of $d_U$ measure the unitarity defect after the numerical
triangular solves. The quantity $\rho_F$ measures the Cholesky
factorization residual; it does not bound the distance to the exact
unitary target. No polar correction is applied. At the main parameters specified in
Section~\ref{subsec:resource-counts}, the selected cases are:
\begin{center}\small
\begin{tabular}{rrrrrr}
\toprule
$n$ & $M$ & $h_{\min}$ & $\mu_{\mathrm{pos}}$ & $\kappa_2(H)$ & $\rho_F$\\
\midrule
2 & 6 & $4.59\times10^{-1}$ & 0.2912 & 3.4338 & $1.75\times10^{-16}$ \\
4 & 20 & $2.19\times10^{-1}$ & 0.2411 & 4.1470 & $2.30\times10^{-16}$ \\
8 & 72 & $1.08\times10^{-1}$ & 0.2223 & 4.4983 & $3.85\times10^{-16}$ \\
12 & 156 & $7.21\times10^{-2}$ & 0.2167 & 4.6149 & $6.67\times10^{-16}$ \\
\bottomrule
\end{tabular}
\end{center}
This finite sample has moderate KLM conditioning. Absolute $h_{\min}$
depends on scale, whereas the
normalized margin does not. The condition number of the input Hermitian matrix is recorded
separately in the machine-readable results and is not substituted for
the condition number of $H_l$.

To vary $l$ while keeping the seed fixed, take $n=4$,
$a=0.04$, $c=0.07$, so $l_{\mathrm{first}}=0.52$, and write
$t=l/l_{\mathrm{first}}$. Selected values from the twelve-point sweep
are:
\begin{center}\small
\begin{tabular}{rrrrr}
\toprule
$t$ & $\mu_{\mathrm{pos}}$ & $\kappa_2(H)$ & $\rho_F$ & $d_U$\\
\midrule
$10^{-8}$ & $4.59\times10^{-9}$ & $2.18\times10^{8}$ & $1.52\times10^{-16}$ & $8.76\times10^{-9}$ \\
$10^{-6}$ & $4.59\times10^{-7}$ & $2.18\times10^{6}$ & $1.73\times10^{-16}$ & $1.38\times10^{-10}$ \\
$10^{-2}$ & $4.53\times10^{-3}$ & $2.21\times10^{2}$ & $1.32\times10^{-16}$ & $1.02\times10^{-14}$ \\
$1/2$ & $1.15\times10^{-1}$ & $8.68\times10^{0}$ & $1.72\times10^{-16}$ & $8.54\times10^{-16}$ \\
$1-10^{-2}$ & $5.83\times10^{-3}$ & $1.72\times10^{2}$ & $2.31\times10^{-16}$ & $3.21\times10^{-15}$ \\
$1-10^{-6}$ & $5.93\times10^{-7}$ & $1.69\times10^{6}$ & $2.09\times10^{-16}$ & $2.00\times10^{-13}$ \\
$1-10^{-8}$ & $5.93\times10^{-9}$ & $1.69\times10^{8}$ & $1.46\times10^{-16}$ & $2.22\times10^{-12}$ \\
\bottomrule
\end{tabular}
\end{center}
The small Cholesky residual persists when the relative positive margin
is small. The measured defect can be substantially larger near an
endpoint, as anticipated by the inverse-singular-value factor in
Section~\ref{subsec:syn-error-margins}. Neither endpoint itself is used
in the table; their different degenerations are described after
Proposition~\ref{prop:syn-endpoint-rates}.

Distance from these definite-window endpoints alone does not control the
input boundary. At $n=4$, $a=0.05$, $l=0.2$, decreasing $c$ from $0.05$
to $10^{-7}$ changes $\kappa_2(H)$ from approximately $10.31$ to
$2.099\times10^6$, although $l_{\mathrm{first}}\to0.75$. This is
consistent with the nonzero $K$ appearing at $c=0$. Conversely, the
four $a=c$ checks have condition numbers between $4.80$ and $5.42$:
eigenvalue multiplicity of a generator does not itself force the
positive form to be ill-conditioned in these samples.

There are 31 conditioning cases: eight main cases, four $a=c$ cases,
twelve fixed-seed choices of $l$, and seven input-boundary cases.
The assembly was cross-checked against the synthesis implementation
used in Section~\ref{subsec:resource-counts} at $n=2,3,4,6$. Controlled within-block perturbations of its actual $2n$
reflection directions verify Lemma~\ref{lem:syn-reflection-error} in
four further tests. Another 27 tests use generic directions with
arbitrary vector phases and possibly nonorthogonal approximate
projections. Six Hermitian-perturbation tests verify the positive-margin
inequality and its sharpness for an adverse rank-one perturbation.

Seven small cases, including the two $10^{-8}$ endpoint distances,
were reassembled independently with 70 decimal digits. Here the
binary64 parameter inputs are treated as exact real numbers, so the
comparison isolates evaluation errors from initial parameter rounding.
The largest relative discrepancy in $h_{\min}$ was about
$4.42\times10^{-8}$, in the ill-conditioned near-wall case; in the
main $n=4$ case it was below $10^{-15}$. For the conditions on
validated data and residuals needed for certification, see
Proposition~\ref{prop:syn-positivity-margin}.

\paragraph{Endpoint rates and parameter selection.}
The matched-benchmark data also record 48 endpoint-rate diagnostics on six
fixed seeds. They verify the interpolation identity and the endpoint
ranks numerically, as illustrations of
Proposition~\ref{prop:syn-endpoint-rates}.
For the earlier seed $n=4$, $a=0.04$, $c=0.07$, a 201-point grid and
local refinement suggest $l\simeq0.2762844$ with
$\kappa_2(H_l)\simeq8.647915$, versus $8.677317$ at the midpoint
$l=0.26$. This numerical candidate is not used to retune the benchmark.

\subsection{Comparison with a Quantum Shannon reference under matched resources}
\label{app:matched-benchmark}
We compare the CNOT counts of words produced by two algorithms on the
same padded targets, using the common protocol below.

\Needspace{28\baselineskip}
\paragraph{Protocol and algorithms.}
\begin{center}\small
\renewcommand{\arraystretch}{1.15}
\begin{tabular}{@{}>{\raggedright\arraybackslash}p{0.14\linewidth}
                  >{\raggedright\arraybackslash}p{0.81\linewidth}@{}}
\toprule
Method & Construction and evaluation conditions \\
\midrule
Structured & Supported reflection product
(Corollary~\ref{cor:syn-directions}). For the all-zero phase, choose the
shorter of a diagonal phase tree and the explicit ancilla-free
square-root recursion of \cite[arXiv version, Section~7.2]{BarencoGates}; borrowed data
wires are restored for every input. Taking this minimum preserves the
polynomial controlled-phase upper bound. Every tested case selects
the phase tree. \\
\addlinespace
QSD reference & Cosine--sine decomposition, unitary demultiplexing
(block-diagonal factorization into common outer unitaries and a
controlled diagonal factor), and recursive QSD
\cite[Theorems~10, 12 and 13]{SBM}, using SciPy's CS and complex Schur
routines. Recursion stops at one-qubit unitaries; the three-CNOT
two-qubit base, the two additional optimizations in the Appendix
of \cite{SBM}, and block-ZXZ are not used.
The unsimplified generic count is $(3/4)4^m-(3/2)2^m$,
the QSD ($\ell=1$) bound in \cite[Table~I]{SBM}; reported counts
come from generated and simulated gate lists. \\
\addlinespace
Common & Same slot/seed bit order, Cholesky realization and identity
padding; $m$ data qubits, arbitrary one-qubit $U(2)$ gates, all-to-all
CNOT and zero ancillary qubits. Identical one-qubit instruction
conventions and deterministic postprocessing: fuse one-qubit gates,
cancel identical CNOTs across disjoint instructions, remove near-scalar
identities at threshold $10^{-13}$, and track global phase.
Accept measured phase-aligned operator-norm error at most $10^{-8}$
relative to the saved floating-point target, on the whole padded
register. Simulate every computational-basis input column, including
unused padding labels. \\
\bottomrule
\end{tabular}
\end{center}
The trace-overlap phase gives an upper bound on $d_{\mathrm{proj}}$.
Distances from the analytically defined matrices additionally require
validated data and evaluation errors
(Section~\ref{subsec:syn-error-margins}).

\paragraph{Measured circuits.}
The main data use the fixed parameter rule of the earlier experiment,
without optimizing $l$ or the encodings for either compiler.
All generators for $n=2,3,4,5,6$ and generators $1,4,7$ for $n=8$
are included. Ranges below are over the listed generators, not over
all admissible parameters.
\begin{center}\small
\begin{tabular}{rrrrrr}
\toprule
$n$ & $M$ & $m$ & Targets & Structured CX & QSD reference CX\\
\midrule
2 & 6 & 3 & 1 & 116 & 36 \\
3 & 12 & 4 & 2 & 198--202 & 144--164 \\
4 & 20 & 5 & 3 & 554--576 & 718--720 \\
5 & 30 & 6 & 4 & 986--1012 & 2016--2976 \\
6 & 42 & 6 & 5 & 1162--1192 & 2304--2976 \\
8 & 72 & 7 & 3 & 3178--3218 & 12096 \\
\bottomrule
\end{tabular}
\end{center}
For seven data qubits, the reported QSD counts equal the
unsimplified generic count $12096$.

Each method is run three times after numerical-library warm-up.
The median construction-plus-postprocessing times, in seconds, have
the following ranges. Target/Cholesky construction and full-matrix
verification are excluded from both clocks and recorded separately.
\begin{center}\small
\begin{tabular}{rrr}
\toprule
$n$ & Structured & QSD reference\\
\midrule
2 & 0.004 & 0.002 \\
3 & 0.005 & 0.008 \\
4 & 0.011--0.017 & 0.040--0.057 \\
5 & 0.021--0.043 & 0.111--0.186 \\
6 & 0.027--0.043 & 0.141--0.197 \\
8 & 0.079--0.087 & 0.775--0.928 \\
\bottomrule
\end{tabular}
\end{center}
These timings concern Python reference implementations with one BLAS
thread, not asymptotic classical compilation bounds or hardware gate
times. The recorded environment is Python~3.13.5, NumPy~2.3.5 and
SciPy~1.17.0 on x86-64 Linux.

At $n=4$, all three generators were also tested for $a=c=0.04$ and
for $(a,c)=(0.04,0.07)$, always with
$l=(1-nc-(n+1)a)/3$. Including these six supplemental targets gives
24 targets and 48 fully simulated circuits. The maximum phase-aligned
operator-norm error is $9.31\times10^{-14}$, comfortably below the
common threshold; padding action and leakage are also recorded.
Ten Haar-unitary tests through five qubits independently check the
reference QSD construction, and six full-register tests check the
borrowed-data all-zero-phase implementation.

\paragraph{Interpretation.}
The structured implementation produces more CNOT factors at $n=2,3$
and fewer than this particular QSD reference for the remaining tested
targets.
The comparison is restricted to these two reference implementations;
optimized QSD, block-ZXZ and quantum-software baselines remain to be
tested. Installation of the Qiskit software package was blocked by the
execution environment's network restriction, so the supplied Qiskit
runner is marked unexecuted.

Complete targets, gate lists and execution instructions are described
in Appendix~\ref{app:repro-matched}.
\section{Concluding remarks and open problems}\label{sec:outlook}

The prior invariant-form and signature results supply the unitary
family analyzed here. Tables~\ref{tab:results-burau} and~\ref{tab:results-other}
separate this parameter calculation from the structural and spectral
conclusions and the elementary-word bounds in a constructed common
encoding. Appendix~\ref{app:additional-definite-regions} applies the
same signature formula beyond $\mathcal T_n$; the structure and
synthesis results retain their stated domains.

Further questions concern the closures of the braid images on chosen
summands such as $\mathcal R_n$, certified classical computation of
the matrix factors, and word lengths for encodings and generating sets
with restricted connectivity. For image closures and the separate
role of independently supplied local operations, see
Section~\ref{sec:gates}.

Another algebraic problem is to determine which parameter loci admit
the quadratic-tensor descriptions of \cite{BauerLloyd}, and to identify
obstructions. That framework describes certain classically tractable
tensors by quadratic functions on abelian groups or, more generally,
(super) Hopf algebras. Membership of the present braid family in such
structured classes remains open.

\begin{remark}[Deligne--Simpson rank data]\label{rem:deligne-simpson}
The multiplicative Deligne--Simpson problem specifies conjugacy classes
and a product relation; its standard quiver data include ranks of
partial products of $(M_j-\xi I)$ \cite[arXiv version, Introduction, p.~2]{CBS}.
Suppose a monodromy tuple is simultaneously unitarizable and
$U_j=\mathsf{C}M_j\mathsf{C}^{-1}$. For an eigenvalue $\xi_j$ one has $|\xi_j|=1$, and
Lemma~\ref{lem:syn-length} immediately gives
\[
 \operatorname{rank}(M_j-\xi_j I)
  =\ell_{\mathrm{ref}}(\xi_j^{-1}U_j).
\]
Choosing an eigenvalue of maximum multiplicity gives the projective
reflection length, interpreting the existing rank data as a resource
count. Direction-preparation costs require additional structure,
supplied here for the Burau--KLM family. Simultaneous positive-definite
unitarizability and extension to $F_n\rtimes B_n$ are separate
requirements beyond prescribed unit-circle spectra; local monodromies
and braid generators must be distinguished.
\end{remark}

\clearpage
\appendix

\section{Quantum-computing background for mathematicians}
\label{app:quantum-background}

This appendix fixes the finite-dimensional quantum-information language
used for gates, encodings, and circuits in this paper. It assumes linear
algebra and basic representation theory, but no quantum-computing
background. General references are \cite{Watrous,WatrousCourse}; the more
specialized results used below are cited where they enter.
Tables~\ref{tab:quantum-dictionary} and~\ref{tab:circuit-dictionary}
summarize the terminology for the gate properties and circuit
constructions, respectively.

\subsection{States, probabilities, and composite systems}
\label{app:states}
A \emph{qudit} is a tensor factor $\C^d$ with its standard basis;
\emph{qubit} and \emph{qutrit} refer to $d=2$ and $d=3$.

We begin with coordinates. A finite quantum system with $N$ levels is
modeled by $\mathcal H=\C^N$, equipped with its standard orthonormal
basis $e_0,\ldots,e_{N-1}$ and inner product
\[
 \langle v,w\rangle=v^\dagger w
 =\sum_{j=0}^{N-1}\overline{v_j}w_j.
\]
Thus vectors are columns, $v^\dagger$ is the conjugate-transpose row,
and operators are matrices in $\operatorname{Mat}_N(\C)$. The inner
product is conjugate-linear in its first argument. Choosing an
orthonormal basis identifies any $N$-dimensional complex Hilbert
space with this model; the choice also fixes the computational basis.
The integer $N$ in this appendix denotes a system's dimension, not
necessarily the input dimension used in the main text.

\begin{definition}[Pure states]\label{def:pure-states}
A \emph{pure state} is a column $\psi=\sum_j\psi_je_j$ with
$\psi^\dagger\psi=1$, where $\psi$ and $e^{it}\psi$ represent the
same state. Equivalently it is the rank-one orthogonal projection
\[
 P_\psi=\psi\psi^\dagger,
 \qquad (P_\psi)_{jk}=\psi_j\overline{\psi_k}.
\]
Thus pure states form $\mathbb P(\C^N)$.
\end{definition}
A \emph{qubit} means $N=2$:
\[
 e_0=\begin{pmatrix}1\\0\end{pmatrix},\quad
 e_1=\begin{pmatrix}0\\1\end{pmatrix},\quad
 \psi=\begin{pmatrix}\alpha\\\beta\end{pmatrix},\qquad
 |\alpha|^2+|\beta|^2=1.
\]
A qutrit has $N=3$, and a $d$-level qudit has $N=d$.

\begin{definition}[Density matrices]\label{def:density-matrices}
A general state is a \emph{density matrix}
$\varrho\in\operatorname{Mat}_N(\C)$ with $\varrho\succeq0$ and
$\operatorname{Tr}\varrho=1$.
\end{definition}
A probabilistic preparation of unit
columns $\psi_j$ with probabilities $p_j$ has
\[
 \varrho=\sum_jp_j\psi_j\psi_j^\dagger,
 \qquad p_j\ge0,\quad\sum_jp_j=1.
\]
Different preparations can give the same density matrix. A measurement
in the computational basis has outcome probabilities
\[
 \Pr(j)=\operatorname{Tr}(e_je_j^\dagger\varrho)=\varrho_{jj};
 \qquad \Pr(j)=|\psi_j|^2\quad\hbox{for a pure state.}
\]
This is the Born rule in coordinates.
\begin{definition}[Measurement probabilities]\label{def:measurements}
A \emph{positive operator-valued measure} is a finite family of positive
semidefinite matrices $M_j$ with $\sum_jM_j=I_N$. On a state $\varrho$,
its outcome probabilities are $\Pr(j)=\operatorname{Tr}(M_j\varrho)$.
The matrices $M_j$ alone need not specify the post-measurement state.
\end{definition}

\paragraph{Dictionary for bra and ket notation.}
Physics notation abbreviates these same matrix operations:
\[
 |j\rangle=e_j,\qquad |\psi\rangle=\psi,\qquad
 \langle\psi|=\psi^\dagger,\qquad
 \langle\psi|\phi\rangle=\psi^\dagger\phi,\qquad
 |\psi\rangle\langle\phi|=\psi\phi^\dagger.
\]
Here a ket is a column and a bra is its conjugate-transpose row;
$\psi\phi^\dagger$ acts by $v\mapsto\psi(\phi^\dagger v)$.
These are not additional mathematical objects. We use columns and
matrices below, so this notation need not be learned first.

For a composite system we specify $\mathcal H_A=\C^p$ and
$\mathcal H_B=\C^q$, with standard bases $e_j$ and $f_k$. We order
$e_j\otimes f_k$ lexicographically, identifying $\C^p\otimes\C^q$
with $\C^{pq}$. A product vector $x\otimes y$ has coordinates
$(x_jy_k)_{j,k}$; a general vector has coordinates $(c_{jk})_{j,k}$.
In physics notation $e_j\otimes f_k$ is also written $|jk\rangle$.
An $s$-qubit register has dimension $2^s$, whereas the full KLM
space in this paper has dimension $n(n+1)$, not $2^n$.

\begin{definition}[Partial trace]\label{def:partial-trace}
The subsystem density matrix $\varrho_A=\operatorname{Tr}_B\varrho$
is given in these bases by
\[
 (\varrho_A)_{ij}=\sum_{k=0}^{q-1}\varrho_{(i,k),(j,k)}.
\]
Equivalently $\operatorname{Tr}(T\varrho_A)
=\operatorname{Tr}((T\otimes I_q)\varrho)$ for every
$T\in\operatorname{Mat}_p(\C)$.
\end{definition}
For two qubits the Bell vector is
\[
 \Phi^+=\frac{e_0\otimes e_0+e_1\otimes e_1}{\sqrt2}
       =\frac1{\sqrt2}(1,0,0,1)^T,
 \qquad
 P_{\Phi^+}=\frac12\begin{pmatrix}
 1&0&0&1\\0&0&0&0\\0&0&0&0\\1&0&0&1
 \end{pmatrix}.
\]
Both partial traces are $I_2/2$, although $P_{\Phi^+}$ is pure.
The classical mixture $\operatorname{diag}(1/2,0,0,1/2)$ has the
same computational-basis probabilities but lacks the off-diagonal
entries. For instance $\operatorname{Tr}((X\otimes X)P_{\Phi^+})=1$,
whereas the corresponding expectation for that mixture is zero;
$X$ is defined next.

\subsection{Unitary gates, phases, and braid words}
\label{app:gates-basics}
\begin{definition}[Unitary gates]\label{def:quantum-gates}
A reversible quantum gate on $\C^N$ is a matrix $U$ satisfying
$U^\dagger U=I_N$. It acts on pure vectors by $\psi\mapsto U\psi$
and on density matrices by $\varrho\mapsto U\varrho U^\dagger$.
For a specified orthonormal control basis $e_0,\ldots,e_{p-1}$,
a \emph{controlled unitary} on $\C^p\otimes\mathcal H_B$ has the form
\[
 \sum_{r=0}^{p-1}(e_re_r^\dagger)\otimes V_r,
 \qquad V_r\in U(\mathcal H_B).
\]
It preserves each subspace $\C e_r\otimes\mathcal H_B$ and acts
there by $V_r$.
\end{definition}
Multiplying $U$ by a scalar of modulus one leaves its induced action
on $\mathbb P(\mathcal H)$ and on density matrices unchanged. Consequently
$PU(\mathcal H)=U(\mathcal H)/U(1)$ is the appropriate group when
only an isolated gate's action on states is considered.

\begin{definition}[Projective order and image]\label{def:projective-entangling}
Let $N\ge1$ and write
\[
 \pi_N:U(N)\longrightarrow PU(N)
 :=U(N)/\{zI_N:|z|=1\}.
\]
For $U\in U(N)$, its \emph{projective order} is the group-theoretic
order of $[U]=\pi_N(U)$. Equivalently, it is the least integer $k\ge1$
for which $U^k=e^{i\theta}I_N$ for some $\theta\in\R$; if none exists,
the projective order is infinite. For a subgroup $G\subset U(N)$,
its \emph{ordinary image} is $G$ itself and its \emph{projective
image} is $\pi_N(G)$.
\end{definition}
Thus ordinary order asks when a matrix becomes the identity, whereas
projective order asks when it becomes a scalar. For example,
$e^{2\pi i\alpha}I_N$ with $\alpha\notin\Q$ has projective order
one and infinite ordinary order. Projective order is independent of
any tensor-product structure or basis. Finiteness of the
order of each generator alone does not establish finiteness of the
group that they generate.

Passing to the projective group does not
make every occurrence of a scalar physically irrelevant. On a direct
sum,
\[
 I_{\mathcal H_0}\oplus U\quad\hbox{and}\quad
 I_{\mathcal H_0}\oplus e^{i\theta}U
\]
are generally different projective operators: the phase is relative
to $\mathcal H_0$. Section~\ref{subsec:clifford-benchmark} is an
explicit instance of this distinction. Likewise, choosing a controlled
implementation of a gate requires specifying its relative phase on
the control branches.

Useful qubit gates are
\[
 X=\begin{pmatrix}0&1\\1&0\end{pmatrix},\quad
 Z=\begin{pmatrix}1&0\\0&-1\end{pmatrix},\quad
 H_{\mathrm{Had}}=\frac1{\sqrt2}
       \begin{pmatrix}1&1\\1&-1\end{pmatrix},\quad
 S_{\mathrm{ph}}=\operatorname{diag}(1,i).
\]
On two qubits, in the ordered basis
$e_0\otimes e_0,e_0\otimes e_1,e_1\otimes e_0,e_1\otimes e_1$,
the controlled-NOT gate is the permutation matrix
\[
 \operatorname{CNOT}=
 \begin{pmatrix}1&0&0&0\\0&1&0&0\\0&0&0&1\\0&0&1&0\end{pmatrix},
 \quad
 \operatorname{CNOT}(e_x\otimes e_y)
  =e_x\otimes e_{y+x\bmod2}.
\]
A direct calculation gives
\[
 \operatorname{CNOT}(H_{\mathrm{Had}}\otimes I_2)(e_0\otimes e_0)
 =\Phi^+.
\]
This example exhibits both a tensor-product convention and an input
state; neither is supplied by an abstract list of eigenvalues alone.

Given a representation $U:B_n\to U(\mathcal H)$, a braid word
$b=\sigma_{i_t}^{\epsilon_t}\cdots\sigma_{i_1}^{\epsilon_1}$ acts as
\[
 U(b)=U_{i_t}^{\epsilon_t}\cdots U_{i_1}^{\epsilon_1},
 \qquad\epsilon_j\in\{1,-1\},
\]
where the rightmost factor acts first and $U_i^{-1}=U_i^\dagger$.
The braid relations identify different words with exactly the same
operator. Word length is an algebraic resource count for this chosen
generating set. It becomes a physical gate count only after an
implementation model is specified, and it is not automatically circuit
depth, which also depends on which operations may be performed in
parallel.

Unitarizing a representation is a different operation. The factor
$\mathsf C$ in $\mathsf C^\dagger\mathsf C=H_+$ identifies the Hilbert space with inner
product $v^\dagger H_+w$ with the standard Euclidean Hilbert space;
$\mathsf C A_i\mathsf C^{-1}$ then expresses the same abstract action in orthonormal
coordinates. The matrix $\mathsf C$ need not itself be unitary for the
standard inner product, so its existence does not describe a
physically available preprocessing gate.

\subsection{Entanglement and the choice of tensor factors}
\label{app:entanglement}
\begin{definition}[Product and entangled vectors]\label{def:product-vectors}
A bipartite \emph{tensor-product structure} on $\mathcal H$ is a
specified unitary identification
\[
 \varphi:\mathcal H\longrightarrow\mathcal H_A\otimes\mathcal H_B,
 \qquad \dim\mathcal H_A,\dim\mathcal H_B\ge2.
\]
A nonzero vector $\psi\in\mathcal H$ is a \emph{product vector}
if $\varphi\psi=x\otimes y$ for nonzero $x\in\mathcal H_A$ and
$y\in\mathcal H_B$; it is \emph{entangled} otherwise. Thus product
vectors are the nonzero decomposable tensors in these coordinates.
\end{definition}
The nonzero decomposable tensors form a set, not a linear subspace.
Their projective classes form the \emph{Segre variety}, the image of
\[
 \mathbb P(\mathcal H_A)\times\mathbb P(\mathcal H_B)
 \longrightarrow\mathbb P(\mathcal H_A\otimes\mathcal H_B),
 \qquad ([x],[y])\longmapsto[x\otimes y].
\]
In orthonormal bases, write
$\varphi\psi=\sum_{j,k}c_{jk}e_j\otimes f_k$. It is a product
vector exactly when the coefficient matrix $(c_{jk})$ has rank one.

\begin{definition}[Schmidt coefficients and entropy]\label{def:schmidt-data}
For a unit vector $\psi$, the nonzero singular values
$s_1,\ldots,s_r$ of $(c_{jk})$ are its \emph{Schmidt coefficients}.
The singular-value decomposition gives
\[
 \varphi\psi=\sum_{a=1}^{r}s_a u_a\otimes v_a,
 \qquad s_a>0,\quad\sum_{a=1}^{r}s_a^2=1,
\]
with orthonormal $u_a$ and $v_a$. The \emph{Schmidt rank} is $r$,
and the \emph{entanglement entropy} is
\[
 S(\psi)=-\sum_{a=1}^{r}s_a^2\log_2(s_a^2).
\]
\end{definition}
The Schmidt rank is one precisely for product states. The nonzero
eigenvalues of either reduced density operator are the $s_a^2$,
so $S(\psi)$ is the entropy of either subsystem. These quantities
depend on the tensor-product structure, but not on the choice of
orthonormal bases within the factors.
\begin{definition}[Separable mixed states]\label{def:separable-states}
For mixed states, \emph{separable} means a convex combination of
product density operators.
\end{definition} The entangling statements in this paper
concern pure product inputs, not a classification of mixed-state
separability.

\begin{definition}[Entangling gates]\label{def:entangling-gates}
Fix $\varphi:\mathcal H\to\mathcal H_A\otimes\mathcal H_B$ as in
Definition~\ref{def:product-vectors}. For $U\in U(\mathcal H)$, write
\[
 U_\varphi=\varphi U\varphi^{-1}:
 \mathcal H_A\otimes\mathcal H_B
 \longrightarrow\mathcal H_A\otimes\mathcal H_B.
\]
The gate $U$ is \emph{entangling} for this structure if there exist
nonzero $x\in\mathcal H_A$, $y\in\mathcal H_B$ such that
$U_\varphi(x\otimes y)$ is not decomposable. Equivalently, it does
not preserve the set of nonzero decomposable tensors, or its
projectivization does not preserve the Segre variety. Restricting
$x,y$ to unit vectors gives the same condition.
A \emph{local} or \emph{product gate} instead satisfies
$U_\varphi=X\otimes Y$ with $X\in U(\mathcal H_A)$ and
$Y\in U(\mathcal H_B)$. Such operators form the image of the
tensor-product representation of
$U(\mathcal H_A)\times U(\mathcal H_B)$; that representation is
not injective, since $(zI,z^{-1}I)$ acts trivially for $|z|=1$.
\end{definition}
It need not entangle every product input. In particular, CNOT
entangles $e_+\otimes e_0$, where
$e_+=(e_0+e_1)/\sqrt2$, but sends every
computational-basis vector to another product vector.
A unitary preserving all product vectors is a product
unitary, or, for equal factor dimensions, a product unitary composed
with the swap \cite[Thm.~4.1]{Johnston}. The swap itself does not
create entanglement. For $M=n(n+1)$ the equal-factor case is absent,
because $n^2<M<(n+1)^2$.

A decomposition $\mathcal H=\mathcal H_1\oplus\mathcal H_2$
and a factorization $\mathcal H=\mathcal H_A\otimes\mathcal H_B$
are different structures. Invariant summands of a braid
representation are not automatically separate qudits. Moreover, a
global change of orthonormal coordinates need not be local for a
previously specified tensor product. For instance, a global unitary
can take $e_0\otimes e_0$ to $\Phi^+$. This is why
Theorem~\ref{thm:entangling} explicitly quantifies over all unitary
identifications with bipartite tensor products: its spectral argument
excludes product operators for each such identification. It asserts
neither a canonical encoding nor a single input vector entangled
simultaneously for every encoding.

\subsection{Pauli operators, Clifford gates, and classical simulation}
\label{app:clifford}
The operators $X$ and $Z$ above satisfy $ZX=-XZ$.
\begin{definition}[Qubit Pauli and Clifford groups]\label{def:qubit-clifford}
On a qubit register,
the finite Pauli group consists of tensor products of $I,X,Z,XZ$
with scalar phases from $\{1,i,-1,-i\}$. Its Clifford group is the
unitary normalizer.
\end{definition}
A gate outside this normalizer is called \emph{non-Clifford}.
The same construction can be made for other local dimensions, as follows.

\begin{definition}[Finite-dimensional Weyl--Heisenberg frame]
\label{def:wh-frame}
Let $M=d_1\cdots d_k$ with $d_j\ge2$, and put
$h_{\mathbf d}=\operatorname{lcm}(d_1,\ldots,d_k)$. On
$\bigotimes_{j=1}^k\C^{d_j}$ define the local shift and phase operators
\[
X_j|r\rangle=|r+1\bmod d_j\rangle,\qquad
Z_j|r\rangle=\exp(2\pi i r/d_j)|r\rangle,
\]
where $|r\rangle=e_r$ denotes the $r$th standard basis vector;
the operators act as the identity on the other factors. We use the finite
Weyl--Heisenberg group
\[
 \mathcal P_{\mathbf d}
 =\langle X_j,Z_j\ (1\le j\le k),\,
       e^{\pi i/h_{\mathbf d}}I\rangle.
\]
Thus its allowed scalar phases are the $2h_{\mathbf d}$-th roots
of unity, not all of $U(1)$. A \emph{Weyl--Heisenberg frame} on an
$M$-dimensional Hilbert space $\mathcal H$ specifies the local dimensions,
their standard bases, and a unitary identification
$\varphi:\mathcal H\to\bigotimes_j\C^{d_j}$. Transport the finite
group to $\mathcal H$ by
\[
 \mathcal P_{\varphi,\mathbf d}
 :=\varphi^{-1}\mathcal P_{\mathbf d}\varphi.
\]
Its \emph{generalized Clifford group} is the normalizer
\[
 N_{U(\mathcal H)}(\mathcal P_{\varphi,\mathbf d})
 =\{U\in U(\mathcal H):
 U\mathcal P_{\varphi,\mathbf d}U^{-1}
       =\mathcal P_{\varphi,\mathbf d}\}.
\]
It contains all scalar unitaries. This convention extends the
finite-dimensional Pauli and Clifford conventions of \cite[arXiv version, \S\S~II.A--B]{HDM}.
\end{definition}

The projective Clifford group is finite, including for unequal local
dimensions. Indeed, conjugation on the finite group
$\mathcal P_{\varphi,\mathbf d}$ has image in the finite group
$\operatorname{Aut}(\mathcal P_{\varphi,\mathbf d})$. Its kernel is the
unitary centralizer of $\mathcal P_{\varphi,\mathbf d}$, which is
$U(1)I_{\mathcal H}$: the transported tensor Weyl monomials span
$\operatorname{End}(\mathcal H)$.
Consequently the normalizer modulo $U(1)I$ is finite.

For example,
\[
 H_{\mathrm{Had}}XH_{\mathrm{Had}}^\dagger=Z,\qquad
 H_{\mathrm{Had}}ZH_{\mathrm{Had}}^\dagger=X,\qquad
 S_{\mathrm{ph}}XS_{\mathrm{ph}}^\dagger=iXZ.
\]
CNOT conjugates Pauli tensors to Pauli tensors as well. These gates
are therefore Clifford gates relative to the indicated frame
\cite{HDM,Gottesman}.

In contrast,
$T_{\mathrm{ph}}=\operatorname{diag}(1,e^{\pi i/4})$ is
non-Clifford in the standard qubit frame, since
\[
 T_{\mathrm{ph}}XT_{\mathrm{ph}}^\dagger
  =\frac{X+Y}{\sqrt2},\qquad Y:=iXZ,
\]
which is not a scalar multiple of a Pauli operator. Nevertheless
$T_{\mathrm{ph}}^8=I$. Thus non-Cliffordness does not imply infinite
projective order. Theorem~\ref{thm:nonclifford} uses the converse
sufficient test: infinite projective order rules out membership in
\emph{every} finite-dimensional Weyl--Heisenberg Clifford frame,
not just the initially chosen one. For rational $(a,c)$ it proves
finite order, but does not by itself decide whether a gate normalizes
the Weyl--Heisenberg group in a specified frame, or whether the group
generated by all braid generators is finite.

For an $s$-qubit register, take an abelian Pauli subgroup of order
$2^s$ that does not contain $-I$. Its common fixed space is
one-dimensional, and a unit vector spanning it is a stabilizer state.
Such a state can be recorded by a number of binary data polynomial
in the qubit count, even though its vector has exponentially many
coordinates. Clifford conjugation updates its Pauli stabilizers.
This yields efficient classical simulation of circuits with
stabilizer preparations, Clifford gates and Pauli measurements
(including the usual classically controlled updates)
\cite{Gottesman}. The restrictions on preparations and measurements
matter: the statement is not a simulation theorem for Clifford gates
with arbitrary states or arbitrary observables. Entanglement alone
is not an obstruction to this simulation, as the Bell-state example
already illustrates.

Likewise, proving that a gate is non-Clifford is not a lower bound
on the classical complexity of a computational model containing it.
Complexity requires an asymptotic family, specified inputs and
measurements, and a precision model; see the discussion in
Section~\ref{sec:outlook}.

\subsection{Universality and image closure}
\label{app:universality}
\begin{definition}[Exact and approximate universality]\label{def:universal-gates}
For a fixed $N$-dimensional logical system, let $\mathcal G\subset U(N)$
be the available gates, with inverses allowed. It is \emph{exactly
universal} if
\[
 \pi_N(\langle\mathcal G\rangle)=PU(N),
\]
and \emph{approximately universal} if
\[
 \overline{\pi_N(\langle\mathcal G\rangle)}=PU(N).
\]
Here the closure is in the usual topology of $PU(N)$.
Exact universality means that every unitary, up to scalar phase, is
a finite product of the available gates and their inverses.
Approximate universality equivalently means that for every $V\in U(N)$
and $\epsilon>0$, some word $W$ in $\mathcal G^{\pm1}$ satisfies
\[
 \min_{\theta\in\R}\|V-e^{i\theta}W\|_{\mathrm{op}}<\epsilon.
\]
\end{definition}
Here universality refers to the ability to generate or approximate
arbitrary unitary gates up to global phase; it is distinct from a
universal property in category theory.
Approximate universality is a topological assertion. A synthesis theorem additionally
controls word length and the cost of finding such a word as functions
of the tolerance and the target's description.
For the particular gates constructed here,
Corollary~\ref{cor:syn-finite-precision} supplies a finite-alphabet
circuit-size bound at specified error, not an image-density statement
about the original braid generators. Physically local gate
sets and scalable registers require further uniformity conditions.
Dense braid images in particular topological models are the subject
of results such as \cite{FLW}; they do not follow just from unitarity.

The distinction between an invariant direct sum and a tensor product
also matters here: a gate may be entangling while preserving proper
invariant summands. Section~\ref{sec:gates} explains why the full braid
image in this paper is not universal and why irreducibility of an
individual summand does not settle its image closure.
An invariant encoded subspace, extra gates mixing summands, or a
measurement-assisted model changes the universality question and
must be specified separately. No such additional resources are
assumed in the present spectral theorems.
Section~\ref{app:fixed-encoding} gives an elementary example under
explicitly added local operations and distinguishes that resource
model from braid-only image density.

\Needspace{16\baselineskip}
\subsection{Encodings and logical gates}
\label{app:encodings}
\begin{definition}[Encodings and logical gates]\label{def:encodings}
An encoding of a $k$-dimensional logical system into an $N$-dimensional
physical system is a matrix
$J\in\operatorname{Mat}_{N\times k}(\C)$ with $J^\dagger J=I_k$.
Thus $J:\C^k\to\C^N$ is an isometry. Its image $\mathcal C$ is the
encoded subspace, and $P=JJ^\dagger$ is its orthogonal projection.
More generally, write $J:\mathcal H_{\mathrm{log}}\to
\mathcal H_{\mathrm{phys}}$ for this isometry. If $k=N$ and the
logical space is specified as $\bigotimes_j\C^{d_j}$, then $J$ is
unitary and determines the tensor-product structure
$\varphi=J^{-1}=J^\dagger$ on the physical space. Specified
orthonormal bases in the factors give a \emph{computational basis}
by taking their tensor products and transporting them by $J$.
For a proper encoded subspace, $J^\dagger$ is an inverse only on
$\operatorname{im}J$, and no tensor-product structure on the whole
physical space follows.
A physical unitary implements a logical unitary $V$ without leaving
the encoded subspace when $U_{\mathrm{phys}}J=JV$.
\end{definition}
An invariant subspace of a unitary representation supplies such an
abstract encoding and logical action. Preparing, accessing, or measuring
that subspace in a device requires a separate specification.

\begingroup
\small
\setlength{\LTcapwidth}{\linewidth}
\renewcommand{\arraystretch}{1.12}
\begin{longtable}{@{}>{\raggedright\arraybackslash}p{0.27\textwidth}
                      >{\raggedright\arraybackslash}p{0.51\textwidth}
                      >{\raggedright\arraybackslash}p{0.16\textwidth}@{}}
\caption{Quantum-gate terminology: a guide to the definitions.}
\label{tab:quantum-dictionary}\\
\hline
\textbf{Term} & \textbf{Mathematical meaning} & \textbf{Reference}\\
\hline
\endfirsthead
\multicolumn{3}{l}{\tablename~\thetable\ (continued)}\\
\hline
\textbf{Term} & \textbf{Mathematical meaning} & \textbf{Reference}\\
\hline
\endhead
\hline
\endfoot
Tensor-product structure & Specified unitary identification with a tensor product
& Def.~\ref{def:product-vectors}\\
Computational basis / encoding & Chosen product basis / isometric inclusion
& Def.~\ref{def:encodings}\\
Product / entangled vector & Decomposable / non-decomposable for the chosen factors
& Def.~\ref{def:product-vectors}\\
Schmidt coefficients / entropy & Singular values / entropy of their squares
& Def.~\ref{def:schmidt-data}\\
Local / entangling gate & Product operator / operator taking some product vector to an entangled one
& Def.~\ref{def:entangling-gates}\\
Weyl--Heisenberg frame / Clifford group & Specified finite Weyl--Heisenberg group / its unitary normalizer
& Def.~\ref{def:wh-frame}\\
Projective order & Order modulo scalar phases
& Def.~\ref{def:projective-entangling}\\
Exact / approximate universality & Equality / density of the generated projective group in $PU(\mathcal H)$
& Def.~\ref{def:universal-gates}\\
Controlled unitary & Block-diagonal action for a specified control basis
& Def.~\ref{def:quantum-gates}\\
\end{longtable}
\endgroup

\subsection{Circuits, synthesis, and approximation errors}
\label{app:circuits}
Definition~\ref{def:elementary-generators} in
Section~\ref{subsec:syn-elementary} fixes the generating set,
words, evaluation, size, and minimum word length.

\begin{definition}[Synthesis, ancillary systems, and padding]
\label{def:circuit-synthesis}
\emph{Synthesis} constructs a word whose evaluation equals or
approximates a target unitary. An \emph{ancilla} is an extra system
prepared in a specified state and used as working space.
A realization of $W\in U(\mathcal H_m)$ with $t$ \emph{clean ancillas}
uses a word on $\mathcal H_m\otimes(\C^2)^{\otimes t}$ whose
operator $W'$ satisfies
\[
 W'(v\otimes\eta_t)=(Wv)\otimes\eta_t
 \quad\text{for every }v\in\mathcal H_m,
 \qquad \eta_t=e_0^{\otimes t}.
\]
This initializes and restores the extra factors; it imposes no
condition outside that initialized subspace. An \emph{ancilla-free}
word uses $\mathcal H_m$ alone.
A \emph{borrowed} (or \emph{dirty}) work qubit is a data qubit used
temporarily without initialization. The completed subcircuit must act
as the desired operator tensored with the identity on that wire,
including for arbitrary inputs entangled with a reference system.
A \emph{one-control} applies its operation when the control bit is
$1$; a \emph{zero-control} applies it when the bit is $0$, and is
obtained by conjugating the one-control with $X$ on that bit.

For \emph{padding}, an isometric coordinate inclusion
$J:\C^M\to\mathcal H_m$, $2^m\ge M$, gives
\[
 \widehat U=JUJ^\dagger+(I_{2^m}-JJ^\dagger).
\]
The unused subspace is fixed pointwise and is not an extra tensor
factor. All generators use the same $J$; the choice
\eqref{eq:syn-encoding} retains separate slot and seed registers.
For $\mathcal C=\operatorname{im}J$ and $P=JJ^\dagger$, a physical
operator $\widetilde U$ \emph{leaks} from $\mathcal C$ if it does
not preserve $\mathcal C$. Its \emph{leakage} on unit inputs in
$\mathcal C$ is
$\|(I-P)\widetilde U P\|_{\mathrm{op}}$.
\end{definition}

\begin{definition}[Phase reflections and reflection length]
\label{def:phase-reflections}
For a unit vector $v\in\C^{d_0}$ and $\theta\in\R$, the rank-one
phase reflection
\[
 R(v,\theta)=I+(e^{i\theta}-1)vv^\dagger
\]
multiplies $v$ by $e^{i\theta}$ and fixes $v^\perp$
\cite{IKV,MIC}. The length $\ell_{\mathrm{ref}}(U)$ is the minimum
number of such factors, with unrestricted directions and angles;
$\ell_{\mathrm{ref}}^{\mathrm{proj}}(U)$ allows an overall scalar phase.
These lengths use a different generating set from $\mathcal G_m$.
On $\mathcal H_m$, \emph{state preparation} supplies a word with
matrix $W_v$ satisfying $W_ve_{\mathbf0}=e^{i\beta}v$.
The phase cancels in \eqref{eq:syn-phase-preparation}, which conjugates
the \emph{all-zero controlled phase} $R(e_{\mathbf0},\theta)$.
Preparing the direction includes arranging its support labels.
If $v\in\operatorname{im}J$, the completed reflection fixes
$(\operatorname{im}J)^\perp$; its individual factors need not do so.
\end{definition}

\begin{definition}[Approximation error and finite gate alphabets]
\label{def:approximation-error}\label{def:finite-alphabet}
Use the operator norm $\|A\|_{\mathrm{op}}=\sup_{\|x\|=1}\|Ax\|$
on the whole register. For unitary targets and approximants, the
error is $\|U-\widetilde U\|_{\mathrm{op}}$; minimizing over scalar
phases gives the projective distance $d_{\mathrm{proj}}$ defined in
Section~\ref{par:syn-finite-alphabet}. That paragraph also defines
the fixed finite local alphabet $\mathcal G_{\mathrm{fin}}$ and its
placements together with exact CNOT gates, $\mathcal F_m$.
For $\varepsilon>0$, minimizing over words $\mathfrak c$ on
$\mathcal F_m$ defines the \emph{approximate word length}
\[
 \ell_{\mathcal F_m}^{\mathrm{proj}}(V;\varepsilon)
 =\min\{|\mathfrak c|:
 d_{\mathrm{proj}}(V,\operatorname{ev}(\mathfrak c))\le\varepsilon\}.
\]
In a specified Clifford+$T$ word, the \emph{$T$-count} counts
$T_{\mathrm{ph}}$ and $T_{\mathrm{ph}}^\dagger$ letters, rather than all
letters or a minimum over words.
\end{definition}
Corollary~\ref{cor:syn-finite-precision} gives constructive upper
bounds using one-factor Solovay--Kitaev approximation
\cite{DawsonNielsenSK}. Its proof accounts for phase removal and
accumulated error on the whole padded register; approximants need
not fix padding or satisfy the braid relations exactly.
Section~\ref{subsec:syn-error-margins} treats errors in classical
preprocessing separately.

\begingroup
\small
\setlength{\LTcapwidth}{\linewidth}
\renewcommand{\arraystretch}{1.12}
\begin{longtable}{@{}>{\raggedright\arraybackslash}p{0.27\textwidth}
                      >{\raggedright\arraybackslash}p{0.51\textwidth}
                      >{\raggedright\arraybackslash}p{0.16\textwidth}@{}}
\caption{Circuit terminology: a guide to words and their implementation.}
\label{tab:circuit-dictionary}\\
\hline
\textbf{Term} & \textbf{Mathematical meaning} & \textbf{Reference}\\
\hline
\endfirsthead
\multicolumn{3}{l}{\tablename~\thetable\ (continued)}\\
\hline
\textbf{Term} & \textbf{Mathematical meaning} & \textbf{Reference}\\
\hline
\endhead
\hline
\endfoot
One-qubit gate / CNOT & One-factor unitary / controlled bit permutation
& Def.~\ref{def:elementary-generators}\\
Circuit / gate count & Word $\mathfrak c$ / size $|\mathfrak c|$
& Def.~\ref{def:circuits}\\
Minimum gate count / depth & Minimum word size / number of permitted layers
& Def.~\ref{def:circuits}\\
Padding / clean ancillas & Identity extension / initialized extra factors restored
& Def.~\ref{def:circuit-synthesis}\\
Borrowed wire / zero-control & Uninitialized work wire restored / control on the value $0$
& Def.~\ref{def:circuit-synthesis}\\
Leakage & Size of the component leaving the encoded subspace
& Def.~\ref{def:circuit-synthesis}\\
State preparation / reflection & Preparation of a direction / conjugation of the all-zero phase
& Def.~\ref{def:phase-reflections}\\
Finite alphabet / approximate synthesis & Words over $\mathcal F_m$ within the prescribed tolerance
& Def.~\ref{def:finite-alphabet}\\
$T$-count & Selected-letter count in a specified Clifford+$T$ word
& Def.~\ref{def:finite-alphabet}\\
\end{longtable}
\endgroup

\section{Reproducibility and supplementary records}\label{app:repro}
This appendix maps exact algebraic computations, floating-point matrix
and circuit records, and reproduction procedures to the ancillary files.

\subsection{Algebraic, spectral and encoded-gate checks}
The file \texttt{anc/README.txt} lists runnable scripts, with dependencies
in \texttt{anc/requirements.txt}. Finite-size implementation checks
of the action matrices, invariant forms, signatures, and spectral
identities are provided by \texttt{verify\_klm\_v2.py},
\texttt{twisted\_\allowbreak burau\_\allowbreak checks.py}, and
\texttt{charpoly\_\allowbreak and\_\allowbreak relations.py}.
These checks supplement the general proofs;
their finite ranges do not establish statements for arbitrary $n$.
The script \texttt{certificate\_\allowbreak locus\_\allowbreak checks.py} checks
Examples~\ref{ex:additional-negative-window}
and~\ref{ex:additional-positive-window}, using exact rational floor
identities and separate floating-point matrix tests.

The exact finite-field computation in \texttt{exact\_commutant.py}
is part of the proof of Proposition~\ref{prop:generic-irred}.
Its output \texttt{exact\_commutant.json} records the selected minor's
row and column indices and determinant residues.
The floating-point observations in Remark~\ref{rem:numerics}(a)
and~(b) are reproduced by \texttt{commutant\_numerics.py} and
\texttt{jones\_jordan\_scan.py}, respectively.
The script \texttt{fixed\_encoding\_checks.py} checks the matrix
factorizations, encoded restrictions, and circuit identities in
Section~\ref{app:fixed-encoding}.
The circuit counts, approximation errors, and conditioning experiments
are documented separately below and in Section~\ref{app:syn-counts}.

\subsection{Common-basis synthesis files}
\label{app:repro-synthesis}
The counts and verification scope of this implementation are reported
in Section~\ref{subsec:resource-counts}.
The entry point is
\texttt{verify\_reflection\_cost.py} in \texttt{anc/reflection\_synthesis/}.
Its \texttt{results/} directory contains a JSON summary, QASM circuits,
and dense NPZ matrix arrays. The seed register uses the low bits,
the slot register the high data bits, and ancillas are appended.
Python~3.10 or later, NumPy, SciPy, and SymPy are required. The
optional \texttt{--large} cases $n=16,24$ were not used for this table.

\subsection{Conditioning and high-precision records}
\label{app:repro-margins}
Section~\ref{app:syn-margins} contains the eigenvalue margins,
residuals, endpoint sweep and precision comparisons.
The script requires Python~3.10 or later, NumPy, SciPy, and mpmath;
its \texttt{results/diagnostics.json} records parameters, eigenvalue
margins, seed and KLM condition numbers, residuals, and all tests.

\subsection{Reference compilers and saved targets}
\label{app:repro-matched}
The protocol, numerical resource counts, timings and error checks are
reported in Section~\ref{app:matched-benchmark}; the endpoint diagnostics
are discussed in Section~\ref{app:syn-margins}.
The directory \texttt{anc/matched\_benchmark/} contains the reference
compilers, complete circuit lists, saved floating-point target matrices
with SHA-256 hashes, raw and postprocessed counts, timings and
full-register errors. The hashes identify the saved data; they do not
certify their distance from the analytically defined matrices. Re-running \texttt{run\_matched\_benchmark.py} reproduces the
reported protocol without network access. The optional
\texttt{run\_qiskit\_optional.py} is an unexecuted interface to the
Qiskit quantum-circuit software library.

\section{Definite parameter regions beyond \texorpdfstring{$\mathcal T_n$}{Tn}}
\label{app:additional-definite-regions}

The triangle $\mathcal T_n$ gives a simple common parameter domain for
the results in the main text, but it does not exhaust the definite
regions of the canonical form. We apply the signature formula of
\cite{Nsig}, recalled in Theorem~\ref{thm:import-sig}, to classify
these regions explicitly. The calculation concerns
the character-twisted unreduced Burau seed and the same choice
$\lambda^{1/2}=\exp(\pi il)$ as in Definition~\ref{def:eigenangles}.

\Needspace{24\baselineskip}
\begin{proposition}[Canonical definite windows]
\label{prop:additional-definite-regions}
Let $n\ge2$, $0<a<1/(n+1)$, $0<c<1$, and $c+2a\ne1$.
Put $q=\ee a$, $\tau=\ee c$, and $\lambda=\ee l$ with
$0<l<1$ and $\ee l\in E^c$.
Then $K=0$, $R=n(n+1)$, and the canonical form $\Ht(\ee l)$ is
definite exactly in the following four cases:
\begin{center}
\renewcommand{\arraystretch}{1.25}
\begin{tabular}{@{}lll@{}}
\toprule
Sign & Conditions on $(a,c)$ & Window in $l$ \\
\midrule
Positive & $nc+(n+1)a<1$ & $0<l<1-nc-(n+1)a$ \\
Positive ($n=2$) & $\frac14<a<\frac13,\quad 1-2a<c<\frac12$
  & $0<l<1-2c$ \\
Negative A & $1-\frac1n<c<1-2a$ & $n(1-c)<l<1$ \\
Negative B & $1-\frac an<c<1$ & $n+1-nc-a<l<1$ \\
\bottomrule
\end{tabular}
\end{center}
The first row is $\mathcal T_n$; the other three lie outside it.
In Negative A, the inequalities imply $a<1/(2n)$.
In either negative window, $-\Ht(\ee l)$ is a positive definite
invariant form.
\end{proposition}

\begin{proof}
The seed is unitarizable for the stated range of $a$.
Its free-group generators have eigenvalues
$\tau$ (with multiplicity $n$) and $\tau q^2$.
Neither is $1$ under the hypotheses, so $K=0$ and $R=n(n+1)$.
Write $\{x\}=x-\lfloor x\rfloor$ and set
\[
\begin{gathered}
h=\lfloor c+2a\rfloor\in\{0,1\},\qquad
b_0=\lfloor nc\rfloor,\quad b_1=\lfloor nc+a\rfloor,\quad
b_2=\lfloor nc+(n+1)a\rfloor,\\
(\gamma_1,\ldots,\gamma_{n+1})=
\bigl(\{nc\},\underbrace{\{nc+a\},\ldots,\{nc+a\}}_{n-1},
\{nc+(n+1)a\}\bigr).
\end{gathered}
\]
By Lemma~\ref{lem:fulltwist}, these $\gamma_k$ are the eigenangles
of $P_n$, counted with multiplicity. The sum of the seed eigenangles
is $n(n+1)c+2na-nh$, whereas
\[
\sum_k\gamma_k=n(n+1)c+2na-b_0-(n-1)b_1-b_2.
\]
Consequently Theorem~\ref{thm:import-sig} gives
\begin{equation}\label{eq:additional-floor-signature}
\begin{gathered}
w=b_0+(n-1)b_1+b_2-nh,\qquad
m(l)=w+\#\{k:\gamma_k+l\ge1\},\\
\sig\Ht(\ee l)=(R-m(l),m(l)).
\end{gathered}
\end{equation}
There is no equality $\gamma_k+l=1$ at a nonexceptional parameter.

\emph{Case $h=0$: positive definiteness.}
Here $c<1-2a$.
Since
$nc+(n+1)a<n-(n-1)a<n$, all three $b_j$ lie between
$0$ and $n-1$. Thus $0\le w\le n^2-1=R-(n+1)$.
Positive definiteness requires $m(l)=0$, hence $w=0$;
this holds exactly when $b_2=0$, or $nc+(n+1)a<1$.
The crossing count is then zero exactly for
$0<l<1-\max_k\gamma_k=1-nc-(n+1)a$, giving the first row.

\emph{Case $h=0$: negative definiteness.}
The equality $m(l)=R$ forces both
$w=n^2-1$ and a crossing count of $n+1$.
The former requires $b_0=b_1=b_2=n-1$, while the latter
requires every $\gamma_k>0$.
Together these conditions are equivalent to $nc>n-1$.
In this case $\min_k\gamma_k=nc-(n-1)$, so all crossings have
occurred exactly when $l>1-\min_k\gamma_k=n(1-c)$.
This is Negative A.

\emph{Case $h=1$: negative definiteness.}
Here $c>1-2a$.
The upper bounds $b_0\le n-1$ and $b_1,b_2\le n$ again give
$w\le n^2-1$.
Negative definiteness therefore requires
$(b_0,b_1,b_2)=(n-1,n,n)$ and every $\gamma_k>0$.
These conditions are equivalent to $nc+a>n$.
Here the smallest eigenangle is $nc+a-n$: the other two distinct
entries exceed it by $1-a$ and $na$, respectively.
Hence the negative window is $n+1-nc-a<l<1$, giving Negative B.

\emph{Case $h=1$: positive definiteness.}
If $n\ge3$, the lower bounds
$b_0,b_1\ge n-2$ and $b_2\ge n-1$ give
\[
w\ge n^2-2n-1>0,
\]
so $m(l)=0$ is impossible.
If $n=2$, then $b_0\ge0$ and $b_1,b_2\ge1$; thus $w\ge0$,
and $w=0$ holds exactly for $(b_0,b_1,b_2)=(0,1,1)$.
Under $c>1-2a$ and $0<a<1/3$, this is equivalent to
$1-2a<c<1/2$, which also forces $a>1/4$.
The eigenangles are then $2c$, $2c+a-1$, and $2c+3a-1$.
Since $3a<1$, the largest is $2c$, giving the second row.
These cases exhaust $h\in\{0,1\}$ and both definite signs.
\end{proof}

\begin{example}[A negative definite window]
\label{ex:additional-negative-window}
For $n=2$ and $(a,c)=(1/100,9/10)$, Negative A applies.
Here $w=3$ and the boundary-loop eigenangles are
$4/5$, $81/100$, and $83/100$.
Thus $\sig\Ht(\ee l)=(0,6)$ for $1/5<l<1$,
and $-\Ht(\ee l)$ supplies a positive definite invariant form.
\end{example}

\begin{example}[An additional positive definite window]
\label{ex:additional-positive-window}
For $n=2$ and $(a,c)=(3/10,9/20)$, the second row applies.
Here $w=0$ and the boundary-loop eigenangles are
$9/10$, $1/5$, and $4/5$.
Thus $\sig\Ht(\ee l)=(6,0)$ for $0<l<1/10$.
\end{example}

On $\mathcal T_n$, neither the seed eigenangles nor the boundary-loop
eigenangles require reduction modulo $1$, and the positive window
has the simple description in the first row. The other three rows also
yield unitary realizations after a possible change of sign. We do not
extend the main text's unitary decomposition, gate conclusions or
synthesis estimates to those regions here. This restriction does not
narrow the domain of algebraic identities already proved more generally:
Proposition~\ref{prop:charpoly} gives the ambient Long--Moody
characteristic polynomial for every $q,\tau\in\C^\times$.

The classification above concerns the canonical form only; it does not
classify the existence of an arbitrary positive definite invariant
form. It excludes $c=0$ and $c+2a=1$, where $K\ne0$ and the dimension
changes, and it does not classify definite quotients at exceptional
values of $\lambda$.

\section*{Acknowledgements}
The author thanks Rapha\"el Belliard, Manabu Hagiwara, Greyson
Potter, and Megan van Zyl for helpful discussions.

The author is especially grateful to Asian Oceanian Women in Mathematics
and Yukari Ito. Participation in its activities provided opportunities
to meet researchers in quantum computing, and these encounters gave
rise to the present work.

This work was supported by JST SPRING, Grant Number JPMJSP2109.

This manuscript was prepared with the assistance of AI language
models, which were used for literature search, drafting and editing
the exposition, and assistance with code for symbolic, numerical and
finite-field computational certificates. The author takes responsibility for the mathematical content.

\section*{Statements and Declarations}
The author has no competing interests to declare. Data availability:
the scripts reproducing all numerical and symbolic checks are provided
as ancillary files.


\begin{thebibliography}{99}\small
\bibitem{FiniteCompanion}
H.~Negami,
\emph{Finite images of braid group representations and algebraic solutions of KZ-type equations},
preprint, \href{https://arxiv.org/abs/2609.25805v1}{arXiv:2609.25805v1} (2026).

\bibitem{BauerLloyd}
A.~Bauer, S.~Lloyd,
\emph{Quadratic tensors as a unification of Clifford, Gaussian, and free-fermion physics},
preprint, \href{https://arxiv.org/abs/2601.15396v1}{arXiv:2601.15396v1} (2026).

\bibitem{GeorgievIsing}
L.~S.~Georgiev,
\href{https://doi.org/10.1088/1751-8113/42/22/225203}{\emph{Ultimate braid-group generators for coordinate exchanges of Ising anyons from the multi-anyon Pfaffian wavefunctions}},
J.\ Phys.\ A: Math.\ Theor.\ 42 (2009), 225203;
\href{https://arxiv.org/abs/0812.2334v2}{arXiv:0812.2334v2}.

\bibitem{AGW}
A.~Ahlbrecht, L.~S.~Georgiev, R.~F.~Werner,
\href{https://doi.org/10.1103/PhysRevA.79.032311}{\emph{Implementation of Clifford gates in the Ising-anyon topological quantum computer}},
Phys.\ Rev.\ A 79 (2009), 032311.

\bibitem{FunarKohno}
L.~Funar, T.~Kohno,
\href{https://doi.org/10.1007/s10711-013-9847-0}{\emph{On Burau's representations at roots of unity}},
Geom.\ Dedicata 169 (2014), 145--163; \href{https://arxiv.org/abs/0907.0568}{arXiv:0907.0568}.

\bibitem{BT}
S.~Bigelow, J.~P.~Tian,
\href{https://doi.org/10.1142/S0219199708003186}{\emph{Generalized Long--Moody representations of braid groups}},
Commun.\ Contemp.\ Math.\ 10 (2008), suppl.~1, 1093--1102; \href{https://arxiv.org/abs/0807.2881}{arXiv:0807.2881}.

\bibitem{AJL}
D.~Aharonov, V.~Jones, Z.~Landau,
\href{https://doi.org/10.1007/s00453-008-9168-0}{\emph{A polynomial quantum algorithm for approximating the Jones polynomial}},
Algorithmica 55 (2009), 395--421; \href{https://arxiv.org/abs/quant-ph/0511096}{arXiv:quant-ph/0511096}.

\bibitem{JonesHecke}
V.~F.~R.~Jones,
\href{https://doi.org/10.2307/1971403}{\emph{Hecke algebra representations of braid groups and link polynomials}},
Ann.\ of Math.\ (2) 126 (1987), 335--388.

\bibitem{Wenzl}
H.~Wenzl,
\href{https://doi.org/10.1007/BF01404457}{\emph{Hecke algebras of type $A_n$ and subfactors}},
Invent.\ Math.\ 92 (1988), 349--383.

\bibitem{FLW}
M.~H.~Freedman, M.~Larsen, Z.~Wang,
\href{https://doi.org/10.1007/s002200200645}{\emph{A modular functor which is universal for quantum computation}},
Comm.\ Math.\ Phys.\ 227 (2002), 605--622.

\bibitem{HN}
K.~Hiroe, H.~Negami,
\emph{Long--Moody construction of braid representations and Katz middle convolution},
preprint, \href{https://arxiv.org/abs/2303.05770v1}{arXiv:2303.05770v1} (2023).

\bibitem{HDM}
E.~Hostens, J.~Dehaene, B.~De~Moor,
\href{https://doi.org/10.1103/PhysRevA.71.042315}{\emph{Stabilizer states and Clifford operations for systems of arbitrary dimensions and modular arithmetic}},
Phys.\ Rev.\ A 71 (2005), 042315; \href{https://arxiv.org/abs/quant-ph/0408190v2}{arXiv:quant-ph/0408190v2}.

\bibitem{Johnston}
N.~Johnston,
\href{https://doi.org/10.1080/03081087.2011.596540}{\emph{Characterizing operations preserving separability measures via linear preserver problems}},
Linear Multilinear Algebra 59 (2011), 1171--1187; \href{https://arxiv.org/abs/1008.3633}{arXiv:1008.3633}.

\bibitem{Katz}
N.~M.~Katz,
\href{https://doi.org/10.1515/9781400882595}{\emph{Rigid Local Systems}},
Ann.\ of Math.\ Studies 139, Princeton Univ.\ Press, 1996.

\bibitem{Long}
D.~D.~Long,
\href{https://doi.org/10.4310/CAG.1994.v2.n2.a3}{\emph{Constructing representations of braid groups}},
Comm.\ Anal.\ Geom.\ 2 (1994), 217--238.

\bibitem{NKLM}
H.~Negami,
\emph{Long--Moody construction of braid group representations and Haraoka's multiplicative middle convolution for KZ-type equations},
preprint, \href{https://arxiv.org/abs/2503.14840v3}{arXiv:2503.14840v3} (2026).

\bibitem{NHecke}
H.~Negami, \emph{Hecke algebra representations from the Katz--Long--Moody construction}, preprint,
\href{https://arxiv.org/abs/2608.01727v2}{arXiv:2608.01727v2} (2026).

\bibitem{Nsig}
H.~Negami, \emph{A closed signature formula for the Katz--Long--Moody Hermitian form}, preprint,
\href{https://arxiv.org/abs/2609.11793v1}{arXiv:2609.11793v1} (2026).

\bibitem{RW}
E.~C.~Rowell, Z.~Wang, \href{https://doi.org/10.1007/s00220-011-1386-7}{\emph{Localization of unitary braid group representations}}, Comm.\ Math.\ Phys.\ 311 (2012), 595--615.

\bibitem{ST}
A.~Souli\'e, A.~Takano, \href{https://doi.org/10.1016/j.topol.2022.108393}{\emph{Extensions of the Tong--Yang--Ma representation}}, Topology Appl.\ 325 (2023), 108393; \href{https://arxiv.org/abs/2012.03767v5}{arXiv:2012.03767v5}.

\bibitem{Squier}
C.~C.~Squier, \href{https://doi.org/10.1090/S0002-9939-1984-0727232-8}{\emph{The Burau representation is unitary}}, Proc.\ Amer.\ Math.\ Soc.\ 90 (1984), 199--202.

\bibitem{Watrous}
J.~Watrous, \href{https://doi.org/10.1017/9781316848142}{\emph{The Theory of Quantum Information}}, Cambridge University Press, 2018.

\bibitem{WatrousCourse}
J.~Watrous, \emph{Understanding Quantum Information and Computation}, lecture notes, preprint,
\href{https://arxiv.org/abs/2507.11536v1}{arXiv:2507.11536v1} (2025).

\bibitem{Gottesman}
D.~Gottesman, \emph{The Heisenberg representation of quantum computers}, in \emph{Group22: Proceedings of the XXII International Colloquium on Group Theoretical Methods in Physics}, S.~P.~Corney, R.~Delbourgo, P.~D.~Jarvis (eds.), International Press, 1999, 32--43; expanded preprint, \href{https://arxiv.org/abs/quant-ph/9807006}{arXiv:quant-ph/9807006}.

\bibitem{JiangControlled}
G.-L.~Jiang, H.-R.~Wei, G.-Z.~Song, M.~Hua, \href{https://doi.org/10.1002/andp.202200317}{\emph{Synthesis and upper bound of Schmidt rank of bipartite controlled-unitary gates}}, Annalen der Physik 534 (2022), 2200317; \href{https://arxiv.org/abs/2209.04799v1}{arXiv:2209.04799v1}.

\bibitem{DiWei}
Y.-M.~Di, H.-R.~Wei,
\href{https://doi.org/10.1103/PhysRevA.87.012325}{\emph{Synthesis of multivalued quantum logic circuits by elementary gates}},
Phys.\ Rev.\ A 87 (2013), 012325;
\href{https://arxiv.org/abs/1302.0056v1}{arXiv:1302.0056v1}.

\bibitem{BarencoGates}
A.~Barenco, C.~H.~Bennett, R.~Cleve, D.~P.~DiVincenzo, N.~Margolus, P.~Shor, T.~Sleator, J.~A.~Smolin, H.~Weinfurter, \href{https://doi.org/10.1103/PhysRevA.52.3457}{\emph{Elementary gates for quantum computation}}, Phys.\ Rev.\ A 52 (1995), 3457--3467; \href{https://arxiv.org/abs/quant-ph/9503016v1}{arXiv:quant-ph/9503016v1}.

\bibitem{HarrowExact}
A.~W.~Harrow, \href{https://doi.org/10.26421/QIC9.9-10-4}{\emph{Exact universality from any entangling gate without inverses}}, Quantum Inf.\ Comput.\ 9 (2009), 773--777; \href{https://arxiv.org/abs/0806.0631v1}{arXiv:0806.0631v1}.

\bibitem{HLR}
J.~Huang, J.~B.~Lewis, V.~Reiner, \href{https://doi.org/10.1112/jlms.12013}{\emph{Absolute order in general linear groups}}, J.\ London Math.\ Soc.\ (2) 95 (2017), 223--247; \href{https://arxiv.org/abs/1506.03332v2}{arXiv:1506.03332v2}.

\bibitem{IKV}
P.~A.~Ivanov, E.~S.~Kyoseva, N.~V.~Vitanov, \href{https://doi.org/10.1103/PhysRevA.74.022323}{\emph{Engineering of arbitrary $U(N)$ transformations by quantum Householder reflections}}, Phys.\ Rev.\ A 74 (2006), 022323; \href{https://arxiv.org/abs/0708.2811}{arXiv:0708.2811}.

\bibitem{MIC}
E.~Malvetti, R.~Iten, R.~Colbeck, \href{https://doi.org/10.22331/q-2021-03-15-412}{\emph{Quantum circuits for sparse isometries}}, Quantum 5 (2021), 412; \href{https://arxiv.org/abs/2006.00016v2}{arXiv:2006.00016v2}.

\bibitem{SBM}
V.~V.~Shende, S.~S.~Bullock, I.~L.~Markov, \href{https://doi.org/10.1109/TCAD.2005.855930}{\emph{Synthesis of quantum-logic circuits}}, IEEE Trans.\ Computer-Aided Design Integr.\ Circuits Syst.\ 25 (2006), 1000--1010; \href{https://arxiv.org/abs/quant-ph/0406176v5}{arXiv:quant-ph/0406176v5}.

\bibitem{LiLuo}
L.~Li, J.~Luo, \href{https://doi.org/10.4230/LIPIcs.ICALP.2025.113}{\emph{Nearly optimal circuit size for sparse quantum state preparation}}, 52nd ICALP, LIPIcs 334 (2025), 113:1--113:19; \href{https://arxiv.org/abs/2406.16142v2}{arXiv:2406.16142v2}.

\bibitem{RupprechtWoelk}
F.~Rupprecht, S.~W\"olk, \href{https://doi.org/10.22331/q-2026-09-10-2208}{\emph{Sparse quantum state preparation with improved Toffoli cost}}, Quantum 10 (2026), 2208; \href{https://arxiv.org/abs/2601.09388v2}{arXiv:2601.09388v2}.

\bibitem{JordanWocjan} S.~P.~Jordan, P.~Wocjan,
\href{https://doi.org/10.1103/PhysRevA.80.062301}{\emph{Efficient quantum circuits for arbitrary sparse unitaries}},
Phys.\ Rev.\ A 80 (2009), 062301; \href{https://arxiv.org/abs/0904.2211v1}{arXiv:0904.2211v1}.

\bibitem{KrolAlArs} A.~M.~Krol, Z.~Al-Ars,
\href{https://doi.org/10.1103/PhysRevApplied.22.034019}{\emph{Beyond quantum Shannon decomposition: Circuit construction for
$n$-qubit gates based on block-$ZXZ$ decomposition}},
Phys.\ Rev.\ Applied 22 (2024), 034019; \href{https://arxiv.org/abs/2403.13692v2}{arXiv:2403.13692v2}.

\bibitem{GLPS} \begingroup\raggedright N.~Geer, A.~D.~Lauda, B.~Patureau-Mirand, J.~Sussan,
\href{https://doi.org/10.4153/S0008414X25101235}{\emph{Density and unitarity of the Burau representation from a
non-semisimple TQFT}}, Canad.\ J.\ Math., First View (2025), 1--29; \href{https://arxiv.org/abs/2402.13242v2}{arXiv:2402.13242v2}.\par\endgroup

\bibitem{CBS} W.~Crawley-Boevey, P.~Shaw,
\href{https://doi.org/10.1016/j.aim.2005.02.003}{\emph{Multiplicative preprojective algebras, middle convolution and
the Deligne--Simpson problem}}, Adv.\ Math.\ 201 (2006), 180--208; \href{https://arxiv.org/abs/math/0404186v2}{arXiv:math/0404186v2}.

\bibitem{KolpakovSwitch} \begingroup\raggedright A.~Kolpakov,
\emph{Burau representation, Squier's form, and non-Abelian anyons},
preprint, \href{https://arxiv.org/abs/2510.18186v4}{arXiv:2510.18186v4}
(2025; revised 2026).\par\endgroup

\bibitem{DawsonNielsenSK} C.~M.~Dawson, M.~A.~Nielsen,
\href{https://doi.org/10.26421/QIC6.1-6}{\emph{The Solovay--Kitaev algorithm}},
Quantum Inf.\ Comput.\ 6 (2006), 81--95; \href{https://arxiv.org/abs/quant-ph/0505030v2}{arXiv:quant-ph/0505030v2}.

\bibitem{KuperbergSK} \begingroup\raggedright G.~Kuperberg,
\emph{Breaking the cubic barrier in the Solovay--Kitaev algorithm},
preprint, \href{https://arxiv.org/abs/2306.13158v2}{arXiv:2306.13158v2}
(2023; revised 2025).\par\endgroup

\bibitem{DOVCholesky} Z.~Drma\v c, M.~Omladi\v c, K.~Veseli\'c,
\href{https://doi.org/10.1137/S0895479893244717}{\emph{On the perturbation of the Cholesky factorization}},
SIAM J.\ Matrix Anal.\ Appl.\ 15 (1994), 1319--1332.

\bibitem{NayakReview} C.~Nayak, S.~H.~Simon, A.~Stern, M.~Freedman, S.~Das Sarma,
\href{https://doi.org/10.1103/RevModPhys.80.1083}{\emph{Non-Abelian anyons and topological quantum computation}},
Rev.\ Mod.\ Phys.\ 80 (2008), 1083--1159; \href{https://arxiv.org/abs/0707.1889v2}{arXiv:0707.1889v2}.

\bibitem{DRWReview} C.~Delaney, E.~C.~Rowell, Z.~Wang,
\href{https://doi.org/10.15446/recolma.v50n2.62211}{\emph{Local unitary representations of the braid group and their
applications to quantum computing}}, Rev.\ Colombiana Mat.\ 50 (2016), 211--276; \href{https://arxiv.org/abs/1604.06429v1}{arXiv:1604.06429v1}.

\bibitem{DRMiddle} M.~Dettweiler, S.~Reiter,
\emph{On the middle convolution}, preprint,
\href{https://arxiv.org/abs/math/0305311v1}{arXiv:math/0305311v1} (2003).

\bibitem{GoubaultHouseholder} \begingroup\raggedright
T.~Goubault de Brugi\`ere, M.~Baboulin, B.~Valiron, C.~Allouche,
\href{https://doi.org/10.1016/j.cpc.2019.107001}{\emph{Quantum circuits synthesis using Householder transformations}},
Comput.\ Phys.\ Commun.\ 248 (2020), 107001; \href{https://arxiv.org/abs/2004.07710}{arXiv:2004.07710}.\par\endgroup

\bibitem{BrylinskiBrylinski} J.-L.~Brylinski, R.~Brylinski,
\href{https://doi.org/10.1201/9781420035377-5}{\emph{Universal quantum gates}},
in R.~K.~Brylinski and G.~Chen (eds.),
\emph{Mathematics of Quantum Computation}, Chapman \& Hall/CRC, 2002;
\href{https://arxiv.org/abs/quant-ph/0108062}{arXiv:quant-ph/0108062}.

\bibitem{MottonenStatePreparation} M.~M\"ott\"onen, J.~J.~Vartiainen, V.~Bergholm, M.~M.~Salomaa,
\href{https://doi.org/10.26421/QIC5.6-5}{\emph{Transformation of quantum states using uniformly controlled rotations}},
Quantum Inf.\ Comput.\ 5 (2005), 467--473; \href{https://arxiv.org/abs/quant-ph/0407010}{arXiv:quant-ph/0407010}.

\end{thebibliography}
\end{document}